\documentclass[11pt,a4paper]{article}
 
\usepackage{vmargin}
\setmarginsrb{1.0in}{1.0in}{1.0in}{1.0in}{0mm}{0mm}{5mm}{5mm}

\usepackage{bm}
\usepackage{graphicx}
\usepackage{amsmath, amsfonts}
\usepackage[pdftex,dvipsnames]{xcolor}
\usepackage{caption}
\usepackage{comment}
\usepackage{subcaption}
\usepackage[utf8]{inputenc}
\usepackage{multirow}
\usepackage{booktabs}\usepackage{multirow}
\usepackage{booktabs}
\usepackage{afterpage}
\usepackage{pdflscape}
\usepackage{hyperref}

\newtheorem{theorem}{Theorem}[section]

\newtheorem{corollary}[theorem]{Corollary}
\newtheorem{lemma}[theorem]{Lemma}

\newcommand{\qedsymb}{\hfill{\rule{2mm}{2mm}}}
\newcommand{\ra}[1]{\renewcommand{\arraystretch}{#1}}
\newenvironment{proof}{\begin{trivlist}
\item[\hspace{\labelsep}{\bf\noindent Proof: }]
}{\qedsymb\end{trivlist}}

\def\sp {{\bm\sigma}}
\def\u {{\sf U}}

\usepackage{todonotes}

\title{Topological Influence and Locality in  Swap Schelling Games}
\author{Davide Bilò\thanks{University of Sassari, Sassari, Italy, \texttt{davidebilo@uniss.it}} \and Vittorio Bilò\thanks{University of Salento, Lecce, Italy, \texttt{vittorio.bilo@unisalento.it}} \and Pascal Lenzner\thanks{Hasso Plattner Institute, University of Potsdam, Germany, \texttt{firstname.lastname@hpi.de}} \and Louise Molitor\footnotemark[3]}

\date{}

\begin{document}
	\maketitle             
	
	\begin{abstract}
		\noindent  Residential segregation is a wide-spread phenomenon that can be observed in almost every major city. In these urban areas residents with different racial or socioeconomic background tend to form homogeneous clusters. Schelling's famous agent-based model for residential segregation explains how such clusters can form even if all agents are tolerant, i.e., if they agree to live in mixed neighborhoods. For segregation to occur, all it needs is a slight bias towards agents preferring similar neighbors. 
		Very recently, Schelling's model has been investigated from a game-theoretic point of view with selfish agents that strategically select their residential location. In these games, agents can improve on their current location by performing a location swap with another agent who is willing to swap. 

		We significantly deepen these investigations by studying the influence of the underlying topology modeling the residential area on the existence of equilibria, the Price of Anarchy and on the dynamic properties of the resulting strategic multi-agent system. Moreover, as a new conceptual contribution, we also consider the influence of locality, i.e., if the location swaps are restricted to swaps of neighboring agents. We give improved almost tight bounds on the Price of Anarchy for arbitrary underlying graphs and we present (almost) tight bounds for regular graphs, paths and cycles. Moreover, we give almost tight bounds for grids, which are commonly used in empirical studies. For grids we also show that locality has a severe impact on the game dynamics.
	\end{abstract}

	\section{Introduction}
	Today's metropolitan areas are populated by a diverse set of residential groups which differ along ethnical, socioeconomic and other traits. A common finding is that cityscapes are not well-mixed, i.e., the different groups of agents tend to separate themselves into largely homogeneous neighborhoods\footnote{For example, see \url{https://demographics.virginia.edu/DotMap/}.}. This phenomenon is well-known as \emph{residential segregation} and is a subject of study in sociology, mathematics and computer science for at least five decades. The most important scientific model addressing residential segregation was proposed by Schelling~\cite{Sch69,Schelling71} who simply considered two types of residential agents who are located on~a line or on a checkerboard. Each agent is aware of the agents in her neighborhood and is content with her location, if and only if the fraction of neighbors being of her own type is above the tolerance parameter $\tau$, for some $0< \tau \leq 1$. Discontent agents simply move to another location. Using this basic model Schelling showed that starting from an initially mixed state over time segregated neighborhoods will emerge. While this is to be expected for high~$\tau$, Schelling's finding was that this also happens for tolerant agents, i.e., if $\tau \leq \tfrac{1}{2}$. Thus, only a slight bias towards favoring similar neighbors leads to the emergence of segregation. 
	
	Schelling proposed his model as a random process. This has led to an abundance of empirical studies that simulated this process, see, e.g.,~\cite{fossett1998simseg,carver2018} and the references to chapter~$4$ in~\cite{easly}. In these studies, the commonly used underlying topology for modeling the residential area are grid graphs (often toroidal grids where vertices of borders on opposite sides are identified), paths and cycles.
	A recent line of work~\cite{You98,Zha04,Zha04b,Gerhold08,BIK12,BEL14,Bhakta14,BEL16,BIK17,omidvar2018self} rigorously analyzed variants of this random process on paths or grid graphs and it was shown that residential segregation occurs with high probability. However, in reality agents would not move randomly, instead they would move to a location that maximizes their utility.
	
	To address this selfish behavior, a very recent line of work~\cite{CLM18,elkind19,E+19,A+19} initiated the study of residential segregation from a game-theoretic point of view. The residential area is modeled as a multi-agent system consisting of selfish agents who occupy vertices of an underlying graph and try to maximize their utility, which depends on the agents' types in their immediate neighborhood, by strategically selecting locations. Also strategic segregation in social network formation was considered~\cite{aits2019}.  
	
	This paper sets out to significantly improve and deepen the results on game-theoretic residential segregation for the model investigated in~\cite{A+19} which allows pairs of discontent agents of different type to swap their locations to maximize their utility. This variant of Schelling's model becomes more and more realistic as in many cities the percentage of vacant housing is below 1\%. In such settings, location swaps become the only way for agents to improve on their current housing situation. For the model in~\cite{A+19} we consider the influence of the given topology that models the residential area on core game-theoretic questions like the existence of equilibria, the Price of Anarchy and the game dynamics. We thereby focus on popularly studied topologies like grids, paths and cycles.
	Moreover, we follow-up on a proposal by Schelling~\cite{Schelling71} to restrict the movement of agents locally and we investigate the influence of this restriction. Such local swaps are realistic since people want to stay close to their working place or important facilities like schools. This also holds when considering dynamics where agents repeatedly perform local moves since these dynamics can be understood as a process which happens over a long timespan and agents adapt to their new neighborhoods over time.  
	
	\subsection{Model, Definitions and Notation}
	
	We consider a strategic game played on a given underlying connected, unweighted and undirected graph $G=(V,E)$, with $V$ the set of vertices and $E$ the set of edges. We denote the cardinalities of $V$ and $E$ with $n$ and $m$, respectively.
	
	For any vertex $v \in V$ we denote the neighborhood of $v$ in $G$ as $N_v=\{u\in V:\{v,u\}\in E\}$ and $\delta_v=|N_v|$ denotes the degree of~$v$ in $G$. Let $\Delta(G)=\max_{v\in V}\delta_v$ and $\delta(G)=\min_{v\in V}\delta_v$ be the maximum and minimum degree of vertices in $G$, respectively. We call a graph $G$ {\em $\alpha$-almost regular} if $\Delta(G)-\delta(G)=\alpha$ and we call $\alpha$-almost regular graphs \emph{regular} if $\alpha=0$ and \emph{almost regular} when $\alpha=1$. Grid graphs will play a prominent role. We will consider \emph{grid graphs with $4$-neighbors (4-grids)} which are formed by a two-dimensional lattice with $l$ rows and $h$ columns and every vertex is connected to the vertex on its left, top, right and bottom, respectively, if they exist. In \emph{grid graphs with $8$-neighbors (8-grids)}, vertices are additionally also connected to their top-left, top-right, bottem-left and bottom-right vertices, respectively, if they exist.    
	
	For a positive integer $k$, let $[k]$ denote the set $\{1,\ldots,k\}$, moreover, given a graph $G=(V,E)$, let ${\mathcal{T}}_k(G)$ denote the set of $k$-tuples of positive integers summing up to $n=|V|$.
	
	A \emph{Swap Schelling Game with $k$ types ($k$-SSG)} $(G,{\mathbf t})$ is defined by a graph $G = (V,E)$ and a $k$-tuple $\mathbf t = (t_1,\dots,t_k)\in{\mathcal{T}}_k(G)$. There are $n$ strategic agents that need to choose vertices in~$V$ in such a way that every vertex is occupied by exactly one agent. Every agent belongs to exactly one of the $k$ types and there are $t_i$ agents of type $i$, for every $i\in [k]$.
	When $|t_i|=|t_j|$ for each $i,j\in [k]$, we say that the game is {\em balanced}. For convenience and in all of our illustrations, we associate each agent type $i\in [k]$ with a color. When $k=2$, we use colors blue and orange and denote by $b$ and $o=n-b$ the number of blue and orange agents, respectively. Additionally, in case of a game with $k=2$, we will assume that $o \leq b$, i.e., orange is the color of the minority type.
	For any graph $G$ and any $k$-dimensional type vector $\mathbf{t}\in{\mathcal{T}}_k(G)$, let $c:[n] \to [k]$ denote the function which maps any agent $i\in [n]$ to her color $c(i) \in [k]$.
	
	The strategy of an agent is her location on the graph, i.e., a vertex of $G$. A \emph{feasible strategy profile $\sp$} is an $n$-dimensional vector whose $i$-th entry corresponds to the strategy of the $i$-th agent and where all strategies are pairwise disjoint, i.e., $\sp$ is a permutation of $V$, and we will treat $\sp$ as a bijective function mapping agents to vertices, with $\sp^{-1}$ being its inverse function. Thus, any feasible strategy profile $\sp$ corresponds to a coloring of $G$ such that for each $i\in [k]$ exactly~$t_i$ vertices of $G$ are colored with the $i$-th color. We say that agent $i$ \emph{occupies vertex $v$ in~$\sp$} if the $i$-th entry of~$\sp$, denoted as $\sp(i)$, is $v$ and, equivalently, if $\sp^{-1}(v) = i$.
	It will become important to distinguish if two agents $i,j$ occupy neighboring vertices under $\sp$. For this, we will use the notation $1_{ij}(\sp)$ with $1_{ij}(\sp) = 1$ if agents $i$ and $j$ occupy neighboring vertices under~$\sp$ and $1_{ij}(\sp) = 0$ otherwise.   
	
	For an agent $i$ and any feasible strategy profile $\sp$, we denote by $C_i(\sp)=\{v\in V:c(\sp^{-1}(v))=c(i)\}$ the set of vertices of $G$ which are occupied by agents having the same color as agent $i$.
	The utility of agent $i$ in $\sp$ is defined as $\u_i(\sp)=\frac{|N_{\sp(i)}\cap C_i(\sp)|}{\delta_{\sp(i)}},$ i.e., as the ratio of the number of agents with the same type which occupy neighboring vertices and the total number of neighboring vertices, and each agent aims at maximizing her utility. 
	
	Agents can change their strategies only by swapping vertex occupation with another agent. Consider two strategic agents $i$ and $j$ which occupy vertices $\sp(i)$ and $\sp(j)$, respectively. After performing a \emph{swap} both agents exchange their occupied vertex which yields a new feasible strategy profile $\sp_{ij}$, which is identical to $\sp$ except that the $i$-th and the $j$-th entries are exchanged. Thus, in the induced coloring of $G$, the coloring corresponding to $\sp_{ij}$ is identical to the coloring corresponding to $\sp$ except that the colors of vertices $\sp(i)$ and $\sp(j)$ are exchanged. We say that a swap is \emph{local} if the swapping agents occupy neighboring vertices, i.e., if $1_{ij}(\sp)=1$. 
	
	As agents are strategic and want to maximize their utility, we will only consider {\em profitable swaps} by agents, i.e., swaps which strictly increase the utility of both agents involved in the swap. It follows that profitable swaps can only occur between agents of different colors.
	We call a feasible strategy profile $\sp$ a \emph{swap equilibrium}, or simply, equilibrium, if $\sp$ does not admit profitable swaps, that is, if for each pair of agents $i,j$, we have $\u_i(\sp)\geq\u_i(\sp_{ij})$ or $\u_j(\sp)\geq\u_j(\sp_{ij})$. We call $\sp$ a \emph{local swap equilibrium}, or simply local equilibrium, if no profitable local swap exists under $\sp$.
	If agents are restricted to performing only local swaps, then we call the corresponding strategic game \emph{Local Swap Schelling Game with $k$ types (local $k$-SSG)}. Clearly, any swap equilibrium $\sp$ is also a local swap equilibrium but the converse is not true. Thus the set of local swap equilibria is a superset of the set of swap equilibria.  
	
	We measure the quality of a feasible strategy profile $\sp$ by its \emph{social welfare} $\u(\sp)$, which is the sum over the utilities of all agents, i.e., $\u(\sp) = \sum_{i=1}^{n}\u_i(\sp)$. For any game $(G,{\mathbf t})$, let~$\sp^*(G,\mathbf{t})$ denote a feasible strategy profile which maximizes the social welfare and let~$SE(G,\mathbf{t})$ and $LSE(G,\mathbf{t})$ denote the set of swap equilibria and local swap equilibria for $(G,{\mathbf t})$, respectively.  
	We will study the impact of the agents' selfishness on the obtained social welfare for games played on a given class of underlying graphs $\mathcal{G}$ with $k$ agent types by analyzing the \emph{Price of Anarchy (PoA)}~\cite{KP99}, which is defined as $ PoA(\mathcal{G},k) = \max_{G \in \mathcal{G}} \max_{\mathbf{t} \in \mathcal{T}_{k}(G)}\frac{\u(\sp^*(G,\mathbf{t}))}{\min_{\sp \in SE(G,\mathbf{t})}\u(\sp)}.$
	Analogously, we define the \emph{Local Price of Anarchy (LPoA)} as the same ratio but with respect to local swap equilibria. It follows that, for any $k\geq 2$ and class of graphs $\mathcal{G}$, we have $PoA(\mathcal{G},k) \leq LPoA(\mathcal{G},k)$.
	
	We will also investigate the dynamic properties of the (local) $k$-SSG, i.e., we analyze if the game has the \emph{finite improvement property (FIP)} \cite{MS96}. In our model, a game possesses the FIP if every sequence of profitable (local) swaps is finite. This is equivalent to the existence of an ordinal potential function which guarantees that sequences of profitable (local) swaps will converge to a (local) swap equilibrium of the game. The FIP can be disproved by showing the existence of an \emph{improving response cycle (IRC)}, which is a sequence of feasible strategy profiles $\sp^0, \sp^1,\ldots, \sp^\ell$, with $\sp^\ell = \sp^0$, where $\sp^{q+1}$ is obtained by a profitable swap by two agents in $\sp^q$, for $q\in [\ell-1]$. For investigating the FIP, the following function $\Phi$ mapping feasible strategy profiles to natural numbers will be important: $\Phi(\sp) = \left|\left\{\{u,v\}\in E \mid c(\sp^{-1}(u)) = c(\sp^{-1}(v))\right\}\right|.$ Hence, $\Phi(\sp)$ is the number of edges of $G$ whose endpoints are occupied by agents of the same color under the feasible strategy profile $\sp$. We will denote such edges as \emph{monochromatic edges} and $\Phi(\sp)$ as the \emph{potential of $\sp$}. We will see that potential-preserving profitable swaps exist. For analyzing such swaps, we will consider the \emph{extendend potential} $\Psi(\sp)$ which essentially is $\Phi(\sp)$ augmented with a tie-breaker. It is defined as $\Psi(\sp) = (\Phi(\sp),n-z(\sp)),$ where $z(\sp)$ is the number of agents having utility $0$ under $\sp$. We compare $\Psi$ for different strategy profiles $\sp$ and $\sp'$ lexicographically, i.e., on the one hand we have $\Psi(\sp) > \Psi(\sp')$ if $\Phi(\sp) > \Phi(\sp')$ or $\Phi(\sp) = \Phi(\sp')$ and $z(\sp) < z(\sp')$. On the other hand we have $\Psi(\sp) < \Psi(\sp')$ if $\Phi(\sp) < \Phi(\sp')$ or $\Phi(\sp) = \Phi(\sp')$ and $z(\sp) >z(\sp')$. Note that any profitable swap which increases (decreases) the potential $\Phi$ also increases (decreases) the extended potential $\Psi$.
	
	\subsection{Related Work} 
	We focus on related work on game-theoretic segregation models.
	
	Zhang~\cite{Zha04,Zha04b} was the first who introduced a game-theoretic model related to Schelling's original model. There, agents having a noisy single peaked utility function and preferring to be in a balanced neighborhood were employed. Later, Chauhan et al.~\cite{CLM18} introduced a game-theoretic model which is much closer to Schelling's formulation. In their model there are two types of agents and the utility of an agent depends on the type ratio in her neighborhood. An agent is content if the fraction of own-type neighbors is above $\tau \in (0,1]$. Additionally, agents may have a preferred location. To improve their utility, agents can either swap with another agent who is willing to swap (Swap Schelling Game) or jump to an unoccupied vertex (Jump Schelling Game). Their main contribution is an investigation of the convergence properties of many variants of the model. Moreover they provide basic properties of stable placements and their efficiency. Echzell et al.~\cite{E+19} strengthen these results but omitted location preferences. Instead they extended the model to more than two agent types and studied the computational hardness of finding optimal placements.
	Elkind et al.~\cite{elkind19} investigated a similar model with~$k$ types where agents are either strategic or stubborn. Only strategic agents are willing to move and strive for maximizing the fraction of own-type neighbors by jumping to a suitable unoccupied location. This corresponds to the jump version of Chauhan et al.~\cite{CLM18} with $\tau=1$. They show that equilibria are not guaranteed to exist, they analyze the complexity of finding optimal placements and they prove that the PoA can be unbounded. Very recently, Agarwal et al.~\cite{A+19} considered swap games in the model of Elkind et al.~\cite{elkind19}. They show that on underlying trees equilibria may not exist and that deciding equilibrium existence and the existence of a state with at least a given social welfare is NP-hard. They also establish that the PoA is in $\Theta(n)$ on underlying star graphs if there are at least two agents of each type and between 2.0558 and 4 for balanced games on any graph. Moreover, for $k\geq 3$ the PoA can be unbounded even in balanced games. Additionally, they give a constant lower bound on the Price of Stability and show that it equals 1 on regular graphs. Finally, they introduce a new benchmark for measuring diversity by counting the number of agents having at least one neighbor of different type. In the present paper, we focus on this very recent model by Agarwal et al.~\cite{A+19} and extend and improve their PoA results.
	
	Hedonic games~\cite{DG80,BJ02} are related to Schelling games. In particular, Schelling games share a number of properties with fractional hedonic games~\cite{bilo2018,monaco2018,aziz2019,carosi2019,monaco2019}, hedonic diversity games~\cite{bredereck2019} and FEN-hedonic games~\cite{igarashi2019,fichtenberger2019,kerkmann2019}. However, one of the main differences is that in Schelling games the neighborhoods of coalitions overlap while in hedonic games agents form disjoint coalitions with identical neighborhoods for all agents within the same coalition.
	
	Investigating a local variant of Schelling's model, although proposed by Schelling~\cite{Schelling71} himself, seems to be a novel approach. To the best of our knowledge, local moves have only been addressed briefly by Vinkovi\'c and Kirnan~\cite{Vin06} in a model which can be understood as a continuous physical analogue of Schelling's model. 
	
	\subsection{Our Contribution}
	
	We follow the model of Agarwal et al.~\cite{A+19},  that is, we consider Swap Schelling Games and investigate, on the one hand, the existence of equilibria and the game dynamics and, on the other hand, the quality of the equilibria in terms of the PoA. The novel feature of our analysis is our focus on the influence of the underlying graph and that we also investigate the impact of restricting the agents to performing only local swaps.
	See Table~\ref{tbl:previous_results} for a result overview.
	
	\afterpage{
		\begin{landscape}
			\begin{table*}[h]
				\centering	
				\scriptsize
				\ra{1.3}
				\begin{tabular}{@{}lcc c crlrlcc }
					\toprule
					& \multicolumn{2}{c}{Equilibrium Existence} & \multicolumn{2}{c}{Finite Improvement Property} & & \multicolumn{2}{c}{Price of Anarchy} & & \\
					\cmidrule(l{2em}r{1em}){2-3} \cmidrule(l{0.5em}r{0.5em}){4-5} \cmidrule(l{0.7em}r{0.5em}){6-10} 
					\hspace*{-2em} & $k$-SSG & local $k$-SSG & $k$-SSG & local $k$-SSG & \multicolumn{2}{c}{$2$-SSG} & \multicolumn{2}{c}{local $2$-SSG} \\
					graph classes & & & & & & $o = 2\alpha + \beta$ & & $n = 3\alpha + \beta$ \\
					\midrule
					\textbf{arbitrary} & $\times$ (\cite{A+19}) & & $\times$ (\cite{A+19}) & & $\infty$ (\cite{A+19}) & $o = 1$ & $\left(2n+\frac{8}{n} - 8, 2n - \frac{8}{n} \right)$ &  $o = \frac{n}{2}$ \\
					& & & & & & & (Thm.~\ref{thm:PoA_max_min_degree}) & \\
				    & & & & & $\leq 3$ (Thm.~\ref{thm:poa_2ssg}) & $o = \frac{n}{2}$ & $\leq 2\left(1+\frac{\Delta-1}{\delta-1}\right)$ (Thm.~\ref{PoA_local}) & $\delta \geq 2$ \\
				    & & & & & $\leq \frac{no(n-o)-n}{o(o-1)(n-o)}$ (Thm.~\ref{thm:poa_2ssg}) & otherwise & $\left( \frac{\Delta(\Delta-1)}{2} - \epsilon, 4(\Delta^2 - \Delta+1) \right)$ & $ \Delta \leq n-2$ \\
				    & & & & & & & (Thm.~\ref{thm:LPoA_connected_and_bounded_max_degree}) & \\
					\textbf{regular} & \checkmark (\cite{E+19}) & \checkmark (\cite{E+19}) & \checkmark (\cite{E+19}) & \checkmark (\cite{E+19}) & $2+\frac{1}{\alpha}$ (Thm.~\ref{PoA_reg}, \ref{PoA_reg_1}) & $\Delta \in (2\alpha, 2\alpha +1 )$ & $2+\frac{1}{\alpha}$ (Cor.~\ref{PoA_reg}, Thm.~\ref{PoA_reg_1}) & $\Delta \in (2\alpha, 2\alpha +1 )$ \\
					\textbf{$1$-regular} & \checkmark (Thm.~\ref{almost-regular}) & \checkmark (Thm.~\ref{almost-regular}) & \checkmark (Thm.~\ref{almost-regular}) & \checkmark (Thm.~\ref{almost-regular}) \\ 
					\textbf{trees} & $\times$ (\cite{A+19}) & \checkmark (Thm.~\ref{thm:equ_trees}) & $\times$ (\cite{A+19}) & & $\left( \frac{\Delta(\Delta-1)}{2} - \epsilon, 4(\Delta^2 - \Delta+1) \right)$ & $ \Delta \leq n-2$ & $\left( \frac{\Delta(\Delta-1)}{2} - \epsilon, 4(\Delta^2 - \Delta+1) \right)$ & $ \Delta \leq n-2$ \\
		 			& & & & & (Cor.~\ref{cor:poa_trees} + Thm.~\ref{thm:LPoA_connected_and_bounded_max_degree}) & & (Cor.~\ref{cor:poa_trees} + Thm.~\ref{thm:LPoA_connected_and_bounded_max_degree}) & \\
					\textbf{cycles} & \checkmark (\cite{E+19}) & \checkmark (\cite{E+19}) & \checkmark (\cite{E+19}) & \checkmark (\cite{E+19}) & $1$ (Thm.~\ref{thm:cycle_global}) & $o = 1$ & $1$ (Thm.~\ref{thm:cycle_local}) & $o = 1$ \\
		 			& & & & & $\frac{n-2}{b+\beta}$ (Thm.~\ref{thm:cycle_global}) & otherwise & $\frac{n-2}{b-o}$ (Thm.~\ref{thm:cycle_local}) & $o\geq 2$, $b\geq2o$ \\
		           	& & & & & & & $\frac{n-2}{\alpha + \beta}$ (Thm.~\ref{thm:cycle_local}) & otherwise \\
					\textbf{paths} & \checkmark (Thm.~\ref{almost-regular}) & \checkmark (Thm.~\ref{almost-regular}) & \checkmark (Thm.~\ref{almost-regular}) & \checkmark (Thm.~\ref{almost-regular}) & $\infty$ (Thm.~\ref{thm:path_global}) & $n = 3$ & $\infty$ (Thm.~\ref{thm:path_local}) & $n = 3$\\
		 			& & & & & $\frac{2n-2}{2n-5}$ (Thm.~\ref{thm:path_global}) & $n > 3$, $o = 1$ & $\frac{2n-2}{2n-5}$ (Thm.~\ref{thm:path_local}) & $n > 3$, $o = 1$ \\ 
					& & & & & $\frac{n-1}{b+1+\beta}$ (Thm.~\ref{thm:path_global}) & $n > 3$, $o \geq 2$, & $\frac{n-1}{b-o-1}$ (Thm.~\ref{thm:path_local}) & $n > 3$, $o \geq 2$, $b \geq 2o$ \\ 
					& & & & & & $\beta \leq 2\alpha +1$ & \\ 
					& & & & & $\frac{n-1}{b+\beta}$ (Thm.~\ref{thm:path_global}) & otherwise & $\frac{n-1}{\alpha}$ (Thm.~\ref{thm:path_local}) & otherwise \\ 
					\textbf{$4$-grids} & \checkmark (Thm.~\ref{almost-almost-regular}) & \checkmark (Thm.~\ref{almost-almost-regular}) & \checkmark (Thm.~\ref{almost-almost-regular}) & \checkmark (Thm.~\ref{almost-almost-regular}) & $\frac{25}{22}$ (Thm.~\ref{if1}) & $o = 1$ & $(3 - \epsilon,3)$ (Thm.~\ref{thm:4grid_local}) & $2\times h$ grid, $h \geq 3$ \\
					& & & &	& $2$ (Thm.~\ref{thm:4grid_poa},~\ref{lbgrid}) & otherwise & $\left(\frac{36}{13} - \epsilon, \frac{36}{13} \right)$ (Thm.~\ref{thm:poa:grid}) & $3\times h$ grid, $h \geq 3$ \\	
					& & & &	& & & $\left(\frac52 -\epsilon, \frac52 +\epsilon \right) $ (Thm.~\ref{thm:poa_8grid_}) & $l \times h$ grid, $h,l \geq 8 + \frac{20}{\epsilon}$ \\	
					\textbf{$8$-grids} & \checkmark (Thm.~\ref{thm:existence}) & \checkmark (Thm.~\ref{thm:existence}) & $\times$ (Thm.~\ref{IRC_grid}) & \checkmark (Thm.~\ref{FIP_grid}) & $\frac{897}{704}$ (Thm.~\ref{thm:8grid_poa}) & $o =1$ & $\leq \frac{9}{4} + \epsilon$ (Thm.~\ref{poa_8grid}) & $l \times h$ grid, $h,l \geq 8 + \frac{18}{\epsilon}$ \\
					& $k = 2$ & $k = 2$ & & & $\leq 8$ (Thm.~\ref{thm:poa_8grid}) & otherwise \\	
					\bottomrule
				\end{tabular}
			\vspace*{+0.1cm}
			\caption{Result overview. The ``\checkmark'' symbol denotes that the respective property holds. Note that a ``\checkmark'' in the ``$k$-SSG'' column implies a ``\checkmark'' in the local $k$-SSG column.The ``$\times$'' symbol denotes that equilibrium existence is not guaranteed and that an IRC exists, respectively. For $k = 2$ we denote by $b$ and $o$ the number of blue and orange agents, respectively and we assume $o \leq b$. If we use $\alpha$ or $\beta$ in the respective bound, their meaning is defined in the top of the  respective column. $\epsilon$ is a constant larger zero.}\label{tbl:previous_results}
		\end{table*}
	\end{landscape}}

	While in~\cite{A+19} it was proven that equilibria may fail to exist for arbitrary underlying graphs and in~\cite{E+19} equilibrium existence was shown for regular graphs, we extend and refine these results by investigating almost regular graphs as well as paths, $4$-grids and $8$-grids. We establish equilibrium existence for all these graph classes and all our results yield polynomial time algorithms for computing an equilibrium. 
	Moreover, we study the PoA in-depth. Since it was shown in~\cite{A+19} that the PoA can be unbounded for $k \geq 3$, we focus on the PoA of the (local) $2$-SSG. 
	We give tight or almost tight bounds on the PoA for all mentioned graph classes which in many cases are significant improvements on the $\Theta(n)$ bound proven in~\cite{A+19}. In particular, we also improve the upper bound for balanced games on arbitrary graphs and~we give PoA bounds which depend on the minimum and maximum degree in the underlying~graph.

	Besides analyzing equilibria in the general model of Agarwal et al.~\cite{A+19}, we introduce and analyze a local variant of the model, which was already suggested by Schelling~\cite{Schelling71} but to the best of our knowledge has not yet been explored for Schelling's model. Our results indicate that the local variant has favorable properties. For instance, equilibria are guaranteed to exists on trees in the local version while in~\cite{A+19} it was shown that this is not the case for the general model. Moreover, for many cases we can show that the PoA in the local version deteriorates only slightly compared to the global version.
	
	\section{Equilibrium Existence and Dynamics}

	We start by providing a precise characterization which ties equilibria in $2$-SSGs with the sum of the utilities experienced by any two agents of different colors.

	\begin{lemma}\label{sum_in_eq}
		A strategy profile $\sp$ for a $2$-SSG is an equilibrium if and only if, for any two agents~$i$ and $j$ with $c(i)\neq c(j)$ and $\delta_{\sp(i)}\leq\delta_{\sp(j)}$, it holds that $\u_i(\sp)+\u_j(\sp)\geq 1-\frac{1_{ij}(\sp)}{\delta_{\sp(i)}}$.
	\end{lemma}

	\begin{proof}
		Fix an equilibrium $\sp$ and consider two agents $i$ and $j$ such that $c(i)\neq c(j)$ and $\delta_{\sp(i)}\leq\delta_{\sp(j)}$. Assume without loss of generality that $i$ is orange and $j$ is blue. Let $o_i$ be the number of orange neighbors of $\sp(i)$ and $b_j$ be the number of blue neighbors of $\sp(j)$. It holds that
		\begin{eqnarray*}
			\u_i(\sp)=\frac{o_i}{\delta_{\sp(i)}},\ \u_j(\sp)=\frac{b_j}{\delta_{\sp(j)}}
 		\end{eqnarray*}
		and
		\begin{eqnarray*}
			\u_i(\sp_{ij})=\frac{\delta_{\sp(j)}-b_j-1_{ij}(\sp)}{\delta_{\sp(j)}},\ \u_j(\sp_{ij})=\frac{\delta_{\sp(i)}-o_i-1_{ij}(\sp)}{\delta_{\sp(i)}}.
		\end{eqnarray*}

		\noindent As $\sp$ is an equilibrium, it must be either $\u_i(\sp)\geq\u_i(\sp_{ij})$ or $\u_j(\sp)\geq\u_j(\sp_{ij})$. 

		In the first case, we get \[\u_i(\sp)+\u_j(\sp)\geq 1-\frac{1_{ij}(\sp)}{\delta_{\sp(j)}},\] in the second one, we get \[\u_i(\sp)+\u_j(\sp)\geq 1-\frac{1_{ij}(\sp)}{\delta_{\sp(i)}}.\] Thus, given that $\delta_{\sp(i)}\leq\delta_{\sp(j)}$, in any case we have that $\u_i(\sp)+\u_j(\sp)\geq 1-\frac{1_{ij}(\sp)}{\delta_{\sp(i)}}$.

		Now fix a strategy profile $\sp$ such that, for any two agents $i$ and $j$ with $c(i)\neq c(j)$ and $\delta_{\sp(i)}\leq\delta_{\sp(j)}$, it holds that $\u_i(\sp)+\u_j(\sp)\geq 1-\frac{1_{ij}(\sp)}{\delta_{\sp(i)}}$. 

		Assume, by way of contradiction, that $\sp$ is not an equilibrium. Then, there exist an orange agent $i$ and a blue agent $j$ such that $\u_i(\sp)<\u_i(\sp_{ij})$ and $\u_j(\sp)<\u_j(\sp_{ij})$. Let $o_i$ be the number of orange neighbors of $\sp(i)$ and $b_j$ be the number of blue neighbors of $\sp(j)$. It holds that
  		\begin{eqnarray*}
			\u_i(\sp)=\frac{o_i}{\delta_{\sp(i)}},\ \u_j(\sp)=\frac{b_j}{\delta_{\sp(j)}}
 		\end{eqnarray*}
		and
		\begin{eqnarray*}
			\u_i(\sp_{ij})=\frac{\delta_{\sp(j)}-b_j-1_{ij}(\sp)}{\delta_{\sp(j)}},\ \u_j(\sp_{ij})=\frac{\delta_{\sp(i)}-o_i-1_{ij}(\sp)}{\delta_{\sp(i)}}.
		\end{eqnarray*}
		By $\u_i(\sp)<\u_i(\sp_{ij})$, we obtain $$\u_i(\sp)+\u_j(\sp)\geq 1-\frac{1_{ij}(\sp)}{\delta_{\sp(j)}}.$$ 

		\noindent Similarly, by $\u_j(\sp)<\u_j(\sp_{ij})$, we obtain $$\u_i(\sp)+\u_j(\sp)\geq 1-\frac{1_{ij}(\sp)}{\delta_{\sp(i)}}.$$

		\noindent At least one of the two derived inequalities contradicts the assumption on $\sp$. Thus, $\sp$ is an equilibrium.
	\end{proof}

	\noindent By exploiting the potential $\Phi$, Echzell et al. \cite{E+19} show that, for any $k\geq 2$, $k$-SSGs played on regular graphs have the FIP and that any sequence of profitable swaps has length at most~$m$. This result can be extended to $\alpha$-almost regular graphs for some values of~$\alpha$. First, we need the following technical lemma.

	\begin{lemma}\label{findamental}
		Fix a $k$-SSG $(G,{\mathbf t})$, with $k\geq 2$, a strategy profile $\sp$ and a profitable swap in~$\sp$ performed by vertices $i$ and $j$ such that $\delta_{\sigma(i)}\leq\delta_{\sigma(j)}$. If $\delta_{\sigma(j)}-\delta_{\sigma(i)}\leq 1$, then the swap is $\Phi$-increasing. If $\delta_{\sigma(j)}-\delta_{\sigma(i)}\leq 2$, then the swap is either $\Phi$-increasing or $\Phi$-preserving, with the swap being $\Phi$-preserving only if $\u_j(\sp)\in\left(\frac12,1\right)$.
	\end{lemma}

	\begin{proof}
		Assume, without loss of generality, that $c(i)$ is orange and $c(j)$ is blue; moreover, define~$\sigma(i)=u$ and $\sigma(j)=v$. Let $o_u$ be the number of orange agents occupying vertices adjacent to~$u$ in $\sp$, $x_u$ be the number of neither orange not blue agents occupying vertices adjacent to $u$ in~$\sp$, $b_v$ be the number of blue agents occupying vertices adjacent to $v$ in $\sp$ and~$x_v$ be the number of neither orange not blue agents occupying vertices adjacent to $v$ in $\sp$. We have 
		$$\u_i(\sp)=\frac{o_u}{\delta_u},\ \ \u_j(\sp)=\frac{b_v}{\delta_v}$$ and $$\u_i(\sp_{ij})=\frac{\delta_v-b_v-x_v-1_{ij}(\sp)}{\delta_v},\ \ \u_j(\sp_{ij})=\frac{\delta_u-o_u-x_u-1_{ij}(\sp)}{\delta_u}.$$
		As $i$ and $j$ perform a profitable swap in $\sp$, we have $\u_i(\sp)<\u_i(\sp_{ij})$ and $\u_j(\sp)<\u_j(\sp_{ij})$ which imply
		\begin{equation}\label{eq1}
			\delta_u b_v+\delta_v o_u+\delta_ux_v+\delta_u 1_{ij}(\sp)<\delta_u \delta_v
		\end{equation}
		and
		\begin{equation}\label{eq2}
			\delta_u b_v+\delta_v o_u+\delta_vx_u+\delta_v 1_{ij}(\sp)<\delta_u \delta_v.
		\end{equation}
		Moreover, we have
		\begin{eqnarray*}
			\Phi(\sp_{ij})-\Phi(\sp) & = & \delta_u-1_{ij}(\sp)-o_u-x_u+\delta_v-1_{ij}(\sp)-b_v-x_v-o_u-b_v\\
			& = & \delta_u+\delta_v-x_u-x_v-2(o_u+b_v+1_{ij}(\sp)).
		\end{eqnarray*}

		\begin{itemize}

			\item If $\delta_u=\delta_v:=\delta'$, (\ref{eq1}) implies $o_u+b_v+1_{ij}(\sp)+x_v<\delta',$ while (\ref{eq2}) implies $o_u+b_v+1_{ij}(\sp)+x_u<\delta'$ which together yield $$\Phi(\sp_{ij})-\Phi(\sp)=2\delta'-x_u-x_v-2(o_u+b_v+1_{ij}(\sp))>0.$$ 

			\item If $\delta_u=\delta_v-1$, (\ref{eq1}) implies $o_u+b_v+1_{ij}(\sp)+x_v<\delta_v-1+\frac{b_v+x_v+1_{ij}(\sp)}{\delta_v},$ while (\ref{eq2}) implies $o_u+b_v+1_{ij}(\sp)+x_u<\delta_v-1+\frac{b_v}{\delta_v}.$ As $b_v+x_v+1_{ij}(\sp)\leq\delta_v$ by definition, we get $o_u+b_v+1_{ij}(\sp)+x_v\leq\delta_v-1$ and $o_u+b_v+1_{ij}(\sp)+x_u\leq\delta_v-1$ which together yield $$\Phi(\sp_{ij})-\Phi(\sp)=2\delta_v-1-x_u-x_v-2(o_u+b_v+1_{ij}(\sp))>0.$$

			\item If $\delta_u=\delta_v-2$, (\ref{eq1}) implies $o_u+b_v+1_{ij}(\sp)+x_v<\delta_v-2+\frac{2(b_v+x_v+1_{ij}(\sp))}{\delta_v},$ while (\ref{eq2}) implies $o_u+b_v+1_{ij}(\sp)+x_u<\delta_v-2+\frac{2 b_v}{\delta_v}.$ As $b_v+x_v+1_{ij}(\sp)\leq\delta_v$ by definition, we get $o_u+b_v+1_{ij}(\sp)+x_v\leq\delta_v-1$ and $o_u+b_v+1_{ij}(\sp)+x_u\leq\delta_v-1$ which together yield $$\Phi(\sp_{ij})-\Phi(\sp)=2\delta_v-2-x_u-x_v-2(o_u+b_v+1_{ij}(\sp))\geq 0.$$ However, note that equality occurs only in the case in which $\frac{2 b_v}{\delta_v}>1$ which requires $b_v>\frac{\delta_v}{2}$, that is, $\u_j(\sp)>\frac12$. Clearly, as $j$ improves after the swap, it must also be $\u_j(\sp)<1$.
		\end{itemize}
	\end{proof}
	Given the above lemma, existence and efficient computation of equilibria for $k$-SSGs played on almost regular graphs can be easily obtained for any $k\geq 2$. 

	\begin{theorem}\label{almost-regular}
		For any $k\geq 2$, $k$-SSGs played on almost regular graphs has the FIP. Moreover, at most $m$ profitable swaps are sufficient to reach an equilibrium starting from any initial strategy profile.
	\end{theorem}

	\begin{proof}
		The claim comes from Lemma \ref{findamental}, as in any almost regular graph $G$ it holds that $\Delta-\delta=1$.
	\end{proof}

	\noindent Theorem \ref{almost-regular} cannot be extended beyond almost regular graphs as Agarwal et al.~\cite{A+19} provide a $2$-SSG played on a $2$-almost regular graph (more precisely, a tree) admitting no equilibria.
	However, in the next theorem, we show that positive results can be still achieved in games played on $2$-almost regular graphs obeying some additional properties.

	\begin{theorem}\label{almost-almost-regular}
		Let $G$ be a $2$-almost regular graph such that $\Delta(G)\leq 4$ and every vertex of degree~$\delta$ is adjacent to at most $\delta-1$ vertices of degree $\Delta(G)$. Then, for any $k\geq 2$, every $k$-SSG played on $G$ possesses the FIP. Moreover, at most $O(nm)$ profitable swaps are sufficient to reach an equilibrium starting from any initial strategy profile.
	\end{theorem}

	\begin{proof}
		By Lemma \ref{findamental}, we know that any profitable swap occurring in a strategy profile $\sp$ is $\Phi$-increasing unless it involves an agent $i$ occupying vertex $\sigma(i)=u$, with $\delta_u=\delta$, and an agent $j$ occupying vertex $\sigma(j)=v$, with $\delta_v=\Delta$, and such that $\u_j(\sp)\in(\frac12,1)$. As $G$ is connected, we have $\delta\geq 1$, which yields $\Delta\in\{3,4\}$. This fact, together with $\u_j(\sp)\in(\frac12,1)$ implies $\u_j(\sp)\in\{\frac23,\frac34\}$. As $\u_j(\sp_{ij})>\u_j(\sp)$, we get $\u_j(\sp_{ij})=1$ which implies that all vertices adjacent to $u$ are occupied by agents of the same color of agent $j$, which implies $\u_i(\sp)=0$. So we can conclude that, in order to have a $\Phi$-preserving profitable swap, we need a profitable swap involving a vertex $u$ of degree $\delta$ such that $\u_{\sigma^{-1}(u)}(\sp)=0$ and $\u_{\sigma_{ij}^{-1}(u)}(\sp)=1$. Thus, in order for an agent occupying $u$ to perform once again a $\Phi$-preserving profitable swap, all vertices in $ N_u$ need to change their colors, i.e., all agents occupying vertices adjacent to~$u$ must perform a profitable swap. By Lemma \ref{findamental}, any agent occupying a vertex $v\in N_u$ can be involved in a $\Phi$-preserving swap only if $\delta_v=\Delta$. By assumption $u$ has at least a neighbor of degree different than $\Delta$. Thus, between any two consecutive $\Phi$-preserving profitable swaps involving an agent residing at a fixed vertex, a $\Phi$-increasing profitable swap has to occur. This immediately implies that no more than $n$ consecutive $\Phi$-preserving profitable swaps are possible.
	\end{proof}

	\noindent As 4-grids meet the conditions required by Theorem \ref{almost-almost-regular}, we get the following corollary.

	\begin{corollary}\label{cor-grid}
		For any $k\geq 2$, every $k$-SSG played on a 4-grid possesses the FIP. Moreover, at most $O(nm)$ profitable swaps are sufficient to reach an equilibrium starting from any initial strategy profile.
	\end{corollary}

	\noindent As mentioned before, Agarwal et al.~\cite{A+19} pointed out that $2$-SSGs played on trees are not guaranteed to admit equilibria. We show that this is no longer the case in local $k$-SSGs for any value of $k\geq 2$.

	\begin{theorem}\label{thm:equ_trees}
		For any $k \geq 2$, every local $k$-SSG played on a tree has an equilibrium which can be computed in polynomial time.
	\end{theorem}

	\begin{proof}
		Root the tree $T$ at a vertex $r$. We will place the agents color by color, starting with color $1$ and ending with color $k$. Before we place an agent at an inner vertex $v$ all of $v$'s descendants in $T$ have to be occupied. Hence, we place the agents starting from the leaves, and the root $r^\prime$ of every subtree $T^\prime$ is the last vertex in $T^\prime$ which will be occupied. Thus, we ensure that, if the root $r^\prime$ of a subtree $T^\prime$ is occupied by an agent of color $i \in [k]$, $T^\prime$ contains only agents of color $i^\prime \leq i$. Clearly, this construction yields a feasible strategy profile, that we denote by $\sp$, and can be implemented in polynomial time.

	Consider two agents $i$ and $j$ of different colors that occupy two adjacent vertices $u$ and~$v$, respectively. Without loss of generality, we assume that $u$ is the parent of $v$ in $T$. Since $c(j) < c(i)$, the subtree of $T$ rooted at $v$ contains no vertex of color $c(i)$. As a consequence $\u_i(\sp_{ij})=0$. Hence $\sigma$ is a LSE.
	\end{proof}

	\noindent Note that, as we move from $4$-grids to $8$-grids, Corollary \ref{cor-grid} does not hold any more. In fact, for the latter class of graphs, we show that the FIP is guaranteed to hold only for local games.

	\begin{lemma}\label{potential_decreasing}
		Fix a local $2$-SSG played on an $8$-grid, a strategy profile $\sp$ and a profitable swap in~$\sp$ performed by agents $i$ and $j$. It holds that
		\begin{itemize}
			\item[i)] If $\delta_{\sigma(i)} = 3$ and $\delta_{\sigma(j)} = 8$, then the swap is $\Phi$-decreasing by $1$ if $\u_i(\sp) = 0$ and $\u_j(\sp) = \frac58$ otherwise it is a $\Phi$-increasing swap. 
			\item[ii)] If $\delta_{\sigma(i)} = 5$ and $\delta_{\sigma(j)} = 8$, then the swap is $\Phi$-decreasing by $1$ if $\u_i(\sp) = 0$ and $\u_j(\sp) = \frac68$ otherwise it is a $\Phi$-increasing swap. 
		\end{itemize}
	\end{lemma}

	\begin{proof}
		Assume, without loss of generality, that $c(i)$ is orange and $c(j)$ is blue; moreover, define $\sigma(i)=u$ and $\sigma(j)=v$. Let $o_u$ be the number of orange agents occupying vertices adjacent to~$u$ in $\sp$ and $b_v$ be the number of blue agents occupying vertices adjacent to $v$ in $\sp$.
		\begin{itemize}
			\item[{\em i)}] We have $$\u_i(\sp)=\frac{o_u}{3},\ \ \u_j(\sp)=\frac{b_v}{8}$$ and $$\u_i(\sp_{ij})=\frac{7-b_v}{8},\ \ \u_j(\sp_{ij})=\frac{2-o_u}{3}.$$ As $i$ and $j$ perform a profitable swap in $\sp$, we have $\u_i(\sp)<\u_i(\sp_{ij})$ and $\u_j(\sp)<\u_j(\sp_{ij})$ which imply
				\begin{equation}\label{blue}
					b_v<\frac{16}{3}-\frac{8}{3}o_u.
				\end{equation}
				Moreover, we have
				\begin{eqnarray*}
					\Phi(\sp_{ij})-\Phi(\sp) & = & 3-1-o_u+8-1-b_v-o_u-b_v =  9 -2o_u - 2b_v.
				\end{eqnarray*}
				Since $o_u$ is in the set $\{0,1,2\}$, we have the following cases:

				If $o_u = 0$, (\ref{blue}) implies $b_v<\frac{16}{3}$ which yields $\Phi(\sp_{ij})-\Phi(\sp) > \frac{-5}{3} $.

				If $o_u = 1$, (\ref{blue}) implies $b_v<\frac{8}{3}$ which yields $\Phi(\sp_{ij})-\Phi(\sp) > \frac{5}{3} $.

				If $o_u = 2$, $i$ and $j$ cannot perform a profitable local swap since $\u_j(\sp_{ij}) = 0$.

				Since $\Phi(\sp)$ is integral, the statement follows.

			\item[{\em ii)}] We have $$\u_i(\sp)=\frac{o_u}{5},\ \ \u_j(\sp)=\frac{b_v}{8}$$ and $$\u_i(\sp_{ij})=\frac{7-b_v}{8},\ \ \u_j(\sp_{ij})=\frac{4-o_u}{5}.$$  As $i$ and $j$ perform a profitable swap in $\sp$, we have $\u_i(\sp)<\u_i(\sp_{ij})$ and $\u_j(\sp)<\u_j(\sp_{ij})$ which imply
				\begin{equation}\label{red}
					b_v<\frac{32}{5}-\frac{8}{5}o_u.
				\end{equation}

				Moreover, we have
				\begin{eqnarray*}
					\Phi(\sp_{ij})-\Phi(\sp) & = & 5-1-o_u+8-1-b_v-o_u-b_v =  11 -2o_u - 2b_v.
				\end{eqnarray*}

				Since $o_u$ is in the set $\{0,1,2,3,4\}$, we have the following cases:
	
				If $o_u = 0$, (\ref{red}) implies $b_v<\frac{32}{5}$ which yields $\Phi(\sp_{ij})-\Phi(\sp) > \frac{-9}{5} $.

				If $o_u = 1$, (\ref{red}) implies $b_v<\frac{24}{5}$ which yields $\Phi(\sp_{ij})-\Phi(\sp) > \frac{-3}{5} $.

				If $o_u = 2$, (\ref{red}) implies $b_v<\frac{16}{5}$ which yields $\Phi(\sp_{ij})-\Phi(\sp) > \frac{3}{5} $.

				If $o_u = 3$, (\ref{red}) implies $b_v<\frac{8}{5}$ which yields $\Phi(\sp_{ij})-\Phi(\sp) > \frac{9}{5} $.

				If $o_u = 4$, $i$ and $j$ cannot perform a profitable local swap since $\u_j(\sp_{ij}) = 0$.

				Since $\Phi(\sp)$ is integral, we just have to show that, if $o_u = 1$, the swap is in fact not $\Phi$-preserving, but $\Phi$-increasing. 
				Notice that $b_v$ is an integer as well. Hence, since (\ref{red}) implies $b_v<\frac{24}{5}$, it holds that $b_v \leq 4$ which yields $\Phi(\sp_{ij})-\Phi(\sp) \geq 1$.
			\end{itemize}
		\end{proof}

	\begin{theorem}\label{FIP_grid}
		Any local $2$-SSG played on an 8-grid possesses the FIP. 
	\end{theorem}

	\begin{proof}
		As shown in Lemma~\ref{findamental} and Lemma~\ref{potential_decreasing}, there are only a few local swaps which can preserve or decrease the potential $\Phi$ and all of them decrease it by at most $1$. We will show that after such a $\Phi$-preserving or $\Phi$-decreasing swap a number of swaps must happen before at the same pair of vertices another $\Phi$-preserving or $\Phi$-decreasing swap can occur. We will show that in total the extended potential $\Psi$ increases which implies the FIP. Remember, that the extended potential $\Psi$ is simply a more fine-grained version of the potential $\Phi$. Thus, for simplicity, for some parts of the proof we will simply work with $\Phi$ instead of $\Psi$.
	
		By Lemma~\ref{findamental}, we know that any profitable swap occurring in a strategy profile $\sp$ is $\Phi$-increasing, unless it involves two agents $i$ and $j$ occupying vertices $\sigma(i)=u$ and $\sigma(j)=v$ with $\delta_u\neq\delta_v$. We assume, without loss of generality, $\delta_u <\delta_v$ and that $c(i)$ is orange and $c(j)$ is blue. Therefore, we only have to consider the following cases: 
		\begin{itemize}
			\item[{\em i)}] $\delta_u = 3$ and $\delta_v = 5$.
			\item[] By Lemma~\ref{findamental}, we know that we have a $\Phi$-increasing swap unless $\u_j(\sp) \in \left(\frac12,1\right)$, in which case we may have a $\Phi$-preserving swap. If this happens, as $\delta_v = 5$ and $\u_j(\sp_{ij}) > \u_j(\sp) > \frac12$, it must be that $\u_j(\sp) = \frac35$ and $\u_j(\sp_{ij})=\frac23$ which imply that, in $\sp$, all vertices adjacent to $u$ are occupied by blue agents, so $\u_i(\sp)=0$. Hence, in $\sp_{ij}$, all vertices in $N_u\setminus\{v\}$ are occupied by blue agents. Thus, in order for agent $j$ (occupying vertex $u$ in $\sp_{ij}$) to be involved once again in a $\Phi$-preserving profitable swap, all vertices in $N_u\setminus\{v\}$ need to change their colors. 
		
			Consider $\sp_{ij}$ in Figure~\ref{local_8grid_1}.
			For all $w \in  N_u\setminus\{v\}$, we have $\u_{\sp_{ij}^{-1}(w)}(\sp_{ij}) > 0$. Let $w_1$ and~$w_2$ be the unique vertices in $N_u\setminus\{v\}$ with $\delta_{w_1} = 5$ and $\delta_{w_2} = 8$, respectively. For vertex~$w_1$ to 	change its color, the agent occupying $w_1$ in $\sp_{ij}$ can either swap with another agent occupying a vertex $z_1$ with $\delta_{z_1} = 5$ which is $\Phi$-increasing by Lemma~\ref{findamental} or with an agent on $z_2$ with $\delta_{z_2} = 8$ which is $\Phi$-increasing by Lemma~\ref{potential_decreasing}. Also a swap with $j$ is not possible, since the agent occupying $w_1$ has the same color as $j$.
		
      	  \begin{figure}[t]
      		  \centering
				\begin{subfigure}[c]{0.18\textwidth}
					\includegraphics[width=0.8\textwidth]{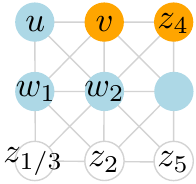}
					\subcaption{~}
					\label{local_8grid_1}
				\end{subfigure}
				~~~~~~
				\begin{subfigure}[c]{0.16\textwidth}
					\includegraphics[width=0.8\textwidth]{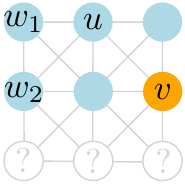}
					\subcaption{~}
					\label{local_8grid_4}
				\end{subfigure}
				~~~~~~
				\begin{subfigure}[c]{0.16\textwidth}
					\includegraphics[width=0.8\textwidth]{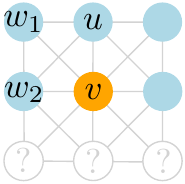}
					\subcaption{~}
					\label{local_8grid_4a}
				\end{subfigure}	
			\caption{The coloring of $G$ in $\sp_{ij}$ after a $\Phi$-preserving or $\Phi$-decreasing swap of agents $i$ and $j$ occupying vertices $u$ and $v$. Vertices with question marks which are neither blue nor orange can be occupied by an agent of any type. Symmetric cases are omitted.}
			\label{local_8grid_case1}
		\end{figure}
    
		If the agent on $w_1$ performs a $\Phi$-preserving swap with an agent on $z_3$ with $\delta_{z_3} = 3$ this implies that the vertex $w_2$ is surrounded by at least $3$ orange agents which implies that $\u_{\sp_{ij}^{-1}(w_2)} \leq \frac{5}{8}$. Hence, by Lemma~\ref{potential_decreasing}, the agent occupying $w_2$ can only perform a $\Phi$-increasing swap with an agent on a vertex with degree $5$. Also, again by Lemma~\ref{potential_decreasing}, any swap with an agent on a vertex with degree $3$ must be $\Phi$-increasing, since either this swap is with an orange agent on $z_4$, that is a degree $3$ neighbor of $v$, or with an orange agent on $z_5$, which is the remaining possible degree $3$ vertex in $w_2$'s neighborhood. A swap with an agent on $z_4$ must be $\Phi$-increasing since, as $v$ is occupied by an orange agent, the agent on $z_4$ has non-zero utility. If $z_5$ is occupied with an orange agent, then the agent on $w_2$ has a utility of at most $\frac{4}{8}$ since $v$, $z_3$, $z_4$ and $z_5$ are occupied by orange type agents.   
		
		Thus, in order for an agent occupying $u$ to perform once again a profitable $\Phi$-preserving swap, a profitable $\Phi$-increasing swap has to occur.
		
		\item[{\em ii)}] $\delta_u = 5$ and $\delta_v = 8$.
		\item[] By Lemma~\ref{potential_decreasing} we know that we have a $\Phi$-decreasing swap by $1$ if and only if $\u_i(\sp) = 0$ and $\u_j(\sp) = \frac{6}{8}$, which implies that all agents occupying vertices adjacent to $u$ in $\sp$ are blue. Thus, in order for agent $j$ (occupying vertex $u$ in $\sp_{ij}$) to be involved once again in a $\Phi$-decreasing profitable swap, all vertices in $N_u\setminus\{v\}$ need to change their colors. Note that by Lemma~\ref{potential_decreasing}, a $\Phi$-preserving swap for agent $j$ is impossible.
		We distinguish several cases:
				\textbf{Case 1.} In the first case, we assume that $w_1 \in  N_u\setminus\{v\}$ is a corner vertex, i.e., $\delta_{w_1} = 3$. Let $w_2 \in N_u\setminus\{v\}$ be a vertex adjacent to $w_1$ with $\delta_{w_2} = 5$, cf. Figure~\ref{local_8grid_4} and Figure~\ref{local_8grid_4a}. Notice, that the agents occupying $w_1$ and $w_2$ have utility $\u_{\sp_{ij}^{-1}(w_1)}(\sp_{ij}) > 0$ and $\u_{\sp_{ij}^{-1}(w_2)}(\sp_{ij}) > 0$, respectively, since both neighboring vertices are occupied by blue agents. Hence, the agent occupying vertex $w_1$ must perform at least two $\Phi$-increasing swaps to leave the neighborhood of $u$, which is necessary in order for agent $j$ to perform once again a profitable $\Phi$-decreasing swap.
				
				\textbf{Case 2.} In the second case, we assume that $w_1 \in  N_u\setminus\{v\}$ is a border vertex, i.e., $\delta_{w_1} = 5$. Let $w_2 \in N_u\setminus\{v\}$ be a vertex adjacent to $w_1$ with $\delta_{w_2} = 8$. Notice, that the agents occupying $w_1$ and $w_2$ have utility $\u_{\sp_{ij}^{-1}(w_1)}(\sp_{ij}) > 0$ and $\u_{\sp_{ij}^{-1}(w_2)}(\sp_{ij}) > 0$, respectively. Let $w_1'$, $w_2'$ and $w_3'$ be adjacent to $w_1$ and $w_2$ as depicted in Figure~\ref{local_8grid_5} and Figure~\ref{local_8grid_5a}.
   				 \begin{figure}[t]
   					\centering
					\begin{subfigure}[c]{0.22\textwidth}
						\includegraphics[width=0.8\textwidth]{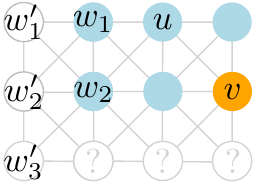}
						\subcaption{~}
						\label{local_8grid_5}
					\end{subfigure}
					~~~~~~
					\begin{subfigure}[c]{0.22\textwidth}
						\includegraphics[width=0.8\textwidth]{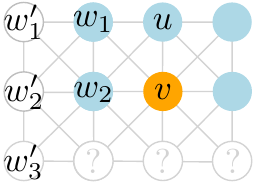}
						\subcaption{~}
						\label{local_8grid_5a}
					\end{subfigure}	
					\caption{The coloring of $G$ in $\sp_{ij}$ after a $\Phi$-decreasing swap of agents $i$ and $j$ occupying vertices $u$ and $v$. The question mark means that the vertex can be occupied by an agent of any type. Symmetric cases are omitted.}
					\label{local_8grid_case2}
				\end{figure}
				
				\textit{Case 2a.} The agents occupying vertices $w_1'$, $w_2'$ and $w_3'$ have utility $\u_{\sp_{ij}^{-1}(w_1')}(\sp_{ij}) > 0$, $\u_{\sp_{ij}^{-1}(w_2')}(\sp_{ij}) > 0$ and $\u_{\sp_{ij}^{-1}(w_3')}(\sp_{ij}) > 0$. In this case, by Lemma~\ref{potential_decreasing}, the agents occupying $w_1$ and $w_2$ cannot leave the neighborhood of $u$ via a $\Phi$-preserving or $\Phi$-decreasing swap. Both must perform $\Phi$-increasing swaps which increases $\Phi$ by at least $2$.

				\textit{Case 2b.} The agent occupying $w_1'$ has utility $\u_{\sp_{ij}^{-1}(w_1')}(\sp_{ij}) = 0$ and the agent residing on $w_3'$ has utility $\u_{\sp_{ij}^{-1}(w_3')}(\sp_{ij}) > 0$. From the former, it follows that the agent on $w_1'$ must be orange and the agent occupying $w_2'$ must be blue. Moreover, the agent on $w_2'$ must have non-zero utility. See Figure~\ref{local_8grid_5_2b} and Figure~\ref{local_8grid_5a_2b}.
				\begin{figure}[h]
    				\centering
					\begin{subfigure}[c]{0.2\textwidth}
						\includegraphics[width=0.8\textwidth]{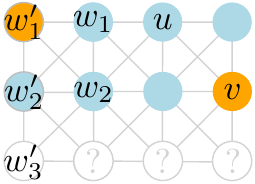}
						\subcaption{~}
						\label{local_8grid_5_2b}
					\end{subfigure}
					~~~
					\begin{subfigure}[c]{0.2\textwidth}
						\includegraphics[width=0.8\textwidth]{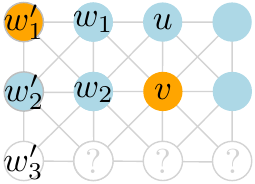}
						\subcaption{~}
						\label{local_8grid_5a_2b}
					\end{subfigure}	
					~~~
					\begin{subfigure}[c]{0.2\textwidth}
						\includegraphics[width=0.8\textwidth]{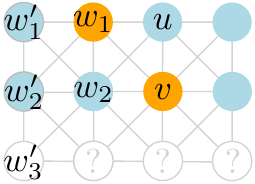}
						\subcaption{~}
						\label{local_8grid_5a_2ba}
					\end{subfigure}	
					~~~
					\begin{subfigure}[c]{0.2\textwidth}
						\includegraphics[width=0.8\textwidth]{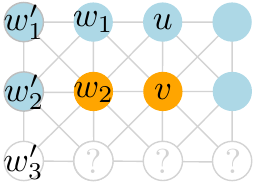}
						\subcaption{~}
						\label{local_8grid_5a_2bb}
					\end{subfigure}	
					\caption{(a) and (b): the coloring of $G$ in $\sp_{ij}$ after a $\Psi$-decreasing swap by $(-1,+1)$ of agents~$i$ and $j$ occupying vertices $u$ and $v$ in case 2b. Symmetric cases are omitted. (c) and (d) show~$\sp_1$ and $\sp_2$, respectively, after a swap starting with the situation in (b).}
					\label{local_8grid_case2_2b}
				\end{figure}
				
				We analyze this case by focusing on the extended potential $\Psi$. Remember that $\Psi$ is essentially $\Phi$ with the number of agents having utility $0$ as a tie-breaker. Since $\Psi$ is a vector, we denote the change in $\Psi$ by a profitable swap as $(\lambda,\mu)$ with $\lambda,\mu \in \{+1,=,-1\}$, where $\lambda$ denotes the change in $\Phi$ and $\mu$ denotes the change in $n-z(\cdot)$. Note that the swap from $\sp$ to $\sp_{ij}$ yields a change in $\Psi$ of $(-1,+1)$. We will now show that for vertex $u$ to be surrounded again by agents of the other color, $\Psi$ will in total increase lexicographically. In particular, it suffices to focus on the change in $\Psi$ induced by vertices $w_1$ and $w_2$ becoming occupied by an orange agent. Clearly, if both swaps are $\Phi$-increasing, then also $\Psi$ increases lexicographically. Hence, we focus on the cases where one of these swaps is not $\Phi$-increasing.  
				
				If $\delta_{w_1'} = 3$ then the agent occupying $w_1$ can be involved in a $\Phi$-preserving swap with the agent on $w_1'$. This swap yields a change in $\Psi$ of $(=,+1)$ since both agents have non-zero utility after the swap. This results in $\sp_1$, see Figure~\ref{local_8grid_5a_2ba}. However, the agent on $w_2$ has non-zero utility since $u$ is in its neighborhood and all agents in her neighborhood on vertices with degree $3$ or $5$ have non-zero utility. Thus, by Lemma~\ref{potential_decreasing}, the agent on $w_2$ can only perform a profitable swap which changes $\Psi$ by $(+1,=)$. In total, $\Psi$ must change by at least $(=,+1)$ which implies a lexicographic increase. 
				
				The only other case is that the agent occupying $w_2$ can be involved in a $\Phi$-decreasing swap with the agent residing on $w_1'$ in both cases $\delta_{w_1'} = 3$ or $\delta_{w_1'} = 5$. Let $\sp_2$ denote the corresponding strategy profile, see Figure~\ref{local_8grid_5a_2bb}. Note that this swap changes $\Psi$ by $(-1,+1)$. Now there are two ways of how $w_1$ can become occupied by an orange agent. First, if the agents on $w_1$ and $w_2$ or $v$ swap, then, by Lemma~\ref{potential_decreasing}, $\Psi$ changes by $(+1,=)$. After this swap, the blue agent on $w_2$ or $v$ with non-zero utility has to perform another swap with an orange agent, which changes $\Psi$ again by $(+1,=)$. The second way of $w_1$ becoming occupied by an orange agent is that first vertex $w_1'$ or $w_2'$ becomes occupied by an orange agent and then this agent swaps with the agent on $w_1$. However, both these swaps each change $\Psi$ by $(+1,=)$. In total, $\Psi$ lexicographically changes by at least $(-1,+1)$, $(-1,+1)$, $(+1,=)$ and $(+1,=)$ which in total yields a lexicographic increase.  
				
				\textit{Case 2c.} The agent residing $w_2'$ has utility $\u_{\sp_{ij}^{-1}(w_2')}(\sp_{ij}) = 0$. It follows that the agent on $w_2'$ is orange and all neighboring agents must be blue. Moreover, the agents on $w_1'$ and $w_3'$ must have non-zero utility. See Figure~\ref{local_8grid_5_2c} and Figure~\ref{local_8grid_5a_2c}. 
				\begin{figure}[h]
  				  	\centering
					\begin{subfigure}[c]{0.2\textwidth}
						\includegraphics[width=0.8\textwidth]{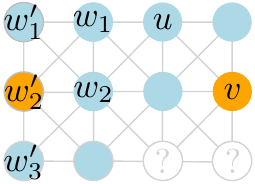}
						\subcaption{~}
						\label{local_8grid_5_2c}
						\end{subfigure}
						~~~~~~
						\begin{subfigure}[c]{0.2\textwidth}
							\includegraphics[width=0.8\textwidth]{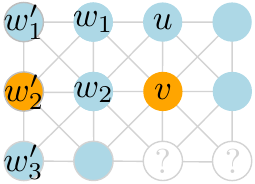}
							\subcaption{~}
							\label{local_8grid_5a_2c}
						\end{subfigure}	
						~~~~~~
						\begin{subfigure}[c]{0.2\textwidth}
							\includegraphics[width=0.8\textwidth]{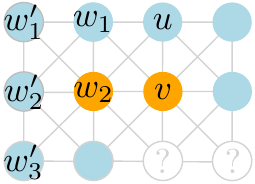}
							\subcaption{~}
							\label{local_8grid_5a_2ca}
						\end{subfigure}	
						\caption{(a) and (b): the coloring of $G$ in $\sp_{ij}$ after a $\Psi$-decreasing swap by $(-1,+1)$ of agents $i$ and $j$ occupying vertices $u$ and $v$ in case 2c. Symmetric cases are omitted. (c) shows $\sp_3$ after a swap starting with the situation in (a) or (b).}
						\label{local_8grid_case2_2c}
				\end{figure}
				
				If $\delta_{w_2'} = 8$, the agents occupying $w_1$ and $w_2$ cannot be involved in a $\Phi$-preserving or $\Phi$-decreasing swap and therefore both must perform $\Phi$-increasing swaps to leave the neighborhood of $u$. In total this yields an increase in $\Phi$ and thus also in $\Psi$. 
				
				If $\delta_{w_2'} = 5$ then the agent occupying $w_2$ can be involved in a $\Phi$-decreasing swap with the agent on $w_2'$. Note that this swap changes $\Psi$ by $(-1,+1)$ and let $\sp_3$ denote the resulting strategy profile, see Figure~\ref{local_8grid_5a_2ca}. This yields a similar situation as in $\sp_2$ in Figure~\ref{local_8grid_5a_2bb} and we can argue analogously that at least two swaps which each change $\Psi$ by $(+1,=)$ must happen. In total, $\Psi$ increases lexicographically.
				
				\textit{Case 2d.} The agents occupying $w_1'$ and $w_3'$ have utility $\u_{\sp_{ij}^{-1}(w_1')}(\sp_{ij}) = 0$ and $\u_{\sp_{ij}^{-1}(w_3')}(\sp_{ij}) = 0$. This implies that the agents on $w_1'$ and $w_3'$ must be orange, the agent on $w_2'$ must be blue and that the latter has non-zero utility. See Figure~\ref{local_8grid_5_2d} and Figure~\ref{local_8grid_5a_2d} 
				\begin{figure}[t]
    \centering
	\begin{subfigure}[c]{0.2\textwidth}
	\includegraphics[height=2.5cm]{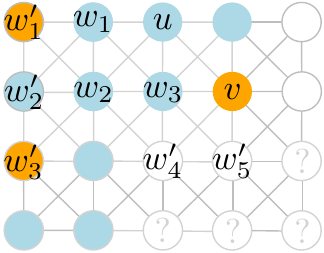}
	\subcaption{~}
	\label{local_8grid_5_2d}
\end{subfigure}
~~~
\begin{subfigure}[c]{0.2\textwidth}
	\includegraphics[height=2.5cm]{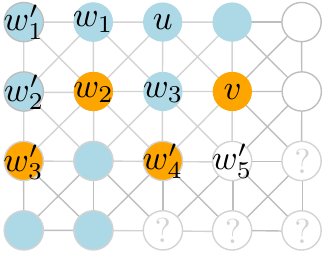}
	\subcaption{~}
	\label{local_8grid_5_2da}
\end{subfigure}	
~~~
\begin{subfigure}[c]{0.2\textwidth}
	\includegraphics[height=2.5cm]{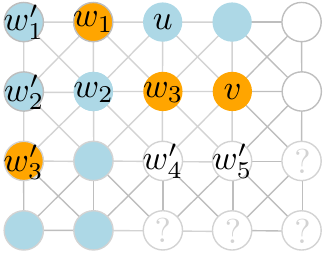}
	\subcaption{~}
	\label{local_8grid_5_2db}
\end{subfigure}	
~~~
\begin{subfigure}[c]{0.22\textwidth}
	\includegraphics[height=2.5cm]{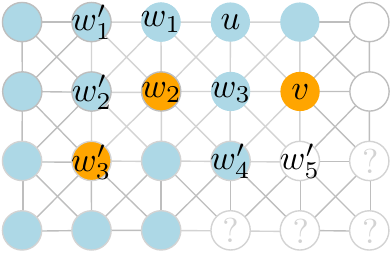}
	\subcaption{~}
	\label{local_8grid_5_2dc}
\end{subfigure}	
	\caption{(a) the coloring of $G$ in $\sp_{ij}$ after a $\Psi$-decreasing swap by $(-1,+1)$ of agents $i$ and $j$ occupying vertices $u$ and $v$ in case 2d. (b) the strategy profile $\sp_4$. (c) the strategy profile $\sp_5$. (d) the strategy profile $\sp_6$. Symmetric cases are omitted.}
		\label{local_8grid_case2_2d}
	\end{figure}

				We consider the settings in Figure~\ref{local_8grid_5_2d} and Figure~\ref{local_8grid_5a_2d} separately. 
				
				We start with the setting in Figure~\ref{local_8grid_5_2d}.\\ Let $\delta_{w_1'} = 3$. In this case a swap of the agents on $w_1'$ and $w_1$ is not profitable. However, by Lemma~\ref{potential_decreasing}, the agent on $w_1'$ could perform a $\Phi$-decreasing swap with the agent on $w_2$ if and only if the agent on $w_4'$ is orange. This swap would change $\Psi$ by $(-1,+1)$ and yields strategy profile $\sp_4$ depicted in Figure~\ref{local_8grid_5_2da}. Now for the agent on $w_1$ to become orange, at least two $\Phi$-increasing swaps are necessary: a swap with the agent on $w_2$ is not profitable, so at least one of the vertices $w_1',w_2',w_3$ must become occupied by an orange agent before a swap with the agent on $w_1$ is possible. Thus, in total $\Psi$ increases, since there are at least two $\Phi$-increasing swaps necessary. In the setting in Figure~\ref{local_8grid_5_2d}, a swap between the agents on $w_1'$ and $w_1$ can only be $\Phi$-neutral, if vertex $w_3$ becomes occupied by an orange agent and vertices $w_2'$ and $w_2$ remain occupied by blue agents. In this case, the swap of the blue agent on $w_3$ must be $\Phi$-increasing. However, this is not a profitable swap for an orange agent, since an agent occupying vertex $w_4'$ or $w_5'$ will not gain additional orange neighbors by swapping to $w_3$. Hence, the strategy profile $\sp_5$ depicted in Figure~\ref{local_8grid_5_2db} is not possible.\\ 
				If $\delta_{w_1'} = 5$, then the agent occupying $w_2$ can be involved in a profitable swap with the agent on $w_1'$ which decreases $\Psi$ by $(-1,+1)$, but, by Lemma~\ref{potential_decreasing}, only if $w_4'$ is occupied by a blue agent. After the swap we get the strategy profile $\sp_6$ depicted in Figure~\ref{local_8grid_5_2dc}. Now, the agent on $w_1$ is in an analogous situtation as in $\sp_4$ depicted in Figure~\ref{local_8grid_5_2da}. By analogous reasoning, at least two $\Phi$-increasing swaps must happen so that vertex $w_1$ can become occupied by an orange agent. This implies that in total $\Psi$ increases. 
            
				\begin{figure}[t]
    \centering
\begin{subfigure}[c]{0.2\textwidth}
	\includegraphics[width=0.8\textwidth]{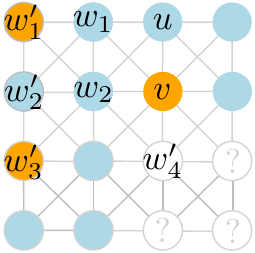}
	\subcaption{~}
	\label{local_8grid_5a_2d}
\end{subfigure}	
~~~
\begin{subfigure}[c]{0.2\textwidth}
	\includegraphics[width=0.8\textwidth]{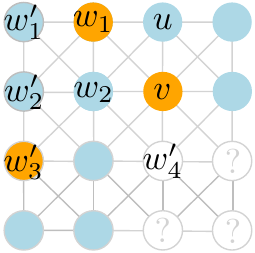}
	\subcaption{~}
	\label{local_8grid_5a_2da}
\end{subfigure}	
~~~
\begin{subfigure}[c]{0.2\textwidth}
	\includegraphics[width=0.8\textwidth]{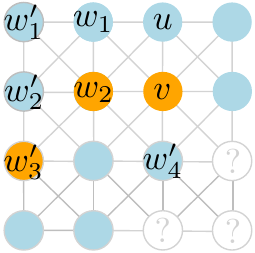}
	\subcaption{~}
	\label{local_8grid_5a_2db}
\end{subfigure}	
~~~
\begin{subfigure}[c]{0.2\textwidth}
	\includegraphics[width=0.8\textwidth]{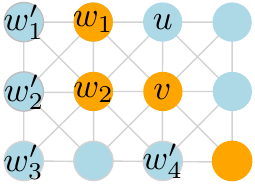}
	\hspace*{5cm}
	\subcaption{}
	\label{local-8grid_5a_2dbe}
\end{subfigure}	
	\caption{(a) the other option for the coloring of $G$ in $\sp_{ij}$ after a $\Psi$-decreasing swap by $(-1,+1)$ of agents $i$ and $j$ occupying vertices $u$ and $v$ in case 2d. (b) the strategy profile $\sp_7$. (c) the strategy profile $\sp_8$. (d) the strategy profile $\sp_9$. Symmetric cases are omitted.}
		\label{local_8grid_case2_2db}
	\end{figure}
				
				Next, we consider the setting depicted in Figure~\ref{local_8grid_5a_2d}.\\
				Let $\delta_{w_1'} = 3$ and $\delta_{w_3'} = 5$. By Lemma~\ref{findamental}, a swap by the agents on vertices $w_1'$ and $w_1$ changes $\Psi$ by $(=,+1)$ and leads to the strategy profile $\sp_7$ depicted in Figure~\ref{local_8grid_5a_2da}. Now, note that since the agent on $w_2$ has a utility of at most $\frac{5}{8}$, a swap with the agent on $w_3'$ must be $\Phi$-increasing, which in total yields an increase in $\Psi$. Another possibility is that in $\sp_{ij}$ depicted in Figure~\ref{local_8grid_5a_2d} the agents on $w_1'$ and $w_2$ swap. By Lemma~\ref{potential_decreasing}, this swap changes $\Psi$ by $(-1,+1)$ if and only if $w_4'$ is occupied by a blue agent. Let $\sp_8$ be the resulting strategy profile which is depicted in Figure~\ref{local_8grid_5a_2db}. Now, note that the agent on $w_1$ is in a similar situation as the agent on $w_1$ in $\sp_2$ in Figure~\ref{local_8grid_5a_2bb}. With an analogous reasoning we get that at least two $\Phi$-increasing swaps must happen so that $w_1$ becomes occupied by an orange agent. In total we get an increase in $\Psi$.\\
				Let $\delta_{w_1'} = 5$ and thus $\delta_{w_3'} = 5$ or $\delta_{w_3'} = 8$. In this case, since the agent on $w_2$ has utility of at most $\frac{5}{8}$ and by Lemma~\ref{potential_decreasing} no $\Phi$-decreasing swaps involving the agents on $w_1$ and $w_2$ are possible. Thus, in total at least two $\Phi$-increasing swaps must occur so that $w_1$ and $w_2$ become occupied by orange agents which implies a total increase in $\Phi$ and thus also in $\Psi$. 
				
				The last remaining situation in the setting depicted in Figure~\ref{local_8grid_5a_2d} is that both $w_1'$ and $w_3'$ are corner vertices, hence, $\delta_{w_1'} = 3$ and $\delta_{w_3'} = 3$. By Lemma~\ref{findamental}, a swap by the agents on vertices $w_1'$ and $w_1$ changes $\Psi$ by $(=,+1)$ and leads to the strategy profile similar to $\sp_7$ depicted in Figure~\ref{local_8grid_5a_2da}. If the agent on $w_2$ has a utility of $\frac{5}{8}$ a swap with the agent on $w_3'$  changes $\Psi$ by $(-1,+1)$ and leads to the strategy profile $\sp_9$ depicted in Figure~\ref{local-8grid_5a_2dbe}. To be situated in the same situation, that the agents occupying $u$, $w_1'$ and $w_3'$ are involved in $\Phi$-decreasing or $\Phi$-preserving swaps, the agent $w_2'$ has to perform two $\Phi$-decreasing swaps to leave the neighborhood of $w_1'$ and $w_3'$. In total $\Psi$ must change by at least $(=,+3)$ which implies a lexicographic increase.
				
				\textit{Case 2e.} The agent occupying $w_3'$ has utility $\u_{\sp_{ij}^{-1}(w_3')}(\sp_{ij}) = 0$ and the agents on $w_1'$ and $w_2'$ have non-zero utility. It follows that $w_3'$ and $w_2'$ must be occupied by an orange and a blue agent, respectively. If $\delta_{w_1'} = 3$, then also $w_1'$ must have a blue agent. Otherwise, the agent on $w_1'$ can also be orange but then it must have another neighboring orange neighbor. See Figure~\ref{local_8grid_case2_2e} for all possible settings.
				\begin{figure}[h]
					\hspace*{1em}
    \centering
\begin{subfigure}[c]{0.18\textwidth}
	\includegraphics[height=2.5cm]{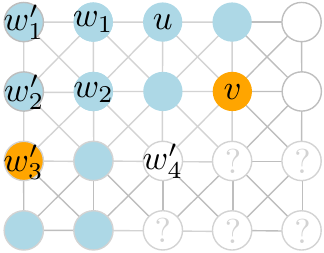}
	\subcaption{~}
	\label{local_8grid_5_2e}
\end{subfigure}	
~~
\begin{subfigure}[c]{0.22\textwidth}
	\includegraphics[height=2.5cm]{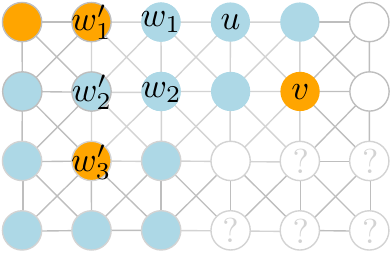}
	\subcaption{~}
	\label{local_8grid_51_2e}
\end{subfigure}	
~~
\begin{subfigure}[c]{0.18\textwidth}
	\includegraphics[height=2.5cm]{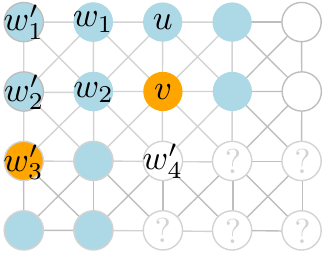}
	\subcaption{~}
	\label{local_8grid_5a_2e}
\end{subfigure}	
~~
\begin{subfigure}[c]{0.22\textwidth}
	\includegraphics[height=2.5cm]{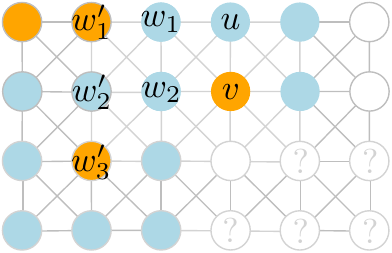}
	\subcaption{~}
	\label{local_8grid_5a1_2e}
\end{subfigure}	
	\caption{Possible strategy profiles $\sp_{ij}$ in case 2e. Symmetric cases are omitted.}
		\label{local_8grid_case2_2e}
	\end{figure}

Let $\delta_{w_3'} = 3$. Then, by Lemma~\ref{potential_decreasing}, a profitable swap of the agents $w_3'$ and $w_2$ which decreases $\Psi$ by $(-1,+1)$ is possible, if the agent on $w_4'$ is orange so that the agent on $w_2$ has utility $\frac{5}{8}$, cf.  Figure~\ref{local_8grid_5a_2e}. After this swap we get a strategy profile which, from the point of view of the agent on $w_1$, is analogous to $\sp_2$ in Figure~\ref{local_8grid_5a_2bb}. Hence, at least two $\Phi$-increaing swaps are necessary so that vertices $v$, $w_1$ and $w_2$ become occupied by orange agents. Thus, in total $\Psi$ increases.
	
	Let $\delta_{w_3'} = 5$ (cf. Figure~\ref{local_8grid_5_2e} and Figure~\ref{local_8grid_5a_2e}). Then, by Lemma~\ref{potential_decreasing}, a profitable swap of the agents on $w_3'$ and $w_2$ which decreases $\Psi$ by $(-1,+1)$ is possible, if the agent on $w_4'$ has a suitable type so that the agent on $w_2$ has utility $\frac{6}{8}$. After this swap we get a strategy profile which, from the point of view of the agent on $w_1$, is analogous to $\sp_2$ in Figure~\ref{local_8grid_5a_2bb} or $\sp_4$ in Figure~\ref{local_8grid_5_2da}. In both cases at least two $\Phi$-increaing swaps are necessary so that vertices $v$, $w_1$ and $w_2$ become occupied by orange agents. Thus, in total $\Psi$ increases.
	
	Let $\delta_{w_3'} = 8$ (cf. Figure~\ref{local_8grid_51_2e} and Figure~\ref{local_8grid_5a1_2e}). In this case no $\Phi$-decreasing or $\Phi$-preserving swaps which involve the agents on $w_1$ or $w_2$ are possible. Thus, at least two $\Phi$-increasing swaps must happen so that $w_1$ and $w_2$ become occupied by orange agents. Hence, in total $\Phi$ and thus also $\Psi$ increases. 
	
	Since we have completed all possible combinations for agents with zero utility on the vertices $w_1',w_2',w_3'$ this finishes case (ii).

		\item[iii)] $\delta_u(G) = 3$ and $\delta_v(G) = 8$
		\item[] By Lemma~\ref{potential_decreasing} we know that we have a $\Phi$-decreasing swap by $1$ if and only if $\u_i(\sp) = 0$ and $\u_j(\sp) = \frac{5}{8}$. This implies that all vertices adjacent to $u$ are occupied by blue agents. Thus, in order for agent $j$ (occupying vertex $u$ in $\sp_{ij}$) to be involved once again in a $\Phi$-decreasing profitable swap, all vertices in $N_u\setminus\{v\}$ must become occupied by orange agents. 
		
		Consider $\sp_{ij}$ in Figure~\ref{local_8grid_2}. Notice, that all neighboring vertices of $w_1$ and $w_2$ must be occupied by agents with non-zero utility, since $v$ is occupied by the orange agent $i$ in $\sp_{ij}$. Hence, no neighboring agent of $w_1$ and $w_2$ can be included in a $\Phi$-decreasing swap before agent $i$ on vertex $v$ performs another profitable swap. Hence, we have to distinguish between two cases.
		 \begin{figure}[t]
    \centering
	\begin{subfigure}[c]{0.16\textwidth}
		\includegraphics[width=0.8\textwidth]{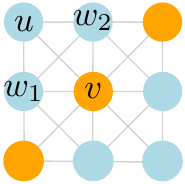}
		\subcaption{~}
		\label{local_8grid_2}
	\end{subfigure}
~~~
	\begin{subfigure}[c]{0.16\textwidth}
	\includegraphics[width=0.8\textwidth]{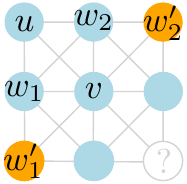}
	\subcaption{~}
	\label{local_8grid_3}
\end{subfigure}
	\caption{We focus on the change in $\Phi$  induced by vertices $w_1$ and $w_2$ where the swaps are not $\Phi$-increasing. (a) strategy profile $\sp_{ij}$ after a $\Phi$-preserving or $\Phi$-decreasing swap of agents $i$ and $j$ occupying vertices $u$ and $v$. (b) shows the strategy profile before the agent occupying $v$ can perform another $\Phi$-decreasing swap.}
		\label{local_8grid_case3}
	\end{figure}
		
		\textbf{Case 1.} We assume that agent $i$ does not perform another profitable swap, before the agents placed on $w_1$ and $w_2$ swap. As already mentioned, no neighboring agent of $w_1$ and $w_2$ has utility zero and since the agents on $w_1$ and $w_2$ have positive utility as well, two $\Phi$-increasing swaps will occur before the agent occupying $u$ can perform once again a $\Phi$-decreasing swap. Thus, in total $\Phi$ increases.
		
		\textbf{Case 2.} We assume that agent $i$ will perform another profitable swap before the agents placed on $w_1$ and $w_2$ swap. Hence, it is possible that an agent in the neighborhood of $w_1$ or $w_2$ has utility zero and is involved in a $\Phi$-preserving or $\Phi$-decreasing swap. However, the swap of agent $i$ is $\Phi$-increasing and will be performed with a blue agent. Hence, since~$j$ residing on $u$ is also blue, the color of the agent on $v$ has to change to orange before the agent on $u$ can perform another $\Phi$-decreasing swap. Consider Figure~\ref{local_8grid_3}. If the agent on $w_1'$ or $w_2'$ can perform another $\Phi$-decreasing swap, this is only possible with the agent occupying $v$. Assume, without loss of generality, there is a profitable $\Phi$-decreasing swap between the agents residing on $w_1'$ and $v$. Then, afterwards, for the agent residing on $w_1$ to leave the neighborhood of $u$, there will be at least two $\Phi$-increasing swaps since the agent occupying $w_1'$ is blue and has positive utility. However, this is necessary in order for an agent occupying vertex $u$, to perform once again a profitable $\Phi$-decreasing swap. Thus, in total $\Phi$ increases.
	\end{itemize}
	We have shown that after a $\Psi$-decreasing profitable local swap involving agents on two vertices $u$ and $v$ some additional swaps are necessary before another $\Psi$-decreasing swap can happen again involving the same vertices. Moreover, we have shown that in total these additional swaps increase $\Psi$ more than it was decreased by the initial swap. Thus, in total $\Psi$ increases.
\end{proof}

\noindent{Now we will see that compared to the local $k$-SSG, the $k$-SSG on $8$-grids behaves differently. There the FIP does not hold.}

\begin{theorem}
There cannot exist a potential function for the $k$-SSG played on an 8-grid, for any $k \geq 2$.
\label{IRC_grid}
\end{theorem}

\begin{proof}
\begin{figure}[h]
\begin{subfigure}[c]{0.24\textwidth}
\includegraphics[width=0.9\textwidth]{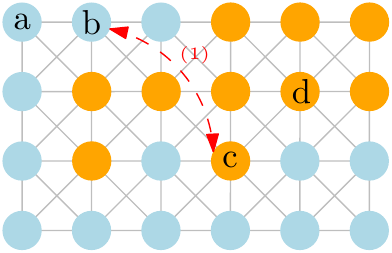}
\subcaption{Initial strategy\\ profile~~~}
\label{IRC_1}
\end{subfigure}
\begin{subfigure}[c]{0.24\textwidth}
	\includegraphics[width=0.9\textwidth]{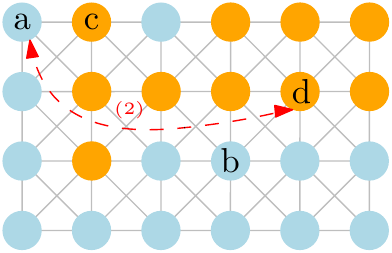}
	\subcaption{Strategy profile\\ after the first swap}
	\label{IRC_2}
\end{subfigure}
\begin{subfigure}[c]{0.24\textwidth}
	\includegraphics[width=0.9\textwidth]{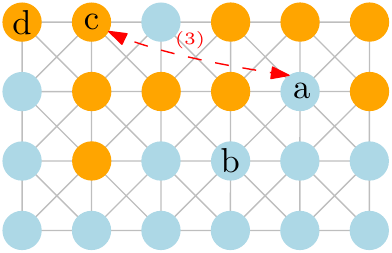}
	\subcaption{Strategy profile\\ after the second swap}
	\label{IRC_3}
\end{subfigure}
\begin{subfigure}[c]{0.24\textwidth}
	\includegraphics[width=0.9\textwidth]{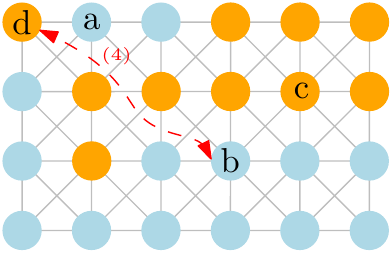}
	\subcaption{Strategy profile\\ after the third swap}
	\label{IRC_4}
\end{subfigure}
\caption{An improving response cycle for the $k$-SSG played on a 8-grid. The agent types are marked orange and blue.}
\label{IRC}
\end{figure}

We prove the statement by providing an improving response cycle $\sp^0,\ldots,\sp^4$. The construction is shown in Figure~\ref{IRC}, where vertices are labeled with the agent occupying them. We have orange and blue agents. Agents with other types can be placed in a grid outside of the depicted cutout. 

In the initial strategy profile $\sp^0$ (Figure~\ref{IRC_1}), $\u_b(\sp^0) = \frac35$ and $\u_c(\sp^0) = \frac38$. Both agents $b$ and $c$ improve by swapping, since in $\sp^1:=\sp^0_{bc}$ we have $\u_b(\sp^1) = \frac58$ and $\u_c(\sp^1) = \frac25$.
After the first swap (Figure~\ref{IRC_2}), agents $a$ and $d$ can perform a profitable swap, since $\u_a(\sp^1) = \frac13$, $\u_d(\sp^1) = \frac58$ and in $\sp^2:=\sp^1_{ad}$ we have $\u_a(\sp^2) = \frac38$ and $\u_d(\sp^2) = \frac23$.
Then (Figure~\ref{IRC_3}), agents $a$ and $c$ can swap and improve from $\u_a(\sp^2) = \frac38$ and $\u_c(\sp^2) = \frac35$ to $\u_a(\sp^3) = \frac25$ and $\u_c(\sp^3) = \frac58$, respectively, with $\sp^3:=\sp^2_{ac}$. 
Finally (Figure~\ref{IRC_4}), agents $b$ and $d$ can improve by swapping, since $\u_b(\sp^3) = \frac58$, $\u_d(\sp^3) = \frac13$ and in $\sp^4:=\sp^3_{bd}$ we have $\u_b(\sp^4) = \frac23$ and $\u_d(\sp^4) = \frac38$.
Now observe that the coloring induced by $\sp^4$ is the same as the one induced by $\sp^0$ (see Figure~\ref{IRC_1}, where $a$ exchanges position with $b$ and $c$ exchanges position with $d$). So, the sequence of profitable swaps defined above can be repeated over and over mutatis mutandis.
\end{proof}

\noindent However, even if convergence to an equilibrium is not guaranteed for $k \geq 2$, they are guaranteed to exist for $k = 2$.

\begin{theorem}
	Every $2$-SSG played on an $8$-grid has an equilibrium which can be computed in polynomial time.
	\label{thm:existence}
\end{theorem}
\begin{proof}
Assume without loss of generality that the grid is such that $\ell\leq h$. If this is not the case, simply rotate the grid by ninety degrees. We give two different constructions depending on how the number of orange agents compares with the threshold $2\ell-1$.

If $o\geq 2\ell-1$, let $\sp$ be the strategy profile in which orange agents occupy the grid starting from the upper-left corner and proceedings towards the right filling the grid in increasing order of rows, see Figure \ref{fig-eq-grid} for a pictorial example. Denote by $x$ the number of entirely orange rows and by $y$ the number of orange vertices in the unique row containing both orange and blue vertices (if this row exists, otherwise set $y=0$). Moreover, whenever $y\neq 0$, let $u$ be the last orange vertex (i.e., the $y$-th vertex along the $(x+1)$-th row) and $v$ be the first blue one (i.e., the vertex at the right of $u$); again, see Figure \ref{fig-eq-grid} for an example. Observe that, by the assumption $o\geq 2\ell-1$ and the fact that $o\leq b$, the following two properties hold:
\begin{itemize}
\item[{\em (P.1)}] $x\geq 1$ and $x=1$ if and only if $y=\ell-1$;
\item[{\em (P.2)}] $x\leq \frac{h}{2}$ and $x=\frac{h}{2}$ if and only if $y\leq\frac{\ell}{2}$.
\end{itemize}
\begin{figure*}[h]
	\centering
\includegraphics[width=0.3\textwidth]{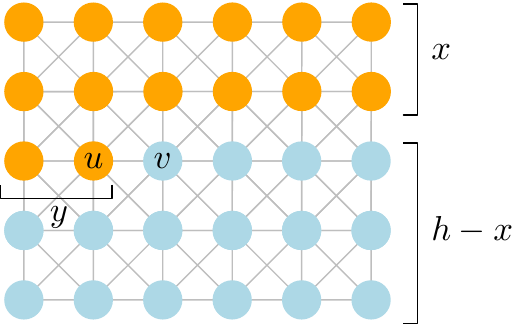}
\caption{The structure of an equilibrium when $o\geq 2\ell-1$.}\label{fig-eq-grid}
\end{figure*}
Fix an orange agent $i$. It is easy to see that, by property {\em (P.1)}, it holds that
\begin{displaymath}
\u_i(\sp)\geq\left\{
\begin{array}{ll} \vspace*{0.5em}
\frac23 & \textrm{ if $\sp(i)$ is a corner vertex},\\ \vspace*{0.5em}
\frac35 & \textrm{ if $\sp(i)$ is a border vertex unless $y=1$ which gives $\u_i(\sp)=\frac25$},\\
\frac58 & \textrm{ if $\sp(i)$ is an inner vertex unless $\sp(i)=u$ which gives $\u_i(\sp)=\frac12$}.
\end{array}\right.
\end{displaymath}
Fix a blue agent $j$. It is easy to see that, by property {\em (P.2)}, it holds that
\begin{displaymath}
\u_j(\sp)\geq\left\{
\begin{array}{ll} \vspace*{0.5em}
\frac23 & \textrm{ if $\sp(j)$ is a corner vertex},\\ \vspace*{0.5em}
\frac35 & \textrm{ if $\sp(j)$ is a border vertex unless $y=\ell-1$ which gives $\u_j(\sp)=\frac25$},\\
\frac58 & \textrm{ if $\sp(j)$ is an inner vertex unless $\sp(j)=v$ which gives $\u_j(\sp)=\frac12$}.
\end{array}\right.
\end{displaymath}
As $\frac25+\min\{\frac23,\frac35,\frac58\}\geq 1$, it follows by Lemma \ref{sum_in_eq} that profitable swaps are possible in $\sp$ only between an orange agent $i$ and a blue agent $j$ satisfying one of the following three conditions: 

\begin{itemize}
\item[(i)] $\u_i(\sp)=\frac25$ and $\u_j(\sp)=\frac25$,
\item[(ii)] $\u_i(\sp)=\frac25$ and $\u_j(\sp)=\frac12$, 
\item[(iii)] $\u_i(\sp)=\frac12$ and $\u_j(\sp)=\frac25$.
\end{itemize} 

Case {\em (i)} requires $1=y=\ell-1$ which implies $\ell=2$ so that $1_{ij}(\sp)=1$. By $\delta_{\sp(i)}=\delta_{\sp(j)}=5$, we get $\u_i(\sp)+\u_j(\sp)\geq 1-\frac{1_{ij}(\sp)}{5}$ satisfying the condition of Lemma \ref{sum_in_eq}. 

Case {\em (ii)} requires $y=1$ (which yields $\sp(i)=u$) and $\sp(j)=v$ so that $1_{ij}(\sp)=1$. By $\delta_{\sp(i)}=5$ and $\delta_{\sp(j)}=8$, we get $\u_i(\sp)+\u_j(\sp)\geq 1-\frac{1_{ij}(\sp)}{5}$ again satisfying the condition of Lemma \ref{sum_in_eq}. 

Case {\em (iii)} requires $y=\ell-1$ (which yields $\sp(j)=v$) and $\sp(i)=u$ so that $1_{ij}(\sp)=1$. By $\delta_{\sp(j)}=5$ and $\delta_{\sp(i)}=8$, we get $\u_i(\sp)+\u_j(\sp)\geq 1-\frac{1_{ij}(\sp)}{5}$ satisfying the condition of Lemma~\ref{sum_in_eq}. Thus, $\sp$ is an equilibrium and can be constructed in polynomial time.

If $o<2\ell-1$, a more involved construction is needed. For any $o\in [14]$, the proposed strategy profile $\sp$ is depicted in Figure~\ref{fig-eq-grid-1}. We stress that the two assumptions $\ell\leq h$ and $o<2\ell-1$ imply that the grid is large enough to accommodate a coloring implementing $\sp$. It is not difficult to check by direct inspection that $\sp$ is an equilibrium. To this aim, it is important to observe that, when $o\geq 7$, there must be at least two blue agents occupying vertices on the first row, otherwise the assumption $o<2\ell-1$ would be contradicted. 
\begin{figure*}[b]
\centering
\includegraphics[width=0.9\textwidth]{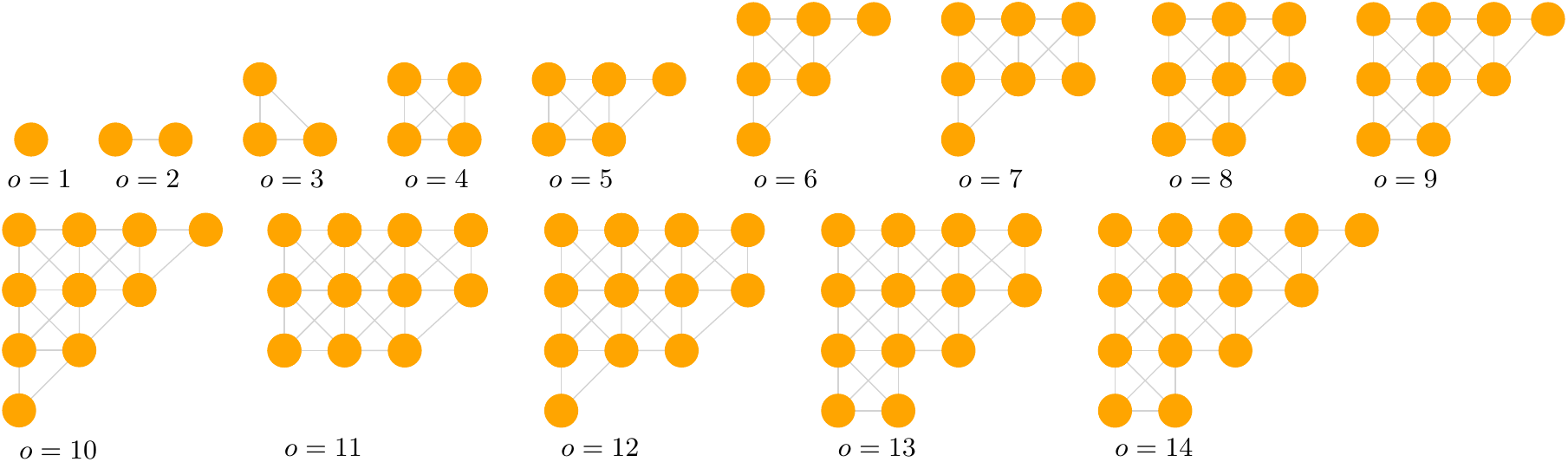}
\caption{The structure of an equilibrium when $o<2\ell-1$ and $o\in [14]$. Only the orange vertices are depicted.}\label{fig-eq-grid-1}
\end{figure*}
Now, for any $15\leq o<2\ell-1$, we propose a general rule, which can be implemented in polynomial time, to construct an equilibrium profile~$\sp$. First, we define some suitable structures. For an integer $x\geq 5$, an $x$-{\em triangle} is a strategy profile obtained as follows: for each $y=x$ down to $1$, $y$ orange agents are assigned to the first $y$ vertices of the $(x+1-y)$-th row, see Figure \ref{fig-eq-grid-2}. Thus, a total of $\frac{x(x+1)}{2}$ orange agents are assigned. 
\begin{figure*}[t]
	\centering
\includegraphics[width=0.2\textwidth]{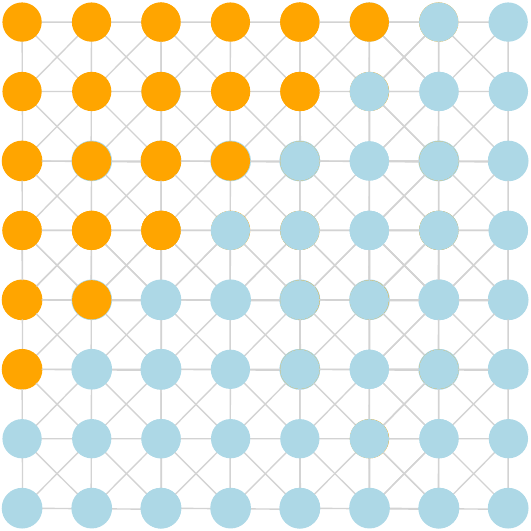}
\caption{The structure of an $x$-triangle, with $x=6$. The grid needs to have additional blue rows and columns which are not depicted.}\label{fig-eq-grid-2}
\end{figure*}
For an integer $x\geq 5$, an $(x,1)$-{\em almost triangle} is a strategy profile obtained by assigning $x$ orange agents to the first $x$ vertices of the first two rows, $x-1$ orange agents to the first $x-1$ vertices of the third row, and then, for each $y=x-3$ down to $2$, $y$ orange agents are assigned to the first $y$ vertices of the $(x+1-y)$-th row, see the top-left picture in Figure \ref{fig-eq-grid-3}. Thus, a total of $\sum_{i=2}^{x-3}i+3x-1=\frac{x(x+1)}{2}+1$ orange agents are assigned. 
\begin{figure*}[h]
	\centering
\includegraphics[width=0.6\textwidth]{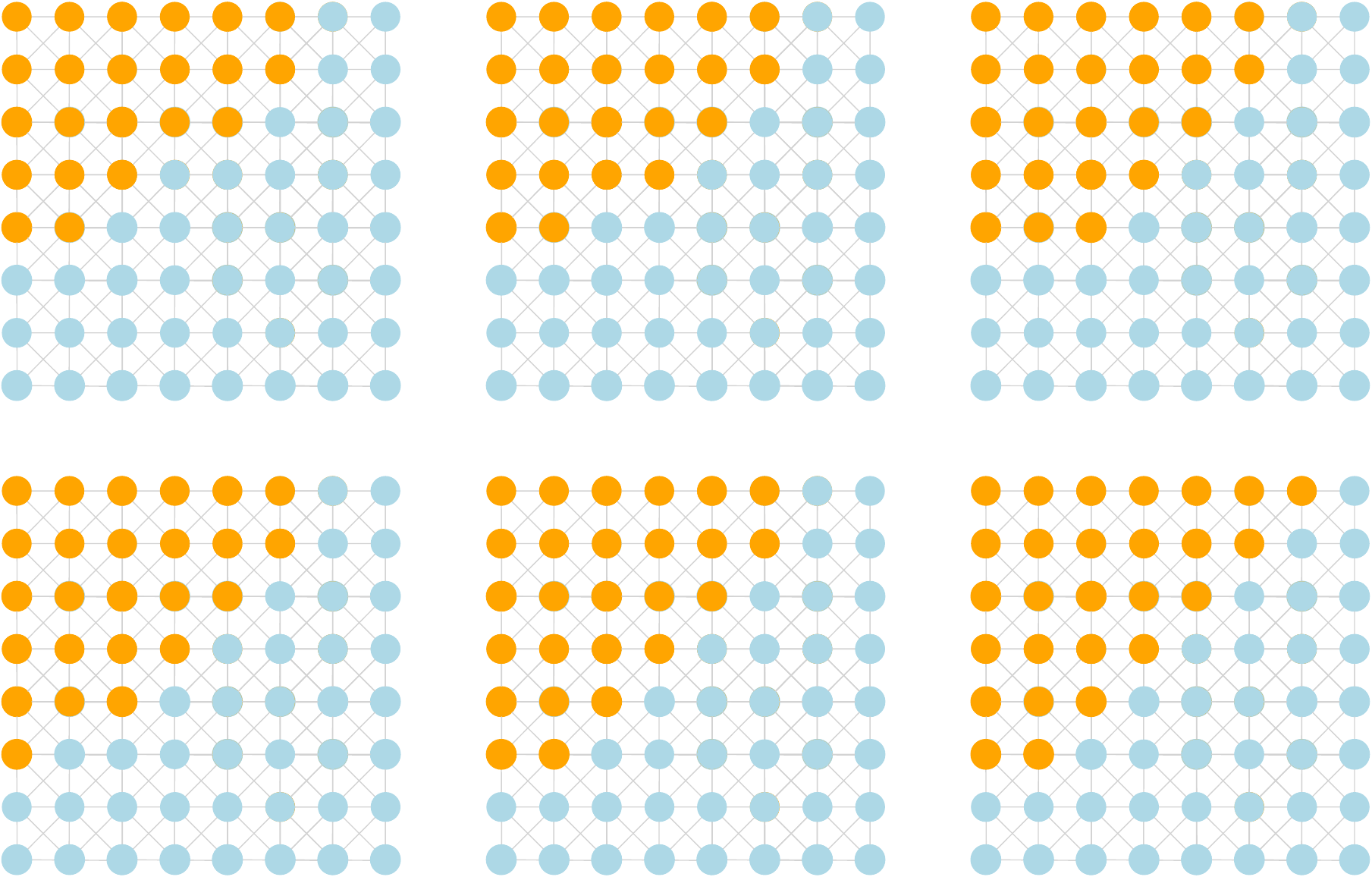}
\caption{The structure of $(x,y)$-triangles, with $x=6$ and $y\in [6]$. The grid needs to have additional blue rows and columns which are not depicted.}\label{fig-eq-grid-3}
\end{figure*}
For a pair of integers $(x,y)$, with $x\geq 5$ and $2\leq y\leq x$, we define an $(x,y)$-{\em almost triangle} as follows: for $2\leq y\leq x-2$, the $(x,y)$-almost triangle is obtained from the $(x,y-1)$-one by locating an orange agent to the first non-orange vertex of the $(y+2)$-th row; the $(x,x-1)$-almost triangle is obtained by locating an orange agent to the first non-orange vertex (i.e., the second) of the $x$-th row of the $(x,x-2)$-one; the $(x,x)$-almost triangle is obtained by locating an orange agent to the first non-orange vertex (i.e., the $(x+1)$-th) of the first row of the $(x,x-1)$-one (see Figure~\ref{fig-eq-grid-3} for a pictorial example). 

Now observe that any number $o\geq 15$ can be decomposed as $o=\frac{x(x+1)}{2}+y$ for some integers~$x$ and $y$ such that $x\geq 5$ and $0\leq y\leq x$. The strategy profile $\sp$ is the $x$-triangle if $y=0$ and the $(x,y)$-almost triangle, otherwise. Clearly, $\sp$ can be constructed in polynomial time. We are left to prove that $\sp$ is an equilibrium. We shall use Lemma~\ref{sum_in_eq} in conjunction with the following claims which can be easily verified with the help of Figures~\ref{fig-eq-grid-2} and \ref{fig-eq-grid-3}. In any $x$-triangle $\sp$ with $x\geq 5$, $\u_i(\sp)\geq \frac25$ for any orange agent $i$ and $\u_j(\sp)\geq \frac58$ for any blue agent $j$. Thus, $\sp$ is an equilibrium. Now, let us consider $(x,y)$-almost triangles. If $y\in [x-3]$, we have $\u_i(\sp)\geq \frac12$ for any orange agent $i$ and $\u_j(\sp)\geq \frac12$ for any blue agent $j$. So, $\sp$ is an equilibrium. If $y=x-2$, $\u_i(\sp)\geq \frac12$ for each orange agent $i$, except for the one occupying the unique orange vertex at the $x$-th row who gets utility equal to $\frac25$; moreover, $\u_j(\sp)\geq \frac58$ for each blue agent $j$, except for the one occupying the first blue vertex of the $x$-th row (see the bottom-left picture in Figure~\ref{fig-eq-grid-3}). Thus, we get $\u_i(\sp)+\u_j(\sp)\geq 1-\frac{1_{ij}(\sp)}{\min\{\delta_{\sp(i)},\delta_{\sp(j)}\}}$ for each orange agent $i$ and blue agent $j$. So, $\sp$ is an equilibrium. If $y=x-1$, $\u_i(\sp)\geq \frac12$ for each orange agent $i$ and $\u_j(\sp)\geq \frac58$ for each blue agent $j$, thus implying that $\sp$ is an equilibrium (see the bottom-middle picture in Figure~\ref{fig-eq-grid-3}). Finally, if $y=x$, $\u_i(\sp)\geq \frac25$ for each orange agent $i$ and $\u_j(\sp)\geq \frac58$ for each blue agent $j$ (see the bottom-right picture in Figure~\ref{fig-eq-grid-3}), and so also in this case $\sp$ is an equilibrium.
\end{proof}

	\section{Price of Anarchy}
	In the following section, we consider the efficiency of equilibrium assignments and bound the PoA for different classes of underlying graphs. In particular, besides investigating general graphs, we analyze regular graphs, cycles, paths, $4$-grids and $8$-grids. Agarwal et al.~\cite{A+19} already proved that the PoA for the $2$-SSG is in $\Theta(n)$ on underlying star graphs if there are at least two agents of each type and between $\frac{921}{448}$ and $4$ for the balanced version, i.e., $o = \frac{n}{2}$. We improve this result by providing an upper bound of $3$ which tends to~$2$ for $n$ going to infinity. Furthermore, the authors of ~\cite{A+19} showed that the PoA can be unbounded for $k \geq 3$. Therefore, we concentrate on the (local) $2$-SSG for several graph classes.

	\subsection{General Graphs}
	Remember that for a $2$-SSG game, we assume that $o$ is the less frequent color.

	We significantly improve and generalize the results of~\cite{A+19} by providing a general upper bound of $\frac{no(n-o)-n}{o(o-1)(n-o)}$ for the case of $o>1$. For balanced games, it yields an upper bound of $\frac{2(n+2)}{n}$ which shows that the PoA tends to $2$ as the number of vertices increases. Moreover, if $\frac{b}{o} \in \mathcal{O}(1)$, the PoA is constant.
	
	With the help of Lemma~\ref{sum_in_eq}, we can now prove our general upper bound for the $2$-SSG.

	\begin{theorem}
		The PoA of $2$-SSGs with $o>1$ is at most $\frac{no(n-o)-n}{o(o-1)(n-o)}$. Hence, the PoA $\in \mathcal{O}\left(\frac{b}{o}\right)$.
		\label{thm:poa_2ssg}
	\end{theorem}

	\begin{proof}
		Fix a $2$-SSG with $o>1$ orange agents played on a graph $G$ with $n$ vertices. First, we observe that the social welfare of a social optimum is at most $n-2+\frac{o-1}{o}+\frac{b-1}{b}=n-\frac{1}{o}-\frac{1}{b}$, as there must be at least one orange vertex that is adjacent to at least one blue vertex, thus getting utility at most $\frac{o-1}{o}$, and at least one blue vertex that is adjacent to at least one orange vertex, thus getting utility at most $\frac{b-1}{b}$.
		
		Given a strategy profile $\sp^\prime$, a {\em feasible pair} is a pair of vertices $(u,v)$ such that $u$ and $v$ are occupied by agents of different colors in $\sp^\prime$ and $\{u,v\}\notin E(G)$, i.e., $u$ and $v$ are not adjacent. Now fix a swap equilibrium $\sp$ and consider a maximum cardinality matching $M$ of feasible pairs. Clearly $0\leq |M|\leq o$. Hence, $|M|= o-x$ for some $0\leq x\leq o$. If $x>0$, then, there are exactly $x$ orange and at least $x$ blue leftover vertices of $V$ that do not belong to any feasible pair in $M$. As $M$ has maximum cardinality, each orange leftover vertex has to be adjacent to all leftover blue ones and vice-versa. That is, for each leftover vertex $u$, we have $\delta_u(G)\geq x$. Let~$T$ be a set of pairs of vertices obtained by matching each leftover orange vertex with a leftover blue one. By Lemma~\ref{sum_in_eq}, it holds for each $(u,v)\in M$, $\u_{\sp^{-1}(u)}(\sp)+\u_{\sp^{-1}(v)}(\sp)\geq 1$ and for each $(u,v)\in T$, $\u_{\sp^{-1}(u)}(\sp)+\u_{\sp^{-1}(v)}(\sp)\geq 1-\frac{1}{x}$. Thus, the social welfare of $\sp$ is at least $o-x+x(1-\frac{1}{x})=o-1$.    
	\end{proof}

	\begin{corollary}
		The PoA of $2$-SSGs is constant if $\frac{b}{o}$ is constant.
	\end{corollary}

	\noindent We want to emphasize that for the case where both colors are perfectly balanced, the PoA is constant. As for $n = 2$ the $2$-SSG is trivial and has a PoA $= 1$, we get the following corollary.

	\begin{corollary}
		The PoA of balanced $2$-SSGs is at most $\min\left\{3,\frac{2(n+2)}{n}\right\}$.
	\end{corollary}

	\noindent We will now show that in contrast to the balanced $2$-SSG, the balanced local $k$-SSG has a much higher LPoA. 

	\begin{theorem}\label{thm:PoA_max_min_degree}
		The LPoA of local balanced $2$-SSGs with $o >1$ is between $2n + \frac{8}{n} -8$ and $2n - \frac{8}{n}$.
	\end{theorem}

	\begin{proof}
	Fix a $2$-SSG with $o > 1$ orange agents played on a graph $G$ with $n$ vertices. First, we observe that the social welfare of a social optimum is at most $n-2$+$\frac{o-1}{o}+\frac{n-o-1}{n-o} = n - \frac{n}{o(n-o)}$, as there must be at least one orange vertex that is adjacent to at least one blue vertex. 

	Now fix a local swap equilibrium $\sp$. We will show that the social welfare of $\sp$ is at least $\frac12$. First, assume that there is exactly one vertex $v$ with $\delta_v(G) > 1$. Then, $G$ has to be a star and since $o > 1$ there has to be at least one leaf vertex with an agent $i$ with $\u_i(\sp) = 1$. Therefor, there has to be at least two adjacent vertices $v_1$ and $v_2$ with $\delta_i > 1$ for $i \in \{1,2\}$. By Lemma~\ref{sum_in_eq} we know that if $v_1$ and $v_2$ are occupied by agents of different types then $\u_{\sp^{-1}(v_1)} + \u_{\sp^{-1}(v_2)} \geq \frac12$. Hence, assume that there is no such pair $v_1$ and $v_2$ and assume, without loss of generality, that all adjacent vertex pairs $v_1$ and $v_2$, with $\delta_i > 1$ for $i \in \{1,2\}$, are occupied by orange agents. It follows, since $G$ is connected, that all blue agents only occupy leaf vertices. If the social welfare of $\sp$ is less than $\frac12$, all orange agents have to be surrounded by more blue than orange agents. Since one blue agent is only adjacent to one orange agent this contradicts our requirement of a balanced game. Hence, the PoA is upper bounded by $2\left( n - \frac{n}{o(n-o)} \right)$. With $o = \frac{n}{2}$ this is equal to $2n - \frac{8}{n}$.

	For the lower bound consider the graph $G$ in Figure~\ref{local_poa}.
	\begin{figure}[h]
	\centering
			\begin{subfigure}[c]{0.35\textwidth}
				\includegraphics[width=0.5\textwidth]{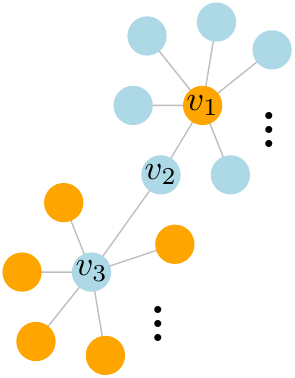}
				\subcaption{Worst equilibrium}
				\label{local_poa_1}
			\end{subfigure}
			\begin{subfigure}[c]{0.35\textwidth}
				\includegraphics[width=0.5\textwidth]{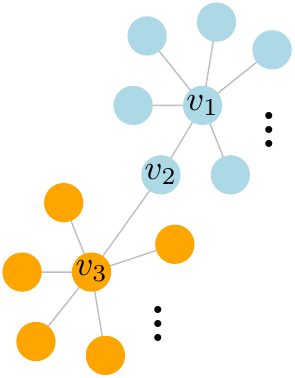}
				\subcaption{Social optimum}
				\label{local_poa_2}
			\end{subfigure}
			\caption{A lower bound for the local balanced $2$-SSG. The agent types are marked orange and blue.}
			\label{local_poa}
		\end{figure}
		$G$ consists of two stars which are connected by a common leaf vertex. Let $v_1$ be the center of the first star, $v_3$ be the center of the second star and $v_2$ be the common vertex. We first prove that the configuration shown in Figure~\ref{local_poa_1} is an equilibrium. Note, that none of the leaf vertices can perform a profitable swap since the agents on $v_1$ and $v_3$, respectively, would receive $\u_{\sp^{-1}(v_1)} = 0$ and $\u_{\sp^{-1}(v_2)} = 0$, respectively. So the only possible swap is between the agents placed on $v_1$ and $v_2$. However the orange agent currently located on $v_1$ would not increase her utility by swapping since she would be surround only by two blue agents placed on $v_1$ and $v_2$ and therefor would receive a utility equals $0$. Hence, no local swap is possible and only the agents placed on $v_2$ and $v_3$ receive positive utility. The social welfare is equal to $\frac12 + \frac{1}{o-1}$ which is for $o = \frac{n}{2}$ equal to $\frac12 + \frac{2}{n-2}$. The social optimum is shown in Figure~\ref{local_poa_2}. This is easy to see, since we meet the trivial upper bound $n-2$+$\frac{o-1}{o}+\frac{n-o-1}{n-o} = n - \frac{n}{o(n-o)}$ which is for $o = \frac{n}{2}$ equal to $n-\frac{4}{n}$. Hence, the PoA is lower bounded by $\frac{2(n-2)^2}{n} = 2n + \frac{8}{n} - 8$.
	\end{proof}

	\noindent If the underlying graph $G$ does not contain leaf vertices, i.e., all vertices have at least degree $2$, we can prove a smaller LPoA. In particular, if the ratio between the maximum and minimum degree of vertices in $G$ is constant, we achieve a constant LPoA. 

	\begin{theorem}
		The LPoA of local $2$-SSGs on a graph $G$ with minimum degree $\delta \geq 2$ and maximum degree $\Delta$ is at most $2\left(1+\frac{\Delta+1}{\delta-1}\right)$.	
		\label{PoA_local}
	\end{theorem}

	\begin{proof}
		Fix  a local swap equilibrium $\sp$ on $G$ with $\delta(G) \geq 2$.
		Let $\rho :=\frac{\delta-1}{2\delta}$ and let $o'$ and $b'$ be the numbers of orange and blue agents that have a utility strictly less than $\rho$, respectively. Clearly, $o-o'$ and $b-b'$ are the numbers of orange and blue agents that have a utility of at least $\rho$, respectively. We first prove that $b-b' \geq \frac{\delta o'}{\Delta}$ as well as that $o-o' \geq \frac{\delta b'}{\Delta}$ and show then how these two inequalities imply the theorem statement. 
	
		We only prove the first inequality, i.e., $b-b' \geq \frac{\delta o'}{\Delta}$ as the proof of the other inequality is similar. Let $i$ and $j$, respectively, be a blue agent and an orange agent that occupy two adjacent vertices in $G$, say $\sigma(i)=u$ and $\sigma(j)=v$, and such that $\u_j(\sp) < \rho$. By Lemma~\ref{sum_in_eq}, we have that $\u_i(\sp)+\u_j(\sp) \geq 1-\frac{1}{\delta}$, from which we derive $\u_i(\sp) > 1-\frac{1}{\delta}-\frac{\delta-1}{2\delta} = \frac{\delta-1}{2\delta}=\rho$.
	
		Let $G'$ be the subgraph of $G$ containing all the non-monochromatic edges, i.e., each edge of $G'$ connects a vertex occupied by an orange agent with a vertex occupied by a blue agent. Clearly, $G'$ is bipartite. Consider the vertex-induced subgraph $H$ of $G'$ in which we have all the~$o'$ orange agents having a utility strictly less than $\rho$ on one side and all the $b-b'$ blue agents having a utility of at least $\rho$ on the other side. Since for each vertex $v$ of $H$ occupied by an orange agent, there are at least $(1-\rho)\delta_v \geq \frac{\delta+1}{2}$ vertices adjacent to $u$ that are occupied by blue agents and each such blue agent have a utility of at least $\rho$, the degree of $v$ in $H$ is at least $\frac{\delta+1}{2}$. Therefore,  
		\begin{equation}\label{eq:lower_bound_size_of_graph}
			|E(H)| \geq \frac{\delta+1}{2} o'.
		\end{equation}
		Furthermore, since each edge of $H$ is incident to a blue agent that has a utility of at least $\rho$, the degree in $H$ of every vertex $u$ that is occupied by a blue agent is at most $(1-\rho)\delta_u \leq \frac{\delta+1}{2\delta}\Delta$. Therefore,
		\begin{equation}\label{eq:upper_bound_size_of_graph}
			|E(H)| \leq \frac{\Delta(\delta+1)}{2\delta} (b-b').
		\end{equation}
		Plugging~(\ref{eq:lower_bound_size_of_graph}) into~(\ref{eq:upper_bound_size_of_graph}) and simplifying gives $b-b' \geq \frac{\delta}{\Delta} o'$. 
	
		Finally, we show how $b-b' \geq \frac{\delta o'}{\Delta}$ and $o-o' \geq \frac{\delta b'}{\Delta}$ imply the theorem statement. 
		The average~utility of all the agents in $H$ is at least
		$
		\frac{\rho(b-b')}{o'+(b-b')} \geq \frac{\rho\frac{\delta}{\Delta}}{1+\frac{\delta}{\Delta}}=\frac{\delta-1}{2(\delta+\Delta)}.
		$
		Similarly, the average utility of the~$b'$ blue agents whose utilities are strictly less than $\rho$ and the $o-o'$ orange agents whose utilities are of at least $\rho$ is also at least $\frac{\delta-1}{2(\delta+\Delta)}$. Therefore, the LPoA is at most $\frac{2(\delta+\Delta)}{\delta-1}=2\left(1+\frac{\Delta+1}{\delta-1}\right)$. 
	\end{proof}

	\noindent We observe that the $LPoA$ on a graph with minimum degree $\delta(G) = 1$ can be unbounded. Consider the star graph with $\Delta$ leaves and let $\sp$ be a strategy profile where the unique orange agent occupies the star center, while all the blue agents occupy the leaves. This is clearly a swap equilibrium of $0$ social welfare. Any configuration in which a blue agent occupies the star center has strictly positive social welfare. 

	However, as the following theorem shows, the LPoA can be upper bounded by a function of~$\Delta$ if we force $n \geq \Delta+2$, i.e., we avoid the pathological star graph of $\Delta+1$ vertices.

	\begin{theorem}\label{thm:LPoA_connected_and_bounded_max_degree}
		For every $\epsilon > 0$, the LPoA of local $2$-SSGs on a graph $G$ with maximum degree $\Delta \leq n-2$ is between $\frac{\Delta(\Delta-1)}{2}-\epsilon$ and $4(\Delta^2-\Delta+1)$.
	\end{theorem}

	\begin{proof}
	We prove the upper bound first. Let $\sp$ be a local swap equilibrium on $G$ with $\Delta \leq n-2$. We claim that, for every agent $i$, with $\delta_{\sigma_i} \geq 2$, there is an agent $j$, with $\sigma_j \in N_{\sigma_i}$ and $\delta_{\sigma_j} \geq 2$, such that $\u_i(\sp)\geq \frac{1}{\Delta}$ or $\u_j(\sp) \geq \frac12$. Indeed, assume that $\u_i(\sp) <\frac{1}{\Delta}$. This implies that $\u_i(\sp)=0$ and, therefore, that every agent occupying a vertex in $N_{\sigma_i}$ is of type different from that of $i$. Therefore, if for the sake of contradiction we assume that $\u_j(\sp) < \frac12$, then $\u_i(\sigma_{ij}) > 0$ and $\u_j(\sigma_{ij}\geq \frac{\delta_{\sigma_i}-1}{\sigma}\geq \frac12$, thus contradicting that $\sp$ is a local swap equilibrium.

	This implies that all the vertices of the graph can be partitioned into two types of sets: 
	\begin{description}
		\item[type-1 set:] It has a size smaller than or equal to $\Delta+1$ and contains a vertex $u$ occupied by an agent that has a utility of at least $\frac{1}{\Delta}$ together with a subset of $N_u$;
		\item[type-2 set:] It has a size smaller than or equal to $1+\Delta(\Delta-1)=\Delta^2-\Delta+1$ and contains a vertex $u$ occupied by an agent that has a utility of at least $\frac12$ together with a subset of $N_u \cup \bigcup_{v \in N_u}N_v$.
	\end{description}
	The average utility of all the agents contained in type-1 sets is at least $\frac{1}{\Delta^2+\Delta}$, while the average utility of all the agents contained in type-2 sets is at least $\frac{1}{2(\Delta^2-\Delta+1)}$. Therefore, as $\Delta \geq 2$, the average utility of an agent is at least
	\[ \min\left\{\frac{1}{\Delta^2+\Delta}, \frac{1}{2(\Delta^2-\Delta+1)}\right\} = \frac{1}{2(\Delta^2-\Delta+1)}.\]
	The upper bound of the LPoA follows.

	For the lower bound of the LPoA, it is enough to consider the instance with $o$ orange agents and $b=(\Delta-2)o$ blue agents, thus, $n=(\Delta-1)o$, consisting of a cycle of length $o$, whose vertices are all occupied by the orange agents, to which we add $\Delta-2$ degree-1 vertices appended to each vertex. Clearly, all the degree-1 vertices are occupied by the blue agents. It is easy to check that the given strategy profile is a local swap equilibrium. Now, observe that each blue agent has a utility of $0$, while each orange agent has a utility of $\frac{2}{\Delta}$. The social welfare of this local swap equilibrium is equal to $\frac{2o}{\Delta}=\frac{2n}{\Delta(\Delta-1)}$. If we assume that $o$ is a multiple of $\Delta-1$, then the social optimum shown in the picture has a social welfare equal to $n-4+4\frac{\Delta-1}{\Delta}=\frac{n-4}{\Delta}$. Therefore, if we choose $n \geq \frac{2(\Delta-1)}{\epsilon}$, we have that the LPoA is lower bounded by
	\[ \left(n-\frac{4}{\Delta}\right)\frac{\Delta(\Delta-1)}{2n} = \frac{\Delta(\Delta-1)}{2}-\frac{2(\Delta-1)}{n} \geq \frac{\Delta(\Delta-1)}{2}-\epsilon.\]
	\end{proof}

	\noindent If we desist from star graphs, the class of trees meet the conditions required by Theorem~\ref{thm:LPoA_connected_and_bounded_max_degree} and we get the following corollary.

	\begin{corollary}\label{cor:poa_trees}
		For every $\epsilon > 0$, the LPoA of local $2$-SSGs on a tree graph $G$ with maximum degree $\Delta \leq n-2$ is at least $\frac{\Delta(\Delta-1)}{2}-\epsilon$.
	\end{corollary}

	\begin{proof}
	Consider the lower bound construction given in Theorem~\ref{thm:LPoA_connected_and_bounded_max_degree} in which we remove one edge from the cycle. There is a threshold value $f(\Delta,\epsilon)$ such that for every $n \geq f(\Delta,\epsilon)$, the LPoA is at least $\frac{\Delta(\Delta-1)}{2}-\epsilon$.
	\end{proof}

	\subsection{Regular Graphs}

	In this section we provide upper and lower bounds to the LPoA for regular graphs, i.e., for graphs where all vertices have the same degree. The key is the following technical lemma.

	\begin{lemma}\label{lm:average_utility_Delta_window}
		Let $\sp$ be a local swap equilibrium, and let $\Delta=2\alpha+\beta$, with $\alpha \in {\mathbb N}$ and $\beta \in \{0,1\}$. Let $X\subseteq V$ be a subset of vertices such that $\delta_v=\Delta$ for every $v \in N_X:=\bigcup_{x \in X}N_x$. Finally, let $Z \subseteq N_X$ be the set of vertices occupied by the agents that have a utility strictly larger than $\rho :=\frac{\alpha}{2\alpha+1}$. Then, the average utility of the agents that occupy the vertices in $X \cup Z$ is at least~$\rho$.
	\end{lemma}

	\begin{proof}
		Let $X_o\subseteq X$ (respectively, $X_b \subseteq X$) be the set of vertices occupied by the orange (respectively, blue) agents that have a utility strictly less than $\rho$. Similarly, let $Z_o \subseteq N_X$ (respectively, $Z_b \subseteq N_X$) be the set of vertices occupied by the orange (respectively, blue) agents that have a utility strictly larger than $\rho$. We show that the average utility of the agents that occupy the vertices $X_o \cup Z_b$ (respectively, $X_b \cup Z_o$) is at least $\rho$. Notice that this immediately implies the theorem statement.
	
		In the rest of the proof, without loss of generality, we prove that the average utility of the agents that occupy the vertices in $X_o \cup Z_b$ is at least $\rho$. First of all, we observe that the utility of each agent in $N_X$ is in the set $\{\frac{\ell}{\Delta} \mid \ell=0,\dots,\Delta\}$. Let $o_\ell$ be the numbers of orange agents that occupy the vertices of $X$ and whose utilities are equal to $\frac{\ell}{\Delta}$. Similarly, let $b_\ell$ be the numbers of orange agents that occupy the vertices of $N_X$ and whose utilities are equal to $\frac{\ell}{\Delta}$. Since we are interested to the orange agents occupying the vertices of $X_o$, we consider the values $o_\ell$ such that $\frac{\ell}{\Delta} < \rho$, or, equivalently, $\ell \leq \alpha-1$. Similarly, since we are interested to the blue agents occupying the vertices of $Z_b$, we consider the values $b_{\Delta-\ell-1}$ such that $\frac{\Delta-\ell-1}{\Delta} > \rho$, or, equivalently, $\ell \leq \alpha-1$. We prove that, for every $0 \leq h \leq \alpha-1$,
		\begin{equation}\label{eq:relation_between_b_and_o}
			\sum_{\ell=0}^{h} (\ell+1)b_{\Delta-\ell-1} \geq \sum_{\ell=0}^{h}(\Delta-\ell)o_\ell.
		\end{equation}
	We observe that if any orange agent $i$ that occupies a vertex $v \in X_o$ has a utility of $\frac{\ell}{\Delta}$, where $0 \leq \ell \leq \alpha-1$, then, since we are in a local swap equilibrium, any of the $\Delta-\ell$ blue agents that occupy the vertices in $N_v$ has a utility of at least $\frac{\Delta-\ell-1}{\Delta}> \rho$ by Lemma~\ref{sum_in_eq}. 
	This implies that $v$ has at least $\Delta-\ell-1$ vertices in its neighborhood that are occupied by  blue agents, and therefore, at most $\ell+1$ vertices in its neighborhood that are occupied by orange agents. Let $G'$ be the (bipartite) subgraph of $G$ containing all the non-monochromatic edges. Consider the subgraph $H$ of $G'$ that is induced by the vertices in $X_h\subseteq X_o$ that are occupied by agents having a utility of at most $\frac{h}{\Delta}$ and the agents in $Z_h \subseteq Z_b$ having a utility of at least $\frac{\Delta-h-1}{\Delta}$. By construction, the degree of a vertex of $X_h$ occupied by an agent of utility equal to $\frac{\ell}{\Delta}$, with $\ell \leq h$, is equal to $\Delta-\ell$. Therefore, if $\delta_v(H)$ denotes the degree of $v$ in $H$, we have that 
	\begin{equation}\label{eq:counting_edges}
		|E(H)|=\sum_{v \in X_h}\delta_v(H)=\sum_{\ell=0}^{h}(\Delta-\ell)o_\ell.
	\end{equation} 
		Since the degree in $H$ of each vertex  in $Z_h$ that is occupied by a blue agent whose utility is equal to $\frac{\Delta-\ell-1}{\Delta}$, with $\ell \leq h$, is upper bounded by $\ell+1$, we have that 
	\begin{equation}\label{eq:upper_bound_couting_edges}
	|E(H)| \leq \sum_{v \in Z_h}\delta_v(H) = \sum_{\ell=0}^h (\ell+1)b_{\Delta-\ell-1}.
	\end{equation}
	Combining~(\ref{eq:counting_edges}) with~(\ref{eq:upper_bound_couting_edges}) gives~(\ref{eq:relation_between_b_and_o}).
	We are now able to compute the average utility with respect to the agents occupying the vertices in $X_o \cup Z_b$. The average utility of such agents equals
	\[\u_{\text{avg}}:=\frac{\sum_{\ell=0}^{\alpha-1}\left(\frac{\Delta-\ell-1}{\Delta}b_{\Delta-\ell-1}\right)+\sum_{\ell=0}^{\alpha-1}\left(\frac{\ell}{\Delta}o_\ell\right)}{\sum_{\ell=0}^{\alpha-1}b_{\Delta-\ell-1}+\sum_{\ell=0}^{\alpha-1}o_\ell}.	\]
	Now we prove that $\u_{\text{avg}}\geq \rho$. We assume that the values of all the $o_\ell$'s are fixed and that there is at least one $o_\ell$, with $0 \leq \ell \leq \alpha-1$, that is strictly greater than $0$. Since $\frac{\ell}{\Delta} < \rho$, while $\frac{\Delta-\ell-1}{\Delta} > \rho$, we have that $\u_{\text{avg}}$ is minimized when the values we can assign to the $b_{\Delta-\ell-1}$'s -- that must satisfy~(\ref{eq:relation_between_b_and_o}) for every $0 \leq h \leq \alpha-1$ -- are somehow minimized. 
	
	Since, for every $\ell < \ell'$ and every $0 < \epsilon < b_{\Delta-\ell'-1}$, $$\frac{\Delta-\ell-1}{\Delta} > \frac{\Delta-\ell'-1}{\Delta}$$ as well as $$(\ell'+1)(b_{\Delta-\ell-1}+\epsilon)+(\ell+1)(b_{\Delta-\ell'-1}-\epsilon) > (\ell'+1)b_{\Delta-\ell-1}+(\ell+1)b_{\Delta-\ell'-1},$$ we have that $\u_{\text{avg}}$ is minimized exactly when $b_{\Delta-\ell-1}=\frac{\Delta-\ell}{\ell+1}o_\ell$.\footnote{We are relaxing the constraint that $b_{\Delta-\ell-1}$ must be an integer.} Therefore, if we denote by $\Psi=\{\ell \mid 0 \leq \ell \leq \alpha-1 \wedge o_\ell > 0\}$, we have that
	\begin{align*}
		\u_{\text{avg}}	& \geq \frac{\sum_{\ell \in \Psi}\left(\frac{(\Delta-\ell-1)(\Delta-\ell)}{\Delta(\ell+1)}o_\ell\right)+\sum_{\ell \in \Psi}\left(\frac{\ell}{\Delta}o_\ell\right)}{\sum_{\ell \in \Psi}\left(\frac{\Delta-\ell}{\ell+1}o_\ell\right)+\sum_{\ell \in \Psi}o_\ell} \\
		& = \frac{\sum_{\ell \in \Psi} \frac{2\ell^2-2(\Delta-1)\ell+\Delta(\Delta-1)}{\Delta(\ell+1)}}{\sum_{\ell \in \Psi}{\frac{\Delta+1}{\ell+1}}} 
		\geq \min_{\ell \in \Psi}\frac{2\ell^2-2(\Delta-1)\ell+\Delta(\Delta-1)}{\Delta(\Delta+1)}.
	\end{align*}
	We complete the proof by showing that
	\begin{equation}\label{eq:parabola}
		\min_{\ell \in \Psi} \frac{2\ell^2-2(\Delta-1)\ell+\Delta(\Delta-1)}{\Delta(\Delta+1)} \geq \rho.
	\end{equation}
	The numerator of the left-hand side of~(\ref{eq:parabola}) is a parabola with respect to the variable $\ell$ and is therefore minimized when $\ell$ is chosen as closest as possible to the value $\frac{\Delta-1}{2}$. 
	
	As $\left\lfloor\frac{\Delta-1}{2}\right\rfloor \geq \alpha-1$ and $\ell \leq \alpha-1$, if follows that the value of $\ell$ that minimizes~(\ref{eq:parabola}) is $\ell=\alpha-1$. Therefore, $$\frac{2(\alpha-1)^2-2(2\alpha-1)(\alpha-1)+2\alpha(2\alpha-1)}{2\alpha(2\alpha+1)}=\rho.$$ Hence, $\u_{\text{avg}} \geq \rho$.
	\end{proof}

	\begin{corollary}
		The LPoA of local $2$-SSG on a regular graph $G$ with $\Delta(G) = 2\alpha+\beta$, with $\alpha \geq 1$ and $\beta \in \{0,1\}$ is at most $2+\frac{1}{\alpha}$.
		\label{PoA_reg}
	\end{corollary}

	\begin{proof}
		The corollary follows from Lemma \ref{lm:average_utility_Delta_window} by $X=V$.
	\end{proof}

	\noindent The matching lower bound is provided in the following.

	\begin{theorem}
		The LPoA of local $2$-SSG on a regular graph $G$ with $\Delta(G) = 2\alpha+\beta$, with $\alpha \geq 1$ and $\beta \in \{0,1\}$ is equal to $2+\frac{1}{\alpha}$.
		\label{PoA_reg_1}
	\end{theorem}

	\begin{proof}
		For a fixed degree $\Delta\geq 3$, we define the $\Delta$-regular graph $G(\Delta):=G$ as follows. There are $q:=t(\Delta+1)$ gadgets $G^1,\ldots,G^q$. For each $i\in [q]$, gadget $G^i$ is obtained from a complete graph of $\Delta+1$ vertices, denoted as $v_0^1,\ldots,v_\Delta^i$, by removing edge $\{v_0^i,v_\Delta^i\}$. Observe that, by construction, for any $i\in [q]$, each vertex $v_j^i$, with $1\leq j\leq\Delta-1$, has degree $\Delta$, while vertices $v_0^i$ and $v_\Delta^i$ have degree $\Delta-1$. We obtain $G$ by connecting the $q$ gadgets through edges $\{v_\Delta^i,v_0^{i+1}\}$ for each $i\in [q-1]$ and edge $\{v_\Delta^q,v_0^1\}$. Call these edges {\em extra-gadget} edges. Thus, $G$ is connected and $\Delta$-regular. Consider now the local $2$-SSG played on $G$ in which there are $\lceil\frac{\Delta+1}{2}\rceil q$ blue agents and $\lfloor\frac{\Delta+1}{2}\rfloor q$ orange ones.

		On the one hand, the social optimum is at least $n-\frac{4}{\Delta}=q(\Delta+1)-4 \Delta$, as in the strategy profile in which all vertices of the first $\lceil\frac{\Delta+1}{2}\rceil t$ gadgets are colored blue and all vertices of the remaining $\lfloor\frac{\Delta+1}{2}\rfloor t$ gadgets are colored orange there are $n-4$ vertices getting utility $1$ and $4$ vertices getting utility $\frac{\Delta-1}{\Delta}$.

		On the other hand, the strategy profile $\sp$ in which the first $\lceil\frac{\Delta+1}{2}\rceil$ vertices of each gadget are colored blue and the remaining ones are colored orange is a swap equilibrium. In fact, as extra-gadget edges connect vertices of different colors, every blue vertex is adjacent to $\lceil\frac{\Delta+1}{2}\rceil-1$ blue ones, while every orange vertex is adjacent to $\lceil\frac{\Delta+1}{2}\rceil$ blue ones. If a blue vertex swaps with an adjacent orange one, it ends up being adjacent to $\lceil\frac{\Delta+1}{2}\rceil-1$ blue vertices. Thus, no profitable swap exists in $\sp$.

		As the social welfare of $\sp$ is
		\begin{eqnarray*}
			& & \frac{q}{\Delta}\left(\left\lceil\frac{\Delta+1}{2}\right\rceil\left(\left\lceil\frac{\Delta+1}{2}\right\rceil-1\right)+\left\lfloor\frac{\Delta+1}{2}\right\rfloor\left(\left\lfloor\frac{\Delta+1}{2}\right\rfloor-1\right)\right)\\
			& = & \left\{
			\begin{array}{ll}\vspace*{0.5em}  \frac{q(\Delta^2-1)}{2\Delta} & \textrm{ if $q$ is odd,} \\ \frac{q\Delta}{2} & \textrm{ if $q$ is even,}
			\end{array}\right.
		\end{eqnarray*}
		we get that the LPoA of the game is lower bounded by $\frac{2\Delta(q(\Delta+1)-4\Delta)}{q(\Delta^2-1)}$ when $\Delta$ is odd and by $\frac{2(q(\Delta+1)-4\Delta)}{q\Delta}$ when $\Delta$ is even. By letting $q$ going to infinity, we get $\frac{2\Delta}{\Delta-1}$ and $\frac{2(\Delta+1)}{\Delta}$,  respectively. By using $\Delta=2\alpha+1$ in the first case, and $\Delta=2\alpha$ in the second one, we finally obtain the lower bound of $2+\frac{1}{\alpha}$.
	\end{proof}

	\subsection{Paths and Cycles}

	In this section we provide upper and lower bounds for the (L)PoA of paths and cycles. We first provide a full characterization of the PoA for cycles.

	\begin{theorem}\label{thm:cycle_global}
		The PoA of $2$-SSGs played on cycles with $n \geq 3$ vertices and $o=2\alpha+\beta$ orange agents, where $\alpha \in {\mathbb N}$, $\beta \in \{0,1\}$, and $b \geq o$, is equal to
		\begin{equation*}
			PoA = 
			\begin{cases}
			1					& \text{if $o=1$;}\\
			\frac{n-2}{b+\beta}	& \text{otherwise.}\\
			\end{cases}	
		\end{equation*}
	\end{theorem}

	\begin{proof}
	The social welfare of the social optimum is clearly equal to $n-2$ and is attained when the cycle contains one path whose vertices are all occupied by the $b$ blue agents and another path whose vertices are all occupied by the $o$ orange agents. Now we prove matching upper and lower bounds for all the  cases.

	When $o=1$ we clearly have that any strategy profile is a swap equilibrium because the unique orange agent always has a utility of $0$, the two blue agents that occupy the vertices adjacent to the vertex occupied by the orange agent have a utility of $\frac12$ each, and the remaining $b-2$ blue agents all have a utility of $1$. Therefore, the social welfare is equal to $n-2$, and the claim follows.

	Let $\sp$ be a swap equilibrium. Let $\ell$ be the number of maximal vertex-induced (sub)paths whose vertices are occupied by orange agents only. Clearly, $\ell$ is also the number of maximal vertex-induced (sub)paths whose vertices are occupied by blue agents only. We claim that $\ell \leq \alpha$ by showing that every agent has a strictly positive utility in $\sp$ (i.e., each of the $2\ell$ maximal paths formed by monochromatic edges contains $2$ or more vertices). Indeed, for the sake of contradiction, assume without loss of generality that there is an orange agent $i$ such that $\u_i(\sp)=0$. This implies that there must be a blue agent $j$ that occupies a vertex $v$ such that $v$ is not adjacent to the vertex occupied by $i$ and $v$ is adjacent to a vertex occupied by an orange agent $i'\neq i$. As a consequence, $\u_j(\sp)\leq \frac12$. In this case, swapping $i$ with $j$ would be an improving move since $u_i(\sp_{ij}) > 0 = u_i(\sp)$ and $1=u_j(\sp_{ij})> \frac{1}{2} \geq u_j(\sp)$, thus contradicting the fact that $\sp$ is a swap equilibrium.

	As a consequence the utility of $2\ell$ orange agents is equal to $\frac12$, while the utility of the other $o-2\ell=n-b-2\ell$ orange agents is equal to $1$; similarly, the utility of $2\ell$ blue agents is equal to $\frac12$, while the utility of the other $b-2\ell$ blue agents is equal to $1$. Therefore, the social cost is at least 
	\[ \frac{1}{2}(2\ell+2\ell)+(n-b-2\ell)+(b-2\ell)=n-2\ell \geq n-2\alpha=b+\beta. \]
	The upper bound to the PoA follows.

	For the matching lower bound, it is enough to consider the strategy profile in which $\ell=\alpha$, i.e., there are $\alpha-1$ maximal vertex-induced paths occupied by orange (respectively, blue) agents only of length $2$ each, and one maximal vertex-induced path occupied by orange (respectively, blue) agents only of length $2+\beta$ (respectively, $b-2\alpha+2$). In this case, the social welfare is exactly equal to 
	\[ \frac{1}{2}2\alpha+\beta+\frac{1}{2}\alpha+(b-2\alpha)=b+\beta.\]
	\end{proof}

	\noindent The following theorem provides almost tight upper bounds to the LPoA for cycles.

	\begin{theorem}\label{thm:cycle_local}
		The LPoA of local $2$-SSGs played on cycles with $n=3\alpha+\beta$ vertices and $b$ blue agents, where $\alpha \in {\mathbb N}$, $\beta \in \{0,1,2\}$, and $b \geq o$, is upper bounded by
		\begin{equation*}
			LPoA \leq 
			\begin{cases}
			1							& \text{if $o=1$;}\\
			\frac{n-2}{b-o}				& \text{if $o \geq 2$ and $b \geq 2o$;}\\
			\frac{n-2}{\alpha + \beta}	& \text{otherwise (i.e., $o \geq 2$ and $b < 2o$).}\\
			\end{cases}	
		\end{equation*}
		The upper bounds are tight when (i) $o=1$ and (ii) $o\geq 2$ and $b \geq 2o$. 
	\end{theorem}

	\begin{proof}
		The social welfare of the social optimum is equal to $n-2$. Now, we prove matching upper and lower bounds for all cases.
	
		When $o=1$, any configuration is a (local) swap equilibrium; therefore the social welfare is equal to $n-2$ and the claim follows.
	
		Now, we consider the case in which $o \geq 2$. Let $o_h$ and $b_h$ be the numbers of orange and blue agents having a utility equal to $h \in \{0,\frac{1}{2},1\}$, respectively. Every configuration can be decomposed into maximal vertex-induced paths whose vertices are all occupied by agents of the same type. Furthermore, if $\ell$ is the overall number of these maximal vertex-induced paths whose vertices are all occupied by  orange agents, then $\ell$ is also the overall number of maximal vertex-induced paths whose vertices are all occupied by blue agents. This implies that that $o_{\frac{1}{2}}=2(\ell - o_0)$ and $b_{\frac{1}{2}}=2(\ell-b_0)$. Therefore, $o=o_0+o_{\frac{1}{2}}+o_1=2\ell-o_0+o_1$ and $b=b_0+b_{\frac{1}{2}}+b_1=2\ell-b_0+b_1$, i.e., $o_1=o-2\ell+o_0$ and $b_1=b-2\ell+b_0$. As a consequence, using the fact that $b+o=n$, the social welfare is equal to
		$\sum_{h \in \{0,\frac{1}{2},1\}}ho_h+\sum_{h \in \{0,\frac{1}{2},1\}}hb_h  =\ell-o_0+o-2\ell+o_0 + \ell-b_0+b-2\ell+b_0 = n-2\ell.	$
		We observe that each orange agent of utility~$0$ occupies a vertex that is adjacent to two vertices occupied by blue agents having a utility of~$\frac{1}{2}$ each. As a consequence, $b_{\frac{1}{2}}=2(\ell - b_0) \geq 2o_0$, or, equivalently, $\ell \geq b_0+o_0$. Therefore, the social welfare is minimized exactly when~$\ell$ is maximized, as shown by the following ILP (where the second and third constraints are of the form $o_0+o_{\frac{1}{2}}\leq o$ and $b_0+b_{\frac{1}{2}} \leq b$, respectively):
		\begin{alignat*}{2}
			& \text{maximize} 	& 		& \ell \\
			& \text{subject to}	& \quad & b_0+o_0 \leq \ell\\
			&						& \quad & 2\ell - o_0 \leq o\\
			&						& \quad & 2\ell - b_0 \leq b\\
			&						& \quad & \ell, b_0,o_0 \in \mathbb{N}.
		\end{alignat*}
		Combining the first 3 inequalities we obtain $2\ell+2\ell \leq o+o_0+b+b_0 \leq n + \ell$, from which we derive $\ell \leq \lfloor \frac{n}{3} \rfloor = \alpha$. Furthermore, since $o_0 \leq \ell$ we have that $\ell \leq 2\ell - o_0 \leq o$. Therefore, the value of an optimum solution is upper bounded by $\ell=\min\{o, \alpha\}$. If $b \geq 2o$, then setting~$\ell$, $o_0=o$ and all other variables to $0$ is an optimal solution. If $b < 2o$, then setting $\ell=\alpha$, $o_0=2\alpha-o$, and $b_0=2\alpha - b$ is an optimal solution. The upper bound to the LPoA follows.
	
		For the matching lower bound when $o\geq 2$ and $b \geq 2o$, it is enough to consider the strategy profile in which $\ell=o$, i.e., each orange agent occupies a vertex that is adjacent to vertices occupied by blue agents only. As a consequence, the $o$ orange agents have a utility of $0$, the $2o$ blue agents have a utility of $\frac{1}{2}$ each, while the remaining $b-2o=n-3o \geq 0$ blue agents have a utility of $1$ each. The social welfare in this case is exactly equal to $\frac{1}{2}o+n-3o=n-2o=b-o$. 
	\end{proof}

	\noindent We now prove similar results for paths.

	\begin{theorem}\label{thm:path_global}
		The PoA of $2$-SSGs played on paths with $n \geq 3$ vertices and $o=2\alpha+\beta$ orange agents, where $\alpha \in {\mathbb N}$, $\beta \in \{0,1\}$, and $b \geq o$, is equal to
		\begin{equation*}
			PoA = 
			\begin{cases}
			+\infty					& \text{if $n=3$;}\\
			\frac{2n-2}{2n-5}		& \text{if $n > 3$ and $o=1$;}\\
			\frac{n-1}{b+1+\beta}	& \text{if $n > 3$, $o \geq 2$ and $b \leq 2\alpha+1$;}\\
			\frac{n-1}{b+\beta}		& \text{otherwise (i.e., $o \geq 2$ and $b \geq 2\alpha+2$).}
			\end{cases}	
		\end{equation*}
	\end{theorem}

	\begin{proof}
		For $n\geq 4$, the social welfare of the social optimum is clearly equal to $n-1$ and is attained when the path contains a subpath whose vertices are all occupied by the $b$ blue agents and one subpath whose vertices are all occupied by the $o$ orange agents. For $n=3$, the social welfare of the social optimum is clearly equal to $\frac32$ and is attained when the orange agent occupies one endvertex of the path. Now we prove matching upper and lower bounds for all the cases.

		When $o=1$, we clearly have that any strategy profile is a swap equilibrium. The strategy profile with minimum social welfare is when the orange agent occupies a vertex that is adjacent to an endvertex of the path. In this case, the blue agent that occupies such an endvertex has a utility of $0$, the orange agent has a utility of $0$, the other blue agent that is adjacent to the vertex occupied by the orange agent has a utility of $0$, if $n=3$, and of $\frac12$, if $n\geq 4$, while all the other blue agents (if any) have a utility of $1$ each. Therefore, for $n=3$ the social welfare is $0$, while for $n\geq 4$, the social welfare is equal to $n-\frac52$, and the claim follows.

		Therefore, we are only left to prove the bounds to the PoA when $n > 3$ and $o \geq 2$. Let~$\sp$ be a swap equilibrium. We first show that every agent has a strictly positive utility in $\sp$. Indeed, for the sake of contradiction, assume without loss of generality that there is an orange agent $i$ such that $\u_i(\sp)=0$. This implies that there must be a blue agent $j$ that occupies a vertex~$v$ such that~$v$ is not adjacent to the vertex occupied by $i$ and $v$ is adjacent to a vertex occupied by an orange agent $i'\neq i$. As a consequence, $\u_j(\sp)\leq \frac12$. In this case, swapping $i$ with $j$ would be an improving move since $u_i(\sp_{ij}) > 0 = u_i(\sp)$ and $1=u_j(\sp_{ij})> \frac{1}{2} \geq u_j(\sp)$, thus contradicting the fact that~$\sp$ is a local swap equilibrium.

		Let $\ell$ be the number of maximal vertex-induced (sub)paths whose vertices are all occupied by the orange agents. Since every orange agents has strictly positive utility, it follows that $\ell \leq \alpha$. Let $x$ and $y$ be the number of orange and blue agents that occupy the endvertices of the path, respectively. Clearly $x+y=2$. Let $\ell'$ be the number of maximal vertex-induced (sub)paths whose vertices are all occupied by the blue agents. We have that $\ell' \leq \ell+1$. Furthermore, the utility of $2\ell-x$ orange agents is $\frac12$ while the utility of the other $o-2\ell+x$ orange agents is $1$; similarly, the utility of $2\ell'-y$ blue agents is $\frac12$, while the utility of the other $b-2\ell'+y$ blue agents is $1$. Therefore, the social welfare is at least 
		\[\frac{1}{2}(2\ell-x+2\ell'-y)+(o-2\ell+x)+(b-2\ell'+y)=n+\frac{1}{2}(x+y)-\ell-\ell' \geq n+1-\ell-\ell'.\]
		If $b \leq 2 \alpha+1$, then $\ell'\leq \alpha$ and therefore $n+1-\ell-\ell' \geq n+1-2\alpha=b+1+\beta$. 

		If $b \geq 2\alpha+2$, then $\ell' \leq \ell+1$ and therefore $n+1-\ell-\ell' \geq n-2\alpha=b+\beta$.  

		For the matching lower bound, consider the strategy profile that induces $\ell=\alpha$ maximal vertex-induced paths occupied by orange agents only and $\ell'$ maximal vertex-induced paths that are occupied by blue agents only, where $\ell'=\alpha$ if $b \leq 2\alpha+1$ and to $\ell'+1$ otherwise. In this case, the social welfare is exactly equal to $b+1+\beta$ if $b \leq 2\alpha+1$ and $b+\beta$, otherwise.
	\end{proof}

	\begin{theorem}\label{thm:path_local}
		The LPoA of local $2$-SSGs played on paths with $n=3\alpha+\beta$ vertices and $b$ blue agents, where $\alpha \in {\mathbb N}$, $\beta \in \{0,1,2\}$, and $b \geq o$, is upper bounded by
		\begin{equation*}
			LPoA \leq 
			\begin{cases}
			+\infty				& \text{if $n=3$;}\\
			\frac{2n-2}{2n-5}	& \text{if $n > 3$ and $o=1$;}\\
			\frac{n-1}{b-o-1}	& \text{if $n > 3$, $o \geq 2$, $b \geq 2o$;}\\
			\frac{n-1}{\alpha}	& \text{otherwise (i.e., $n > 3$, $o \geq 2$ and $b < 2o$).}
			\end{cases}	
		\end{equation*}
		The upper bounds are tight when (i) $n=3$, (ii) $n > 3$ and $o=1$, and (iii) $n > 3$, $o \geq 2$, $b \geq 2o$. 
	\end{theorem}

\begin{proof}
As shown in Theorem \ref{thm:path_global}, the social welfare of the social optimum is equal to $n-1$. Furthermore, both the upper and lower bounds to the PoA proved in Theorem \ref{thm:path_global} for $n=3$ as well as for $n>3$ and $o=1$ also hold for the LPoA. Therefore, in the rest of the proof we assume that $n\geq 4$ and $o \geq 2$. 

Let $o_r$ and $b_r$ be the numbers of orange and blue agents having a utility equal to $r \in \{0,\frac12,1\}$, respectively. Let $\ell$ (respectively, $\ell'$) be the overall number of maximal vertex-induced paths whose vertices are all occupied by orange (respectively, blue) agents. We observe that $\ell-1 \leq \ell' \leq \ell+1$. Let $x_r$ (respectively, $y_r$) be the number of orange (respectively, blue) agents that occupy the endvertices of the path and whose utility is equal to $r\in\{0,1\}$.  We have that $x_0+x_1+y_0+y_1 = 2$. Furthermore, we have that $o_{\frac12}=2(\ell - o_0)-x_1$ and $b_{\frac12}=2(\ell'-b_0)-y_1$. Therefore, $$o=o_0+o_{\frac12}+o_1=2\ell-o_0-x_1+o_1$$ and $$b=b_0+b_{\frac12}+b_1=2\ell'-b_0-y_1+b_1,$$ i.e., $o_1=o-2\ell+o_0+x_1$ and $b_1=b-2\ell'+b_0+y_1$. As a consequence, the social welfare is equal to
\begin{align*}
\sum_{h \in \{0,\frac12,1\}}hr_h+\sum_{h \in \{0,\frac12,1\}}hb_h & =\ell-o_0-\frac{1}{2}x_1+o-2\ell+o_0+x_1 + \ell'-b_0-\frac{1}{2}y_1+b-2\ell'+b_0+y_1 \\
& = n-\ell-\ell'+\frac{1}{2}x_1+\frac{1}{2}y_1.
\end{align*}
Now, observe that each orange (respectively, blue) agent that has a utility of $0$ and occupies neither an endvertex of the path nor its adjacent vertex is adjacent to two blue (respectively, orange) agents of utility equal to $\frac12$ each.
Therefore $b_{\frac12} = 2(\ell'-b_0)-y_1 \geq 2(o_0-x_0)$ as well as $o_{\frac12}=2(\ell - o_0)-x_1 \geq 2(b_0-y_0)$, or, equivalently, $\ell' \geq b_0+o_0-x_0+\frac{1}{2}y_1$ as well as $\ell \geq b_0+o_0-y_0+\frac{1}{2}x_1$. Therefore, to minimize the social welfare we need to solve the following ILP.

\begin{alignat*}{2}
  & \text{maximize} 	& 		& \ell+\ell'-\frac{1}{2}x_1-\frac{1}{2}y_1 \\
  & \text{subject to}	& \quad & b_0+o_0-y_0+\frac{1}{2}x_1 \leq \ell\\
  &						& \quad & b_0+o_0-x_0+\frac{1}{2}y_1 \leq \ell'\\
  &						& \quad & 2\ell -o_0-x_1 \leq o\\
  &						& \quad & 2\ell'-b_0-y_1 \leq b\\
  &						& \quad & x_0+x_1+y_0+y_1 = 2\\
  &						& \quad & x_0 \leq o_0\\
  &						& \quad & y_0 \leq b_0\\
  &						& \quad & \ell' \leq \ell+1\\
  &						& \quad & \ell \leq \ell'+1\\
  &						& \quad & \ell,\ell',x_0,x_1,y_0,y_1,b_0,o_0 \in \mathbb{N}.
\end{alignat*}
Combining the first 4 inequalities of the ILP we obtain 
\[
2\ell+2\ell \leq o+o_0+x_1+b+b_0+y_1 \leq n + \frac{1}{2}\ell + \frac{1}{2}y_0 - \frac{3}{4}y_1 + \frac{1}{2}\ell' + \frac{1}{2}x_0 - \frac{3}{4}x_1,
\] from which we derive 
\[
\ell+\ell'-\frac{1}{2}(x_1+y_1) \leq \frac{2}{3}n+\frac{1}{3}(x_0+y_0)=2\alpha+\frac{2}{3}\beta+\frac{2}{3}-\frac{1}{3}(x_1+y_1).
\]
By considering the constraints $0 \leq x_1+y_1 \leq 2$ and the fact that $x_1,y_1,\ell$ and $\ell'$ are all non negative integers, it turns out that the above inequality is maximized exactly when $x_1+y_1=0$ or, equivalently, $x_1$, $y_1=0$, and therefore, $\ell+\ell' \leq \left \lfloor 2\alpha+\frac{2}{3}\beta+\frac{2}{3}\right \rfloor = 2\alpha+\beta$.
Furthermore, combining the seventh inequality of the ILP with the first one, we obtain $o_0 \leq \ell$ and therefore, using the third inequality of the ILP, we obtain that $\ell \leq o$. Since the eighth inequality implies that $\ell' \leq \ell+1 \leq o+1$, we have that the value $\ell+\ell' \leq 2o+1$. As a consequence the value of an optimum solution is upper bounded by
\[
\min\left\{2o+1, 2\alpha+\beta\right\}.
\]
We now divide the proof into two cases:
\begin{description}
\item[Case 1:] $b\geq 2o$. Setting $\ell$, $o_0=o$, $\ell'=o+1$, $y_0,b_0=2$, and all the remaining variables to $0$ gives an optimum solution for the ILP and the corresponding value of the objective function matches the upper bound of $2o+1$. Therefore, the social welfare is at least $n-2o-1=b-o-1$ and the upper bound to the LPoA follows. Furthermore, this upper bound is tight. Indeed, consider the strategy profile in which each orange agent occupies a vertex that is adjacent to two vertices occupied by blue agents only and two orange agents occupy the second and last but one vertex of the path (i.e., the two vertices adjacent to the path endvertices). Observe that there are exactly $2(o-1)$ blue agents having a utility equal to $\frac12$ and $2$ agents having a utility of $0$ (thus $b-2(o-1)-2$ agents having a utility of $1$). The social welfare of this configuration is equal to $\frac{1}{2}2(o-1)+(b-2(o-1)-2)=o-1+b-2o=b-o-1$.
\item[Case 2:] $b < 2o$. The optimum value of the ILP is upper bounded by $2\alpha+\beta$. Hence, the social welfare is at least $n-2\alpha-\beta=\alpha$, and the upper bound to the LPoA follows. 
\end{description}
\end{proof}

\subsection{Grids}

We now turn our focus to grid graphs with $4$- and $8$-neighbors. Remember that grids are formed by a two-dimensional lattice. Hence, we can partition the vertices of an $l \times h$ grid $G$ into three sets: {\em corner vertices}, {\em border vertices} and {\em middle vertices}, denoted, respectively, as $C(G)$, $B(G)$, and $M(G)$. We have $C(G)=\{v_{i,j}:i\in\{1,n\}\textrm{ and }j\in\{1,m\}\}$, $B(G)=\{v_{i,j}:i\in\{1,n\}\textrm{ or }j\in\{1,m\}\}\setminus C(G)$ and $M(G)=V(G)\setminus (C(G)\cup B(G))$.

First, we focus on $2$-SSGs in $4$-grids and start by characterizing the PoA for the case in which one type has a unique representative.

\begin{theorem}\label{if1}
	The PoA of $2$-SSGs played on a $4$-grid in which one type has cardinality $1$ is equal to $\frac{25}{22}$.
\end{theorem}

\begin{proof}
	Assume, without loss of generality, that orange is the type with a unique representative. For this game, any strategy profile  $\sp$ is an equilibrium, since in any profile, the orange vertex~$o$ gets utility zero, the vertices not adjacent to $o$ get utility $1$, while all vertices adjacent to $o$ get less than $1$. Call these last vertices the {\em penalized vertices}. Thus, the PoA is maximized by comparing the social welfare of the strategy profile minimizing the overall loss of the penalized vertices with the one of the strategy profile maximizing it. It is easy to see that the overall loss of the penalized vertices is minimized when $o$ is a corner vertices, while it is maximized when $o$ is a border one in a $4$-grid with $l=2$ and $h=3$. Comparing the two social welfares gives the claimed bound.
\end{proof}

\noindent Clearly, if one type has only one representative, this agent will receive utility zero. However, this is not possible in equilibrium assignments when there are at least two agents of each type.

\begin{lemma}\label{no-zero}
	In any equilibrium for a $2$-SSG played on a $4$-grid in which both types have cardinality larger than $1$ all agents get positive utility.
\end{lemma}

\begin{proof}
	Fix an equilibrium $\sp$ for a game satisfying the premises of the lemma. Let $i$ be a vertex such that $\u_i(\sp)=0$ and assume, without loss of generality, that $i$ is orange. This implies that $i$ is surrounded by blue vertices only. 
	
	Pick another orange vertex $j\neq i$ which is adjacent to at least a blue one $\ell$. If $\ell\notin\setminus N_i(G)$, it follows that $i$ and $\ell$ can perform a profitable swap contradicting the assumption that $\sp$ is an equilibrium. Thus, $\ell$ has to belong to $N_i(G)$. Let us now consider two cases.
	
	If $i$ occupies a corner vertex, $\ell$ needs to be placed on a border one. So, as $\ell$ is adjacent to $i$ and $j$, it holds that $\u_{\ell}(\sp)\leq \frac13$. Thus, as we have $\u_{\ell}(\sp_{i\ell})=\frac12$ and $\u_i(\sp_{i\ell})>0$, $i$ and $\ell$ can perform a profitable swap contradicting the assumption that $\sp$ is an equilibrium.
	
	If $i$ is not located on a corner vertex, as $\ell$ is adjacent to $i$ and $j$, it holds that $\u_{\ell}(\sp)\leq \frac12$. Moreover, $|N_i(G)|\geq 3$ which yields $\u_{\ell}(\sp_{i\ell})=\frac{|N_i(G)|-1}{|N_i(G)|}\geq \frac23$. Thus, also in this case, $i$ and $\ell$ can perform a profitable swap contradicting the assumption that $\sp$ is an equilibrium.
\end{proof}

\noindent When no agent gets utility zero, the minimum possible utility is $\frac14$. Thus, Lemmas \ref{if1} and~\ref{no-zero} together imply an upper bound of $4$ on the PoA. However, a much better result can be shown. 

\begin{theorem}
	The PoA of $2$-SSGs played on $4$-grids is at most $2$.
	\label{thm:4grid_poa}
\end{theorem}

\begin{proof}
Without loss of generality, we consider an $l \times h$ grid, with $l \leq h$. When $l=1$, the PoA of $2$ follows from Theorem~\ref{thm:path_global}. Therefore, we assume that $l$, $h \geq 2$. Furthermore, by Lemma~\ref{if1}, we only need to consider the case in which there are at least two agents per type. By Lemma~\ref{no-zero}, we know that, in this case, the utility of each agent is strictly positive. We prove the claim by showing that the average utility of an agent is at least $\frac12$.
We divide the proof into two cases, depending on the utilities of the {\em middle} agents (i.e., agents occupying the middle vertices).

\noindent {\bf Case 1.} In the first case, we assume that the utility of every middle agent is at least $\frac12$. As corner agents (i.e., agents occupying corner vertices) have a utility of at least $\frac12$ each, we only need to prove the claim when there is at least one border agent (i.e., an agent occupying a border vertex) whose utility is equal to $\frac13$.  This implies that $l+h\geq 5$. Without loss of generality, we assume that there are more orange than blue agents having a utility equal to $\frac13$. Let $I$ be the border vertices occupied by the orange (border) agents having a utility of $\frac13$. As the overall number of border vertices is $2(l-2)+2(h-2)=2l+2h-8$, we have that the number of border agents having a utility greater than or equal to $\frac23$ is at least $2l+2h-8-2|I|$. Therefore, if $|I|=1$ and $l+h \geq 6$, then $2l+2h-8-2|I| \geq 12-8-2=2$; hence, the average utility of an agent is greater than or equal to $\frac12$. If $|I|=1$ and $l+h=5$, then the only configuration in which a swap equilibrium exists, unless of symmetries, is shown in Figure~\ref{4-grids_UB_2}(a). 
\begin{figure}[t]
	\center
		\includegraphics{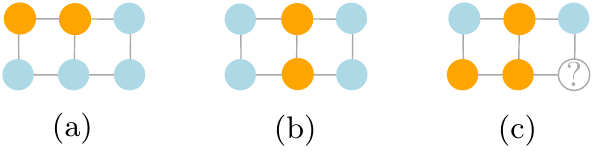}
		\caption{The unique swap equilibrium for $2\times 3$ 4-grids is shown in (a). Indeed, in (b) the blue agent in $v_{1,1}$ can swap with the orange agent in $v_{2,2}$, while in (c) the blue agent in $v_{1,1}$ can swap with the orange agent in $v_{1,2}$ (the question mark in $v_{2,3}$ means that the vertex can be occupied by an agent of any type).}
		\label{4-grids_UB_2}
\end{figure}

We observe that, in such configuration, the average utility of an agent is strictly greater than~$\frac12$.
It remains to prove the case in which $|I|\geq 2$. Since $\sp$ is a swap equilibrium, the utility of a blue agent that occupies a vertex that is not adjacent to all the vertices in $I$ is at least $\frac23$. As each blue agent occupies a vertex that is adjacent to at most 2 vertices in $I$ and because each vertex in $I$ is adjacent to exactly 2 vertices occupied by blue agents, the number of blue agents is at least $2|I|/2 = |I|$. Therefore, if we assume that every blue agent has a utility of at least $\frac23$, then the average utility of an agent would be at least $\frac12$. We observe that this assumption holds when either (a) $|I|\geq 3$ (because there is no blue agent occupying a vertex that is adjacent to all the vertices of $I$) or (b) $|I|=2$ and the two vertices of $I$ are either at $t$-hop distance from each other, with $t\geq 2$, or they are are at $2$-hop distance from each other and the utility of the border agent that occupies the vertex in between is at least $\frac23$. For the remaining case in which $|I|=2$, the two vertices of $I$ are at $2$-hop distance from each other, and the agent occupying the border vertex in between is equal to $\frac13$ -- and thus is of blue type -- we simply observe that the overall number of blue agents is at least $4$. Indeed, without loss of generality, let $v_{1,x-1}$ and~$v_{1,x+1}$ be the two vertices of $I$. As $v_{1,x}$ is occupied by a blue agent that has strictly positive utility, $v_{2,x}$ is also occupied by a blue agent. Furthermore, either $v_{1,x-2}$ or $v_{2,x-1}$ is occupied by a blue agent. Similarly, either $v_{1,x+2}$ or $v_{2,x+1}$ is occupied by a blue agent. Therefore, there are at least $4$ blue agents. Since 3 out of these $4$ blue agents have a utility of at least $\frac23$, again, the average utility of an agent is at least $\frac12$.

{\bf Case 2.} In the second case, we assume that there is at least an agent occupying a middle vertex and whose utility is equal to $\frac14$. Without loss of generality, we assume that there are more orange than blue agents having a utility equal to $\frac14$. Let $I$ be the vertices of the orange agents having a utility of $\frac14$.
We prove that 

(i) every blue agent has a utility of at least $\frac12$; 

(ii) the number of blue agent having utility greater than or equal to $\frac34$ is at least $|I|$; 

(iii) all border and corner agents are of blue type. 

This would clearly imply that the average utility of an agent is $\frac12$ since the utility of border and corner agents would be at least $\frac23$.

Let $v_{x,y}$ be a vertex of $I$ and, without loss of generality, we assume that $v_{x,y-1},v_{x-1,y}$, and~$v_{x,y+1}$ are occupied by blue agents whose utilities are greater than or equal to $\frac12$. Similarly, we can prove that the utility of every other blue agent that occupies a vertex that is not adjacent to all vertices in $I$ is at least $\frac34$. This implies that at least one vertex between $v_{x-1,y-1}$ and~$v_{x-1,y+1}$ is occupied by a blue agent whose utility is greater than or equal to $\frac34$; similarly, at least one vertex between $v_{x+1,y-1}$ and $v_{x+1,y+1}$ is occupied by a blue agent whose utility is greater than or equal to $\frac34$. Therefore, we have proved (ii) for the case in which $|I| \leq 2$. To prove (ii) when $|I| > 2$, it is enough to observe that all blue agents have a utility greater than or equal to $\frac34$ because none of them occupies a vertex that is adjacent to all the vertices in $I$. But this implies that each blue agent of utility of at least $\frac34$ occupies a vertex that is adjacent to at most one vertex in $I$. Hence, the overall number of blue agents is at least $|I|$. 

We now conclude the proof by proving (iii). First of all, we prove that at least one border or corner vertex is occupied by a blue agent. For the sake of contradiction, we assume that all border and corner vertices are occupied by the orange agents. Let $v_{x,y}$ be the leftmost-topmost vertex occupied by a blue agent, i.e., both $v_{x,y-1}$ and $v_{x-1,y}$ are occupied by orange agents and there is no other vertex $v_{x',y'}$ occupied by a blue agent such that $x' < x$ or $x=x'$ and~$y'<y$. We observe that such a vertex always exists because $x$, $y > 1$ and that $v_{x-1,y-1}$ must be occupied by an orange agent. Furthermore, by the choice of $v_{x,y}$, the utility of the two orange agents that occupy the   vertices $v_{x-1,y}$ and $v_{x,y-1}$ must be at least $\frac12$. Since the utility of the blue agent occupying the vertex $v_{x,y}$ has to be at least $\frac12$, $v_{x+1,y}$ and $v_{x,y+1}$ are occupied by blue agents. As a consequence, $N_{v_{x,y}} \cap I = \emptyset$. Therefore, swapping the agent that occupies $v_{x,y}$ with any agent occupying a vertex in $I$ would be an improving move. Now that we know that at least one border or corner agent is of blue type, we prove that all of them must be of blue type. For the sake of contradiction assume that at least one border or corner vertex is occupied by an orange agent. Without loss of generality, let $v_{1,y}$ be a vertex occupied by an orange agent such that $v_{1,y+1}$ is occupied by a blue agent. Since the utility of such a blue agent is at least $\frac12$, the unique middle vertex adjacent to $v_{1,y+1}$, i.e., $v_{2,y+1}$, must be occupied by a blue agent. This implies that $v_{1,y+1}$ cannot be adjacent to any vertex in $I$. As the utility of the agent occupying vertex $v_{1,y+1}$ is at most $\frac23$, swapping the agent occupying the vertex $v_{1,j+1}$ and any agent occupying a vertex in $I$ would be an improving move. This completes the proof.
\end{proof}

\noindent The following lemma gives a sufficient condition for a strategy profile to be an equilibrium.

\begin{lemma}\label{characterization}
Fix a $2$-SSG played on a $4$-grid. Any strategy profile in which corner and middle vertices get utility at least $\frac12$ and border ones get utility at least $\frac23$ is an equilibrium.
\end{lemma}
\begin{proof}
Fix a strategy profile $\sp$ meeting the premises of the claim and two vertices $i$ and $j$ of different color in $\sp$. As $\u_i(\sp)\geq \frac12$ and $\u_j(\sp)\geq \frac12$, it can only be $\u_i(\overline{\sp}_{ij})\leq \frac12$ thus implying that no profitable swaps are possible in $\sp$. 
\end{proof}

\noindent We now show a matching lower bound.

\begin{theorem}\label{lbgrid}
	The PoA of $2$-SSGs played on $4$-grids is at least $2$, even when both types have the same cardinality.
\end{theorem}

\begin{proof}
Fix a $2$-SSG played on an $n\times n$ grid $G$, with $n$ being an even number. We define a strategy profile $\sp$ by giving a coloring rule for any frame of $G$. Clearly, being $n$ an even number, there are $\frac{n}{2}$ frames in $G$ that we number from $1$ to $\frac{n}{2}$, with frame $1$ corresponding to the outer one, i.e., the biggest. Frame $i$, whose size is $n_i:=n-2(i-1)$, is colored as follows: all vertices in the left column and all vertices in the right column except for the first and the last are of the {\em basic color} of $i$, all other vertices (that are the ones on the upper and lower rows except for the vertices falling along the left column) take the other color. Observe that $n_i+n_i-2=2(n_i-1)$ vertices take the basic color of $i$ and $2(n_i-1)$ vertices take the other one, so that every frame evenly splits its vertices between the two colors. Thus, $\sp$ is a well-defined strategy profile for a $2$-SSG with both types having the same cardinality. The basic color of frame $i$ is orange if $i$ is odd and blue otherwise, see Figure \ref{fig1} for a pictorial example. To show that $\sp$ is an equilibrium, it suffices proving that it satisfies the premises of Lemma~\ref{characterization}. 

\begin{figure}[t]
\center
\includegraphics[scale=0.6]{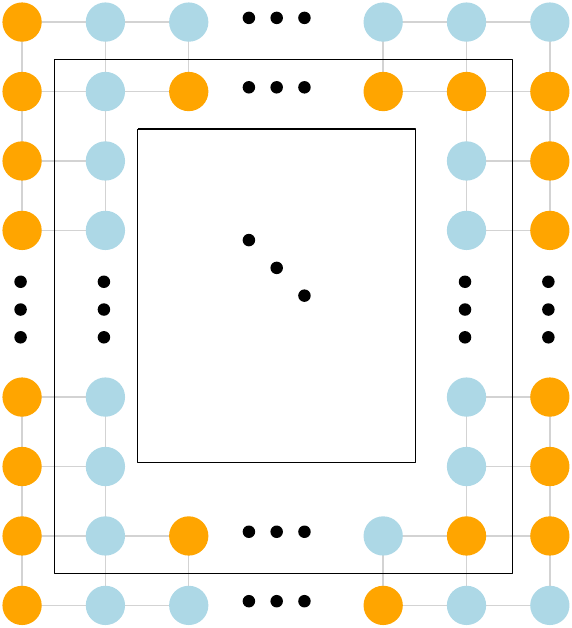}
\caption{Visualization of the first three frames of $G$ with the coloring induced by the strategy profile defined in the proof of Theorem \ref{lbgrid}.}\label{fig1}
\end{figure}

To address corner and border vertices, consider frame $1$, see again Figure \ref{fig1}. It comes by construction that every corner vertices gets utility $\frac12$ and that every border vertices gets utility at least $\frac23$, except for vertices $(1,2)$, $(2,n)$, $(n-1,n)$ and $(n,2)$ for which further investigation is needed. In particular, they get utility $\frac23$ if and only if the following coloring holds: $(2,2)$ is blue, $(2,n-1)$ is orange, $(n-1,n-1)$ is orange and $(n-1,2)$ is blue. This holds by construction and can be verified by a direct inspection of Figure \ref{fig1}.

To address middle vertices, it suffices proving that, any vertex belonging to frame $i>1$ has two orange and two blue neighbors Let $c$ denote the basic color of frame $i$ and $\overline{c}$ be the other color. Consider a generic vertex $v$ belonging to frame $i$. By inspecting all possible positions of~$v$ within the frame as shown in Figure \ref{fig2}, it can be easily verified that the desired property~holds.

By Lemma \ref{characterization}, $\sp$ is an equilibrium.
\begin{figure}[h]
\center
\includegraphics[scale=0.6]{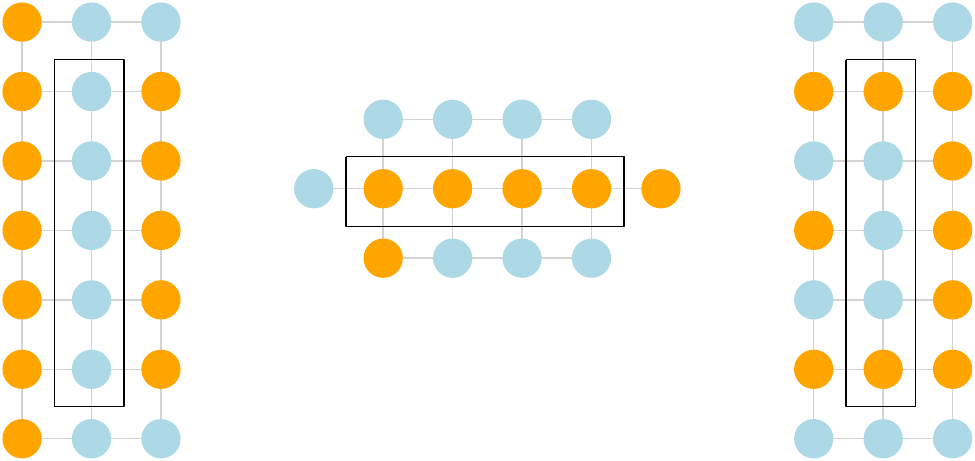}
\caption{Visualization of the neighborhood of vertices belonging to a frame $i>1$. The target vertices are the ones included in the box. On the left, vertices belonging to the left column; on the right, vertices belonging to the right column; on the center, vertices belonging to a row but not to a column.}\label{fig2}
\end{figure}
\end{proof}

\noindent We now show matching upper and lower bounds on the LPoA for local $2$-SSGs played on grids. By inspecting all the possibilities, the LPoA of local $2$-SSGs played on $2\times 2$ grids is $1$. Indeed, assuming $b\geq o$, for $o=1$, all the configurations are isomorphic to each other, while, for $o=2$, the unique (local) swap equilibrium -- up to isomorphisms -- is $\bigl[ \begin{smallmatrix}  o & b \\ o & b \end{smallmatrix} \bigr]$. 

\begin{theorem}\label{thm:4grid_local}
	The LPoA of local $2$-SSGs played on $2 \times h$ $4$-grids, with $h\geq 3$ is $3$. Furthermore, for every $\epsilon >0$, there is a value $h_0$ such that, for every $h \geq h_0$, the PoA of $2\times h$ $4$-grid is at least~$3-\epsilon$.
\end{theorem}

\begin{proof}
For the lower bound consider the strategy profile in which $h$ is a multiple of $6$, $o=b$, odd columns are filled with orange agents, and even column are filled with blue agents (see Figure~\ref{2_times_h_4-grids_LB_3}(a) for an example on a $2\times 6$ $4$-grid). The strategy profile is a local swap equilibrium and the corresponding social welfare is equal to $\frac{1}{3}(n-4)+2=\frac{n+2}{3}$. A social optimum having social welfare of $n-\frac{4}{3}=\frac{3n-4}{3}$ is the strategy profile in which all the orange agents occupy the first $\frac{h}{2}$ columns, and the blue agents occupy the last $\frac{h}{2}$ columns (see  Figure~\ref{2_times_h_4-grids_LB_3}(b) for an example on a $2\times 6$ $4$-grid). Therefore, for every $h \geq \frac{5-\epsilon}{\epsilon}$, we have that the following formula is a lower bound to the LPoA
\[
\frac{3n-4}{n+2}=3-\frac{10}{n+2} =3-\frac{5}{h+1}\geq 3 - \epsilon.
\]
\begin{figure}[t]
	\center
		\includegraphics{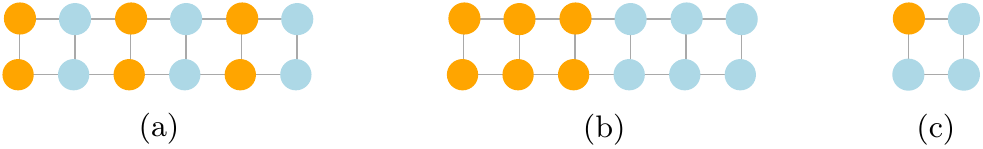}
		\caption{The local swap equilibrium with largest social welfare is shown in (a) and the social optimum is shown in (b). (c) shows the unique local swap equilibrium which contains an agent with utility $0$.}
		\label{2_times_h_4-grids_LB_3}
\end{figure}

\noindent To prove the upper bound of $3$, we show that the average utility of an agent is at least $\frac13$. We consider only the agents that have a utility of $0$ since all the others have a utility of at least $\frac13$ each. When $h$ is equal to $2$, the unique strategy profile (unless of symmetries) that is in local swap equilibrium and contains at least one agent that has $0$ utility is depicted in Figure~\ref{2_times_h_4-grids_LB_3}(c). However, it is easy to check that the average utility of an agent is $\frac12$. Therefore, we only need to prove the claim for $h\geq 3$. We prove that if $x$ is the number of agents whose utilities are equal to $0$, then there are at least $x$ agents that have a utility of at least $\frac23$ each. Indeed, let $i$ be any agent that has a utility equal to $0$. Since $\sp$ is a local swap equilibrium and $h \geq 3$, there is an agent $j$ such that 

(i) $\sigma_j \in N_{\sigma_i}$, 

(ii) the type of $i$ is different from the type of $j$, and 

(iii) $\u_j(\sp)\geq \frac23$. 

\noindent More precisely, (iii) implies that $N_v\setminus\{u\}$ contains only vertices occupied by agents of the same type of $j$. Therefore, we can uniquely assign an agent $j$ that has a utility of at least $\frac23$ to every agent $i$ that has a utility of $0$. The claim follows.
\end{proof}

\begin{theorem}\label{thm:poa:grid}
	The LPoA of local $2$-SSG played on $3 \times h$ $4$-grids, with $h\geq 3$ is $\frac{36}{13}$. Furthermore, for every $\epsilon >0$, there is a value $h_0$ such that, for every $h \geq h_0$, the PoA of $2\times h$ $4$-grid is at least~$\frac{36}{13}-\epsilon$.
\end{theorem}

\begin{proof}
For the lower bound of $\frac{36}{13}-\epsilon$ consider the strategy profile in Figure~\ref{3_times_h_4-grids_LB}. The average utility of the agents that occupy any column from $3$ to $h-2$ is equal to $\frac{13}{36}$.

Now, we prove the upper bound of $\frac{36}{13}$. In the rest of the proof, by utility of the $r$-th column we mean the overall utility of the agents that occupy the vertices of the $r$-th column. We show that the utility of the first (respectively, last) column is of at least $\frac56$  and we show that the average utility of the other columns is at least $\frac{12}{36}$.

\begin{figure}[h]
	\center
	\includegraphics{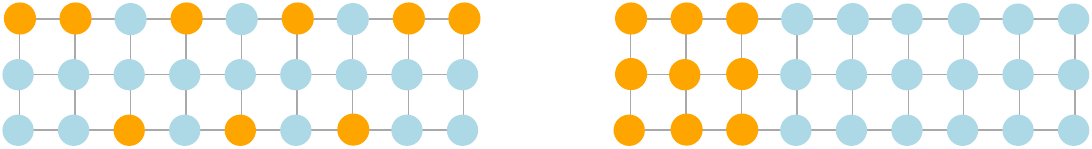}
	\caption{The strategy profile inducing an average agent's utility that can be made arbitrarily close to $\frac{13}{36}$ is shown on the left side via a small example ($3 \times 9$ $4$-grid). On the right side it is shown a strategy profile inducing an average agent's utility arbitrarily close to $1$.}
	\label{3_times_h_4-grids_LB}
\end{figure}

First of all, we observe that, among the agents that occupy the vertices of the $r$-th column, at most one can have a utility of $0$. Indeed, if without loss of generality, $v_{1,r}$ and $v_{2,r}$ are occupied by two agents of utility $0$, then by swapping the two agents, they would both have a strictly positive utility. Moreover, if the two agents having a utility of $0$ occupy the vertices $v_{1,r}$ and~$v_{3,r}$, then by swapping either of the two agents with the agent occupying the vertex $v_{2,r}$ would be an improving move.
This observation implies that the utility of the $r$-th column, with $r \in \{1,h\}$, is lower bounded by $\frac56$.

Now we show that utility of the $r$-th column, with $2 \leq r \leq h-1$, is of at least $\frac{13}{12}$. We divide the proof into cases.

In the first case, we assume that the middle agent has a utility of $0$. In this case both border agents of the column would have a utility of at least $\frac23$ and therefore, the utility of the $r$-th column would be of at least $\frac43 \geq \frac{13}{12}$. 

In the second case, we assume that a border agent has a utility of $0$. This implies that the middle agent has a utility of $\frac34$ and the other border agent a utility of at least $\frac13$. Therefore, the utility of the $r$-th column is at least $\frac{13}{12}$.

In the last case, we assume that all agents that occupy the vertices of the $r$-th column have a strictly positive utility. We observe that the only interesting case to look at, is when the border agents both have a utility of $\frac13$ and the middle agent has a utility of $\frac14$, as in all the other cases, the utility of the $r$-th column would be of at least $\frac{13}{12}$. In this case, the utility of at least one between column $r-1$ and column $r+1$ must be of at least $\frac32$. Indeed, if without loss of generality, $v_{1,r}$ and $v_{2,r}$ are occupied by orange agents, while $v_{3,r}$ is occupied by a blue agent, then, due to the agents' utilities, $v_{1,r-1}$, $v_{2,r-2}$, $v_{1,r+1}$, $v_{2,r+2}$ are occupied by blue agents, while one between $v_{3,r-1}$ and $v_{3,r+1}$ must be occupied by a blue agent as well. In either case, one column between column $r-1$ and column $r+1$ is entirely occupied by blue agents.
As a consequence, for every two columns each of utility equal to $\frac23+\frac14=\frac{11}{12}$, there must be a column of utility of at least~$\frac32$. By averaging among the three considered columns, we obtain $\frac{1}{3}(\frac{11}{6}+\frac{3}{2})=\frac{10}{9}>\frac{13}{12}$. This completes the proof. 
\end{proof}

\begin{theorem}\label{thm:poa_8grid_}
	For every $\epsilon > 0$, the LPoA of local $2$-SSG played on $l \times h$ $4$-grids, with $\ell,h \geq 8+\frac{20}{\epsilon}$ is in the interval $\left(\frac{5}{2}-\epsilon,\frac{5}{2}+\epsilon\right]$.
\end{theorem}

\begin{proof}
Let $X$ be the set of middle vertices that are adjacent neither to border nor to corner vertices. Clearly, $N_X=\bigcup_{v \in X} N_v$ is the set of all the middle vertices. Therefore, the degree of each vertex $v \in N_X$ is equal to $4$. Let $Z \subseteq N_X$ be the set of vertices occupied by agents that have a utility strictly greater than $\frac25$. From Lemma \ref{lm:average_utility_Delta_window}, we have that the average utility of the agents in $X \cup Z$ is at least $\frac25$. 
 
As a consequence, the social welfare is lower bounded by
$\frac{2}{5}|X\cup Z| \geq \frac{2}{5}(l-4)(h-4)>\frac{2}{5}l h-\frac{8}{5}(l+h)$. Therefore, the LPoA can be upper bounded by
\[
\frac{l h}{\frac{2}{5}l h-\frac{8}{5}(l +h)} =\frac{1}{\frac{2}{5}-\frac{8}{5}\frac{l+h}{l h}}\leq \frac{1}{\frac{2}{5}-\frac{8}{5}\frac{2(8+20/\epsilon)}{(8+20/\epsilon)^2}}=\frac{5}{2}+\epsilon.
\]

\noindent For the lower bound, consider the $l \times h$ grid, with $l = 5 l'+1$ and $h=5h'$, that is filled as shown in Figure~\ref{l_times_h_4-grids_LB_5_divided_2}. The social welfare for arbitrarily large values of $l'$ and $h'$ (i.e., $l$ and $h$) can be made arbitrarily close to the average utility of the agents that occupy the vertices of the tiles labeled with $T$. Observe that $\frac25$ is the average utility of the agents that occupy all the vertices of any tile labeled with $T$. As the ratio between blue and orange agents can be made arbitrarily close to $\frac32$, the maximum average utility of an agent is arbitrarily close to $1$ by placing the orange agents over the vertices of the first $\frac{2}{5} h$ columns and the blue agents in the remaining $\frac{3}{5} h$ columns. Therefore, the LPoA is lower bounded by $\frac{5}{2}-\epsilon$.
\begin{figure}[t]
	\center
		\includegraphics{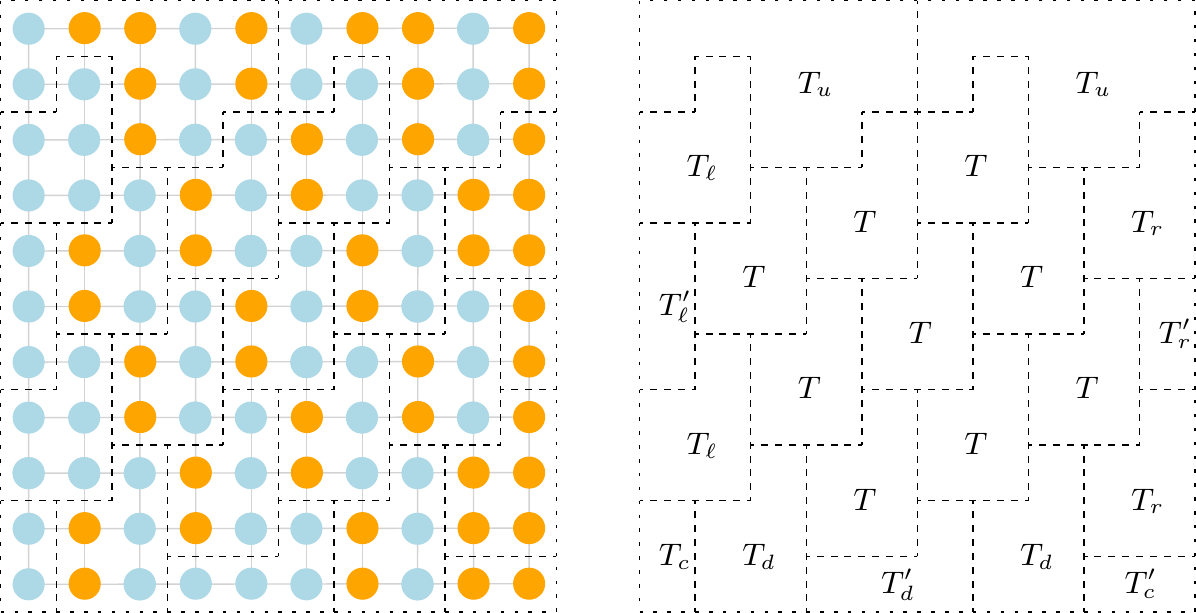}
		\caption{The strategy profile inducing an average agent's utility arbitrarily close to $\frac25$ is shown on the left side via a small example over an $11\times 10$ $4$-grid. On the right side, the tiling showing the pattern we have used for building the instance. The tiles $T_c$ and $T_c'$ are only used in order to fill the bottom-left and bottom-right corners of the $4$-grid. Observe that using exactly the same tiles, one can build arbitrarily large instances. Moreover, for arbitrarily large instances, the average utility of an agent is basically determined by the average utility of the agents that occupy the vertices of any tile $T$, i.e., $\frac25$.}
		\label{l_times_h_4-grids_LB_5_divided_2}
\end{figure}
\end{proof}

\noindent We now turn our focus to the $8$-grid and first consider the case where one type has only one agent.

\begin{theorem}
	The PoA of $2$-SSGs played on an $8$-grid in which one type has cardinality $1$ is equal to $\frac{897}{704}$. 
	\label{thm:8grid_poa} 
\end{theorem}

\begin{proof}
	Assume without loss of generality that orange is the type with a unique representative. For this game, any strategy profile $\sp$ is an equilibrium, since in any profile $\sp$ the orange vertex~$o$ gets utility zero, the vertices not adjacent get utility $1$, while all vertices adjacent to $o$ get strictly less than $1$. Call these last vertices the \textit{penalized vertices}. Thus, the PoA is maximized by comparing the social welfare of the strategy profile minimizing the overall loss of the penalized vertices with the once of the strategy profile maximizing it. The overall loss of the penalized vertices is minimized when $o$ is a corner vertex, while it is maximized when $o$ is a middle one on an $8$-grid with $l = h = 3$. Comparing the two social welfares gives the claimed bound.	
\end{proof}

\noindent Similar to the $4$-grid, if there are at least two agents of each type no agent gets zero utility in an equilibrium.

\begin{lemma}
	In any equilibrium for a $2$-SSG played on an $8$-grid in which both types have cardinality larger than $1$ all agents get positive utility.
	\label{lemma:grid8_pos}
\end{lemma}

\begin{proof}
	Fix an equilibrium $\sp$ for a game satisfying for a $2$-SSG played on an $8$-grid. Let $i$ be a vertex such that  $\u_i(\sp)=0$ and assume without loss of generality that $i$ is orange. This implies that $i$ is surrounded by blue vertices only.
	
	Pick another orange vertex $j \neq i$ which is adjacent to at least one blue agent $r$. If $r \notin N_i(G)$, it follows that $i$ and $r$ can perform a profitable swap contradicting the assumption that $\sp$ is an equilibrium. Thus, $r$ has to belong to $N_i(G)$. We consider two different cases.
	
	If $i$ is placed on a corner vertex, $r$ needs to be either located at a border or a middle one. First assume $r$ occupies a border vertex, and since $r$ is adjacent to $i$ and $j$, it holds that $\u_r(\sp) \leq \frac35$. Thus, as we have $\u_{r}(\sp_{ir})= \frac23$ and $\u_i(\sp_{ir})>0$, $i$ and $r$ can perform a profitable swap contradicting the assumption that $\sp$ is an equilibrium.
	If $r$ occupies a middle vertex, it holds that $\u_r(\sp) \leq \frac68$. So, if $\u_r(\sp) \leq \frac58$, $i$ and $r$ can perform a profitable swap since $\u_{r}(\sp_{ir})= \frac23$ and $\u_i(\sp_{ir})>0$. If $\u_r(\sp) = \frac68$ than $l$ has a blue neighbor $r^\prime$ who is not adjacent to $i$ but to $j$, hence $\u_{r^\prime}(\sp) \leq \frac78$ and therefor $i$ and $r^\prime$ can perform a profitable swap since $\u_{r^\prime}(\sp_{ir^\prime})= 1$ and $\u_i(\sp_{ir^\prime})>0$.
	
	If $i$ is not placed on a corner vertex $\u_{r}(\sp_{ir}) = \frac45$. Since $\u_r(\sp) \leq \frac68$ and $\u_i(\sp_{ir})>0$, swapping $i$ {\tiny {\tiny }}and $l$ is profitable, contradicting the assumption that $\sp$ is an equilibrium.
	
\end{proof}

\noindent When no agent gets utility zero, the minimum possible utility is $\frac{1}{8}$. Thus, Lemma~\ref{lemma:grid8_pos} implies an upper bound on the PoA.

\begin{theorem}
	The PoA of $2$-SSGs played on an $8$-grid is at most $8$.
	\label{thm:poa_8grid}
\end{theorem}

\begin{proof}
	The statement follows directly by Lemma~\ref{lemma:grid8_pos}. Every agent gets at least a utility equals $\frac{1}{8}$ and at most a utility of $1$.
\end{proof}

\noindent We conclude by proving a much better bound for the (L)PoA, if the $8$-grid is large enough.

\begin{theorem}
For every $\epsilon > 0$, the LPoA of local $2$-SSGs played on an $l \times h$ $8$-grid, with $l,h \geq 8+\frac{18}{\epsilon}$ is at most $\frac{9}{4}+\epsilon$. 
	\label{poa_8grid}
\end{theorem}

\begin{proof}
	Let $X$ be the set of middle vertices that are adjacent neither to border nor to corner vertices. Clearly, $N_X=\bigcup_{v \in X} N_v$ is the set of all the middle vertices. Therefore, the degree of each vertex $v \in N_X$ is equal to $8$. Let $Z \subseteq N_X$ be the set of vertices occupied by agents that have a utility strictly greater than $\frac{4}{9}$. From Lemma \ref{lm:average_utility_Delta_window}, we have that the average utility of the agents in $X \cup Z$ is at least $\frac{4}{9}$. 
	As a consequence, the social welfare is lower bounded by
	$\frac{4}{9}|X\cup Z| \geq \frac{4}{9}(l-4)(h-4)>\frac{4}{9}l h-\frac{16}{9}(l + h)$. Therefore, the LPoA is at most
	$$
	\frac{l h}{\frac{4}{9}l h-\frac{16}{9}(l +h)} =\frac{1}{\frac{4}{9}-\frac{16}{9}\frac{l+h}{l h}}\leq \frac{1}{\frac{4}{9}-\frac{16}{9}\frac{2(8+18/\epsilon)}{(8+18/\epsilon)^2}}=\frac{9}{4}+\epsilon.
	$$
\end{proof}

\section{Conclusion and Open Problems}
We have shed light on the influence of the underlying graph topology on the existence of equilibria, the game dynamics and the Price of Anarchy in Swap Schelling Games on graphs. Moreover, we have studied the impact of restricting agents to local swaps. We present tight or almost tight bounds for a variety of graph classes. 

Clearly, improving on the non-tight bounds is an interesting challenge for future work. 
Regarding the local Swap Schelling Game, we leave some interesting problems open. Among them is the question whether local swap equilibria are guaranteed to exist for all graph classes and if the local $k$-SSG always has the finite improvement property. So far, we are not aware of any counter-examples for both questions and extensive agent-based simulations indicate that both equilibrium existence and guaranteed convergence of improving response dynamics may hold. 
Another interesting line of study is to analyze the Jump Schelling Game with respect to varying underlying graphs and locality.

\bibliographystyle{abbrv}
\bibliography{schelling_games}

\begin{thebibliography}{10}

\bibitem{A+19}
A.~Agarwal, E.~Elkind, J.~Gan, and A.~A. Voudouris.
\newblock Swap stability in schelling games on graphs.
\newblock {\em CoRR}, abs/1909.02421, 2019, to appear at AAAI'20.

\bibitem{aits2019}
D.~Aits, A.~Carver, and P.~Turrini.
\newblock Group segregation in social networks.
\newblock In {\em Proceedings of the 18th International Conference on
  Autonomous Agents and Multiagent Systems, {AAMAS} 2019, Montreal QC, Canada,
  May 13-17, 2019}, pages 1524--1532, 2019.

\bibitem{aziz2019}
H.~Aziz, F.~Brandl, F.~Brandt, P.~Harrenstein, M.~Olsen, and D.~Peters.
\newblock Fractional hedonic games.
\newblock {\em ACM Transactions on Economics and Computation}, 7(2):6:1--6:29,
  2019.

\bibitem{BEL14}
G.~Barmpalias, R.~Elwes, and A.~Lewis-Pye.
\newblock Digital morphogenesis via schelling segregation.
\newblock In {\em 2014 IEEE 55th Annual Symposium on Foundations of Computer
  Science (FOCS)}, pages 156--165. IEEE, 2014.

\bibitem{BEL16}
G.~Barmpalias, R.~Elwes, and A.~Lewis-Pye.
\newblock Unperturbed schelling segregation in two or three dimensions.
\newblock {\em Journal of Statistical Physics}, 164(6):1460--1487, 2016.

\bibitem{Bhakta14}
P.~Bhakta, S.~Miracle, and D.~Randall.
\newblock Clustering and mixing times for segregation models on
  $\mathcal{Z}^2$.
\newblock In {\em Symposium on Discrete Algorithms (SODA)}, pages 327--340,
  2014.

\bibitem{bilo2018}
V.~Bil\`o, A.~Fanelli, M.~Flammini, G.~Monaco, and L.~Moscardelli.
\newblock Nash stable outcomes in fractional hedonic games: Existence,
  efficiency and computation.
\newblock {\em Journal of Artificial Intelligence Research}, 62:315--371, 2018.

\bibitem{BJ02}
A.~Bogomolnaia and M.~O. Jackson.
\newblock The stability of hedonic coalition structures.
\newblock {\em Games and Economic Behavior}, 38(2):201--230, 2002.

\bibitem{BIK12}
C.~Brandt, N.~Immorlica, G.~Kamath, and R.~Kleinberg.
\newblock An analysis of one-dimensional schelling segregation.
\newblock In {\em Proceedings of the forty-fourth annual ACM symposium on
  Theory of computing (STOC)}, pages 789--804. ACM, 2012.

\bibitem{bredereck2019}
R.~Bredereck, E.~Elkind, and A.~Igarashi.
\newblock Hedonic diversity games.
\newblock In {\em Proceedings of the 18th International Conference on
  Autonomous Agents and Multiagent Systems, {AAMAS} 2019, Montreal QC, Canada,
  May 13-17, 2019}, pages 565--573, 2019.

\bibitem{carosi2019}
R.~Carosi, G.~Monaco, and L.~Moscardelli.
\newblock Local core stability in simple symmetric fractional hedonic games.
\newblock In {\em Proceedings of the 18th International Conference on
  Autonomous Agents and Multiagent Systems, {AAMAS} 2019, Montreal QC, Canada,
  May 13-17, 2019}, pages 574--582, 2019.

\bibitem{carver2018}
A.~Carver and P.~Turrini.
\newblock Intolerance does not necessarily lead to segregation: A
  computer-aided analysis of the schelling segregation model.
\newblock In {\em Proceedings of the 17th International Conference on
  Autonomous Agents and Multiagent Systems, {AAMAS} 2018, Stockholm, Sweden,
  July 10-15, 2018}, pages 1889--1890, 2018.

\bibitem{CLM18}
A.~Chauhan, P.~Lenzner, and L.~Molitor.
\newblock Schelling segregation with strategic agents.
\newblock In {\em International Symposium on Algorithmic Game Theory (SAGT)},
  pages 137--149. Springer, 2018.

\bibitem{DG80}
J.~H. Drèze and J.~Greenberg.
\newblock Hedonic coalitions: Optimality and stability.
\newblock {\em Econometrica: Journal of the Econometric Society}, pages
  987--1003, 1980.

\bibitem{E+19}
H.~Echzell, T.~Friedrich, P.~Lenzner, L.~Molitor, M.~Pappik, F.~Sch{\"{o}}ne,
  F.~Sommer, and D.~Stangl.
\newblock Convergence and hardness of strategic schelling segregation.
\newblock {\em CoRR}, abs/1907.07513, 2019, to appear at WINE'19.

\bibitem{elkind19}
E.~Elkind, J.~Gan, A.~Igarashi, W.~Suksompong, and A.~A. Voudouris.
\newblock Schelling games on graphs.
\newblock In {\em Proceedings of the Twenty-Eighth International Joint
  Conference on Artificial Intelligence, {IJCAI} 2019, Macao, China, August
  10-16, 2019}, pages 266--272, 2019.

\bibitem{fichtenberger2019}
H.~Fichtenberger, A.~Krivosija, and A.~Rey.
\newblock Testing individual-based stability properties in graphical hedonic
  games.
\newblock In {\em Proceedings of the 18th International Conference on
  Autonomous Agents and Multiagent Systems, {AAMAS} 2019, Montreal QC, Canada,
  May 13-17, 2019}, pages 882--890, 2019.

\bibitem{fossett1998simseg}
M.~A. Fossett.
\newblock Simseg--a computer program to simulate the dynamics of residential
  segregation by social and ethnic status.
\newblock {\em Race and Ethnic Studies Institute Technical Report and Program,
  Texas A\&M University}, 1998.

\bibitem{Gerhold08}
S.~Gerhold, L.~Glebsky, C.~Schneider, H.~Weiss, and B.~Zimmermann.
\newblock Computing the complexity for schelling segregation models.
\newblock {\em Communications in Nonlinear Science and Numerical Simulation},
  (13):2236 -- 2245, 2008.

\bibitem{igarashi2019}
A.~Igarashi, K.~Ota, Y.~Sakurai, and M.~Yokoo.
\newblock Robustness against agent failure in hedonic games.
\newblock {\em Journal of Statistical Physics}, 170(4):748--783, 2018.

\bibitem{BIK17}
N.~Immorlica, R.~Kleinberg, B.~Lucier, and M.~Zadomighaddam.
\newblock Exponential segregation in a two-dimensional schelling model with
  tolerant individuals.
\newblock In {\em Proceedings of the Twenty-Eighth Annual ACM-SIAM Symposium on
  Discrete Algorithms (SODA)}, pages 984--993. SIAM, 2017.

\bibitem{kerkmann2019}
A.~M. Kerkmann and J.~Rothe.
\newblock Stability in fen-hedonic games for single-player deviations.
\newblock In {\em Proceedings of the 18th International Conference on
  Autonomous Agents and Multiagent Systems, {AAMAS} 2019, Montreal QC, Canada,
  May 13-17, 2019}, pages 891--899, 2019.

\bibitem{KP99}
E.~Koutsoupias and C.~Papadimitriou.
\newblock Worst-case equilibria.
\newblock In {\em Annual Symposium on Theoretical Aspects of Computer Science},
  pages 404--413. Springer, 1999.

\bibitem{monaco2018}
G.~Monaco, L.~Moscardelli, and Y.~Velaj.
\newblock Stable outcome in modified fractional hedonic games.
\newblock In {\em Proceedings of the 17th International Conference on
  Autonomous Agents and Multiagent Systems, {AAMAS} 2018, Stockholm, Sweden,
  July 10-15, 2018}, pages 937--945, 2018.

\bibitem{monaco2019}
G.~Monaco, L.~Moscardelli, and Y.~Velaj.
\newblock On the performance of stable outcomes in modified fractional hedonic
  games with egalitarian social welfare.
\newblock In {\em Proceedings of the 18th International Conference on
  Autonomous Agents and Multiagent Systems, {AAMAS} 2019, Montreal QC, Canada,
  May 13-17, 2019}, pages 873--881, 2019.

\bibitem{MS96}
D.~Monderer and L.~S. Shapley.
\newblock Potential games.
\newblock {\em Games and economic behavior}, 14(1):124--143, 1996.

\bibitem{omidvar2018self}
H.~Omidvar and M.~Franceschetti.
\newblock Self-organized segregation on the grid.
\newblock {\em Journal of Statistical Physics}, 170(4):748--783, 2018.

\bibitem{Sch69}
T.~C. Schelling.
\newblock Models of segregation.
\newblock {\em The American Economic Review}, 59(2):488--493, 1969.

\bibitem{Schelling71}
T.~C. Schelling.
\newblock Dynamic models of segregation.
\newblock {\em Journal of mathematical sociology}, 1(2):143--186, 1971.

\bibitem{Vin06}
D.~Vinkovi{\'c} and A.~Kirman.
\newblock A physical analogue of the schelling model.
\newblock {\em Proceedings of the National Academy of Sciences},
  103(51):19261--19265, 2006.

\bibitem{You98}
H.~P. Young.
\newblock {\em Individual strategy and social structure : an evolutionary
  theory of institutions}.
\newblock Princeton University Press Princeton, N.J, 1998.

\bibitem{Zha04}
J.~Zhang.
\newblock A dynamic model of residential segregation.
\newblock {\em The Journal of Mathematical Sociology}, 28(3):147--170, 2004.

\bibitem{Zha04b}
J.~Zhang.
\newblock Residential segregation in an all-integrationist world.
\newblock {\em Journal of Economic Behavior and Organization}, 54(4):533--550,
  2004.

\end{thebibliography}

\end{document}